\documentclass[reqno]{amsart}
\usepackage[svgnames,dvipsnames]{xcolor}
\usepackage{geometry}
\usepackage{amssymb,amsmath,amsthm,enumerate}
\usepackage{hyperref}

\usepackage{xcolor}

\allowdisplaybreaks

\usepackage{dsfont, mathrsfs}

\usepackage{graphicx}
\newcommand\smallO{
	{{\scriptscriptstyle\mathcal{O}}}
}

\newcommand{\pa}{\partial}

\newcommand{\eps}{\varepsilon}

\newcommand{\bR}{\mathbb{R}^d}
\newcommand{\bRp}{[0,\infty)}
\newcommand{\bRfp}{\mathbb{R}^d \times [0,\infty)}

\newcommand{\bS}{\mathbb{S}^{d-1}}

\newcommand{\md}{\mathrm{d}}

\newcommand{\F}{\mathbb{F}}
\newcommand{\G}{\mathbb{G}}
\newcommand{\Q}{\mathbb{Q}}
\newcommand{\Wr}{\mathbb{W}_\rho}
\newcommand{\Fr}{\mathbb{F}_\rho}

\newcommand{\la}{\langle}
\newcommand{\ra}{\rangle}

\newcommand{\Etot}{E^{\la \ra}}
\newcommand{\Vs}{\hat{V} \cdot \sigma}
\newcommand{\bs}{\bar{s}_{ij}}

\newcommand{\buR}{\tilde{b}_{ij}^{ub}(R)}
\newcommand{\blR}{\tilde{b}_{ij}^{lb}(R)}

\newcommand{\buRr}{\tilde{b}_{ij}^{ub}(r,R)}
\newcommand{\blRr}{\tilde{b}_{ij}^{lb}(r,R)}

\newcommand{\dpo}{d_{ij}(r,R)}
\newcommand{\dm}{d_i(R)}

\newcommand{\kub}{\max_{1 \leq i, j  \leq P}(\kappa_{ij}^{ub}) }

\newcommand{\ru}{\rho_{ij}^{ub}}

\newcommand{\nb}{ \left\| b_{ij} \right\|_{L^1_{\bS}} }

\newcommand{\ks}{\bar{k}_*}
\newcommand{\kss}{k_*}

\newcommand{\mP}{\mathrm{m}}
\newcommand{\mM}{\mathrm{m}}

\newcommand{\Cpov}{\mathcal{C}^{ij}_k}

\newcommand{\mfi}{\mathfrak{m}^i}
\newcommand{\mfj}{\mathfrak{m}^j}
\newcommand{\mfl}{\mathfrak{m}^\ell}

\newcommand{\g}{\gamma_{ij}}
\newcommand{\gi}{\bar{\bar{\gamma}}_{i}}
\newcommand{\gj}{\bar{\bar{\gamma}}_{j}}
\newcommand{\gl}{\bar{\bar{\gamma}}_{\ell}}
\newcommand{\gw}{\bar{\gamma}}
\newcommand{\gm}{\bar{\bar{\gamma}}}

\newcommand{\As}{A_{\star}}
\newcommand{\Cs}{C_{\star}}

\newcommand{\Ek}{E_{k}}
\newcommand{\Eks}{E_{\kss}}

\newcommand{\Lqi}{L^1_{q,i}}
\newcommand{\Lq}{L^1_{q}}

\usepackage[normalem]{ulem}

\allowdisplaybreaks

\newtheorem{theorem}{Theorem}[section]
\newtheorem{corollary}{Corollary}

\newtheorem{lemma}[theorem]{Lemma}
\newtheorem{proposition}[theorem]{Proposition}

\newtheorem{definition}[theorem]{Definition}
\newtheorem{remark}{Remark}

\title[]
{The Cauchy problem for Boltzmann bi-linear systems:\\
The mixing of monatomic and polyatomic gases}

\author[Ricardo J. Alonso, Irene M. Gamba and Milana Pavi\'c-\v Coli\'{c}]{}

\email{ricardo.alonso@qatar.tamu.edu}
\email{gamba@math.utexas.edu}
\email{milana.pavic@dmi.uns.ac.rs}

\begin{document}
	
\maketitle

\bibliographystyle{abbrv}

\medskip

\centerline{\scshape Ricardo J. Alonso}
{\footnotesize
	\centerline{Texas A\&M}
	\centerline{Division of Arts and Sciences}
	\centerline{Education City, Doha, Qatar }
}

\medskip
\centerline{\scshape Irene M. Gamba}
{\footnotesize
	\centerline{Department of Mathematics and Oden Institute of Computational Engineering and Sciences}
	\centerline{University of Texas at Austin}
	\centerline{201 E. 24th Street, POB 4.102}
	\centerline{1 University Station (C0200)
		Austin, Texas, 78712-1229,
		USA}
}

\medskip

\centerline{\scshape Milana Pavi\'c-\v Coli\'{c}}
{\footnotesize
	\centerline{Department of Mathematics and Informatics}
	\centerline{Faculty of Sciences, University of Novi Sad}
	\centerline{Trg Dositeja Obradovi\'ca 4, 21000 Novi Sad, Serbia}
}

\medskip
{\footnotesize
	\centerline{Applied and Computational Mathematics}
	\centerline{RWTH Aachen University }
	\centerline{ Schinkelstr. 2, 52062
		Aachen, Germany }
}

\begin{abstract}

From a unified vision of vector valued solutions in weighted  Banach spaces, this manuscript establishes the existence and uniqueness for space homogeneous Boltzmann bi-linear systems with conservative collisional forms arising in complex  gas dynamical structures.  This broader vision is directly applied to dilute multi-component gas mixtures  composed of both monatomic and polyatomic gases. Such models can be viewed as extensions of scalar Boltzmann binary elastic flows, as much as monatomic gas mixtures with disparate masses and single  polyatomic gases, providing a unify approach for vector valued solutions in weighted vector Banach spaces.

Novel aspects of this work include developing the extension of a general ODE theory in vector valued weighted Banach spaces, precise lower bounds for the collision frequency in terms of the weighted Banach norm, energy identities and the consequently angular or compact manifold averaging lemmas which provide coerciveness resulting into global in time stability, a new combinatorics estimate for $p$-binomial forms producing sharper estimates for the $k$-moments of bi-linear collisional forms. These techniques enable the Cauchy problem improvement that resolves the model with initial data corresponding to strictly positive and bounded initial vector valued mass and total energy, in addition to only $2^+$ moment determined by the hard potential rates discrepancy, a result comparable in generality to the classical Cauchy theory of the scalar homogeneous Boltzmann equation. 
\end{abstract}

\noindent
\textbf{Keywords.} System of Boltzmann equations; Compact Manifold  Averaging; Statistical moment estimates for bi-linear integral forms;  Multi-component gas mixtures.

\smallskip
\noindent
\textbf{AMS subject classifications.} 35Q20, 76P05, 82C40.


\section{Introduction}
The goal of the present paper is to establish the rigorous result on existence and uniqueness for the coupled system of space homogeneous Boltzmann equations modelling a mixture composed of $P\geq1$ species of monatomic and polyatomic gases proposed in \cite{LD-ch-ens}.  The idea of such kinetic model is to describe a state of each gas component with its own single-particle distribution function. Fixing one mixture component, say $\mathcal{A}_i$, for any $i\in \left\{1, \dots, P\right\}$, the main mechanism driving the change of the corresponding distribution function is its interactions with all other mixture constituents through Boltzmann-like bi-linear collision forms, leading to a coupled system governing the dynamics of the gas mixture as a whole.\\ 

A core ingredient of the Boltzmann flow is the bi-linear collision form or operator, an integral operator which describes the mutual species-species interactions. In the current setting of a mixture of monatomic and polyatomic molecules there are four types of collision operators depending on whether the fixed species $\mathcal{A}_i$ is monatomic or polyatomic and whether the collision partner of such component belongs to a monatomic or polyatomic mixture component.  We prove estimates on statistical moments of the vector valued collision operator that accounts for different possible interactions among species, which then allow to prove \textit{a priori} estimates on the system solution.  The collision operator's kernels are assumed to be of hard potential form for the total collisional energy for at least one species-species interaction per each mixture component allowing Maxwell interactions as well.   The scattering part of such kernels is assumed integrable, as much as the part related to the energy exchange variables for polyatomic mixture components. Such a family of collision kernels is shown to be relevant in physics and engineering  applications that involve polyatomic gas interactions. Applying a general ODE theory in Banach spaces, the Cauchy problem is resolved for the initial data corresponding to strictly positive and bounded initial species' mass and bounded mixture's total energy, and a $2^+$ moment determined by the hard potential rates discrepancy.  The presented result unify approaches for the classical single monatomic Boltzmann equation and recently obtained results in the case of monatomic mixtures and single polyatomic gases.\\  

The kinetic theory of polyatomic gases and mixtures has recently become an active field of research and the rigorous theory has been developed in certain  physical contexts that can be understood as a special case of the gas mixture model analysed in this paper.  For the system of Boltzmann equations describing a mixture of  solely monatomic gases, which can be seen as a sub-system of the present model when all $P$ species are monatomic,  in the linearized setting, well-posedness, stability, compactness, energy method and hypocoercivity-related issues were studied  in \cite{bri-dau, bri, bon-bou-bri-gre, nous, nicl-1, nous-LY, dau-jun-mou-zam}. In the spatial homogeneous case, questions about well-posedness  and regularity for  the system of non-linear Boltzmann equations, were addressed in \cite{MPC-IG-ARMA, Erica-pspde} with integrable angular part and in \cite{Alonso-Orf} for an angular part modelling long-range interactions (the so-called non-cutoff scattering). \\

Polyatomic gases bring an another level of difficulty. An underlying physical effect is the internal energy exchange during the collision, apart from the usual translational energy of the relative motion of the colliding particles. The microscopic internal energy can be modelled as discrete or continuous variable, leading to the two branches in the kinetic theory of polyatomic gases: the semi-classical \cite{Kusto-book, Gio, Groppi-Spiga} and the continuous \cite{LD-Toulouse, LD-Bourgat, DesMonSalv} approaches. A general framework unifying these two approaches was recently presented in \cite{Bisi-Borsoni-Groppi}.\\

In this paper, we focus on the continuous kinetic approach that uses  the Borgnakke-Larsen procedure for the collision parametrization \cite{Bor-Larsen},  which makes the model accessible both from the  rigorous analytical and computational points of view. From the particle perspective, this  parametrization can be interpreted as a DSMC algorithm for sampling particles' internal energy exchanges \cite{MPC-Dj-T-O}.  The corresponding Boltzmann equation for a single polyatomic gas ($P=1$) or a polyatomic gas mixture ($P>1$) can be also seen as a subsystem of our present model when all species are polyatomic. Recently, a compactness result has been obtained for the linearized polyatomic Boltzmann operator in \cite{Brull-Comp, Brull-Comp-2, nicl-2} and in \cite{Bors-Bou-Salv-comp} for the model of resonant collisions \cite{Bou-Salv-res}.  The global well-posedness for bounded mild solutions near global equilibria on torus is established in \cite{Duan-Li-poly}.  For the space homogeneous setting and the full non-linear Boltzmann operator, well-possedness and $L^1$ regularity were tackled in \cite{MPC-IG-poly}. In particular, a form of the collision kernel which corresponds to hard potentials in both relative velocity and microscopic internal energy is proposed. It is shown to be highly physically relevant, as it provides transport coefficients that match experimental data \cite{MPC-Dj-S, MPC-Dj-T} for polytropic or calorically perfect polyatomic gases, and contains as a special case the collision kernel used in DSMC method with the VHS cross-section \cite{MPC-Dj-T-O}. \\

Motivated by the success in the analysis of space homogeneous problems for  separately  monatomic mixtures \cite{MPC-IG-ARMA, Erica-pspde, Alonso-Orf} and a single polyatomic gas \cite{MPC-IG-poly}, the aim of this paper is to establish  the existence and uniqueness theory  for a system of Boltzmann equations describing a mixture that involves  both monatomic and polyatomic gases.  We consider the model with different types of collision operators  proposed in \cite{LD-ch-ens} and slightly modify it, in order to work with the  $L^1$ plain space in the energy variable. This setting corresponds to the non-weighted one as described in \cite{MPC-Dj-S},  and coincides with the model used in \cite{nicl-3}. Moreover, it reduces to \cite{LD-Bourgat} for the single polyatomic gas model, analysed  in \cite{MPC-IG-poly}.  The approach to prove existence and uniqueness is based on an abstract ODE theory \cite{Martin-ODE}, first proposed by Bressan \cite{Bressan} in the context of scalar kinetic equations. The method was recently revised in \cite{Alonso-cooling,IG-Alonso-BAMS}, and was successfully used not only in \cite{MPC-IG-ARMA, Erica-pspde, MPC-IG-poly}, but also in dissipative kinetic problems \cite{Alonso-cooling} and the weak wave turbulence models for stratified flows \cite{IG-wave-turbulence}.\\

The aforementioned approach is quite general in the context of kinetic operators with integrable kinetic kernels where gain and loss collision operators can be treated independently; the path to be followed is similar in all the these problems consisting of some key steps. The implementation of such steps vary from case to case, however in the current case, due to the complexity of the underlying interactions happening in the general systems considered here, the arguments are more intricate  leading to new ideas in the mathematical treatment which focus on the essential mechanisms of energy transfer between pairwise collisions.  The first notable improvement is capturing the general structure of the energy pairwise collision interchange given in Lemma \ref{Lemma: En Iden}.  Second, based on such general structure, we exploit the natural occurring averaging in the gain collision operators to show the dissipative character of each pairwise interaction.  This character manifests in the fact that higher statistical moments of pairwise collisions are uniformly controlled over time.  For the classical Boltzmann case this result is known as the Povzner Lemma and dates back to \cite{Bob-Moment-ineq, IG-Bob-Panf-inelastic, LD-mom-meth, Wenn}.  We address this result here more precisely as Compact Manifold (or angular) Averaging Lemma, presented and proved in Lemma \ref{Lemma: Averaging} for the general models considered here using a more sophisticated method of proof based on the decomposition of the pairwise interaction domain.  This method is particularly useful for pairwise polyatomic interactions where the averaging effect occurs in the scattering angle and the internal energy exchange variables in a complex manner.  The Compact Manifold  Averaging Lemma is complemented with a coerciveness estimate for the loss collision operator, done in Lemma \ref{Lemma bounds on Bij tilde}.  A dominant effect of the pairwise loss term with respect to the gain in terms of higher moments is deduced from these considerations, given in Lemma \ref{Lemma: coll op estimate}.  A final preparatory step consists in finding \textit{a priori} higher moment estimates of the Boltzmann system solution. Namely, a moment ODI satisfied by a suitable Banach space moment-norm of the solution is derived, in Lemma \ref{Lemma moment ODI}, which yields generation and propagation of statistical moments properties used later to implement an abstract ODI argument, here presented in Section \ref{Sec: exi uni}.\\ 

The analysis is performed under  a quite general assumption on the collision kernel or transition probability -- it is assumed to be of the Maxwell-hard potential form for the collisional total energy (which reduces to the relative speed when only monatomic interactions are involved) with possible different rates $\gamma_{ij}\geq0$ satisfying $\max_j \gamma_{ij} >0$.  In addition,  the  angular part is assumed integrable and the collision kernel is assumed to be bounded from above and below by integrable partition functions of the Borgnakke-Larsen procedure. This assumption is general enough to cover  already established theory for monatomic mixture and single polyatomic models and is compatible with the analysis of the linearized Boltzmann polyatomic operator performed in \cite{Brull-Comp-2, nicl-2, nicl-3, Duan-Li-poly}.\\

The existence and uniqueness of the system solution is proven  for initial data corresponding to strictly positive and bounded species mass (zero order species moment) and  mixture total energy  (second order mixture moment), which are all conserved quantities for the Boltzmann flow in the absence of chemical reactions (energy is interchanged between species though). Moreover, a mixture moment of the order $(2+\max_{ij}\{\gamma_{ij}\} - \min_{i}\max_{j}\{\gamma_{ij}\})^+$ is required to be bounded.  The approach developed in this paper has an optimal value of $2^+$ when $\max_{ij}\{\gamma_{ij}\} = \min_{i}\max_{j}\{\gamma_{ij}\}$, which improves results of \cite{MPC-IG-poly}.  Another significant improvement is to incorporate arguments of \cite{Alonso-Orf} and to allow the range of $\gamma_{ij}$ to be  $[0,2]$ with at least one strictly positive  $\gamma_{ij}$  for each component $i$, which extends the previous results  of \cite{MPC-IG-ARMA, Erica-pspde}  in the monatomic mixture case that require all $\gamma_{ij} > 0$.  All the analysis is performed for general integrable angular scattering.\\

We stress that the physical framework of multi-component gas mixtures composed of both monatomic and polyatomic gases  considered in this paper  is highly relevant in applications since, for instance, air itself is a mixture of monatomic (such as Ar, O, N) and polyatomic (such as O$_2$, N$_2$, CO$_2$) components.  Thus, it is necessary to transcend the classical Boltzmann equation that models an ideal gas composed of identical structureless particles.  This work is a step in that direction. \\

The paper is organized as follows. In Section \ref{Sec: Mixture model}, the system of Boltzmann equations describing a mixture of monatomic and polyatomic gases is presented. The notation and functional spaces are introduced in Section \ref{Sec: notation}, while assumptions on the collision kernel are listed in Section  \ref{Sec: Assump Bij}. Section \ref{Sec: Est coll op} deals with estimates on the vector valued  collision operator. Namely, we prove various energy identities and estimates for different types of species-species interactions in Section \ref{Section: en id} that allow to prove the Averaging Lemma in Section \ref{Section: Av lemma}, yielding estimates on statistical moments for the collision operator firstly written in a bi-linear form in Section \ref{Sec: bi-linear form coll op} and then in the vector valued form in Section \ref{Sec: estimates vector coll op}. This study allows to derive polynomial moments \textit{a priori} estimates  on the system solution in Section \ref{Sec: poly mom}. Finally, existence and uniqueness theory is established in Section \ref{Sec: exi uni} and some  technical results and general theorems used in the paper are listed in Appendix \ref{Sec: App}.

\section{System of Boltzmann equations modelling a gas mixture composed of monatomic and polyatomic gases}\label{Sec: Mixture model}

We consider a mixture of $M$ monatomic and $(P-M)$ polyatomic gases.  Each monatomic component  $\mathcal{A}_i$, $i=1,\dots,M$, is   described with the distribution function $f_i(t,v)\geq0$ depending on time $t>0$ and molecular velocity $v\in \bR$. Polyatomic gases are modelled based on the continuous internal energy approach \cite{LD-Bourgat, LD-Toulouse, DesMonSalv} which amounts  to assume that a polyatomic component $\mathcal{A}_i$, $i=M+1,\dots,P$, of the mixture is described with the distribution function $f_i(t,v,I)\geq0$ depending on time $t>0$, molecular velocity $v \in \bR$ and also molecular microscopic internal energy $I \in \bRp$.   Assuming that distribution functions change due to mutual interactions, the dynamics of the mixture is characterized by the system of Boltzmann equations, here written in the space homogeneous setting,
\begin{equation}\label{BS}
	\begin{split}
	&\partial_tf_i(v) = \sum_{j=1}^{P} Q_{ij}(f_i, f_j)(v), \qquad i=1,\dots,M, \\
	&\partial_tf_i(v, I) = \sum_{j=1}^{P} Q_{ij}(f_i, f_j)(v, I), \qquad i=M+1,\dots,P,
	\end{split}
\end{equation}
where for the brevity we omit to write dependence on $t$. Note that \eqref{BS} contains    four collision operators that describe collisions of various type of gases and, as such, are different in nature. Namely, for the fixed monatomic species $\mathcal{A}_i$, $i\in \{1,\dots,M\}$, the pairwise interaction can be mono-mono when $j\in \left\{1, \dots,M\right\}$ or mono-poly for $j\in \left\{M+1, \dots, P\right\}$. Similarly, when the species $\mathcal{A}_i$ is polyatomic i.e. $i\in\{M+1,\dots,P\}$, we distinguish  poly-mono interaction when $j\in \left\{1, \dots,M\right\}$  or poly-poly for  $j\in \left\{M+1, \dots, P\right\}$. The form of  corresponding four collision operators  was introduced in \cite{LD-ch-ens} following the approach given in \cite{LD-Toulouse, DesMonSalv}, called the weighted setting in \cite{MPC-Dj-S}. In this paper, motivated by the rigorous analysis for a single polyatomic gas \cite{MPC-IG-poly}, we will follow \cite{LD-Bourgat} and rewrite the collision operators of \cite{LD-ch-ens} in the non-weighted setting in the spirit of \cite{MPC-Dj-S}.

The system of Boltzmann-like equations \eqref{BS} can be written in a vector form by introducing the vector valued distribution function $\F$ and the vector valued collision operator $\Q(\F)$ 
\begin{equation}\label{F Q vector}
		\F = \left[ \begin{matrix}  \Big[f_i(t, v)\Big]_{i=1,\dots,M} \\  \Big[f_i(t, v, I)\Big]_{i=M+1,\dots,P} \end{matrix} \right], \qquad 	
		\Q(\F) = \left[ \begin{matrix}  \Big[ \sum_{j=1}^{P} Q_{ij}(f_i, f_j)(v)\Big]_{i=1,\dots,M} \\  \Big[ \sum_{j=1}^{P} Q_{ij}(f_i, f_j)(v, I) \Big]_{i=M+1,\dots,P} \end{matrix} \right].
\end{equation}
Therefore, the system \eqref{BS}   in the vector valued form, and together with initial data reads
\begin{equation}\label{BS vector}
	\partial_t \F = \Q(\F), \qquad \F(0)=\F_0.
\end{equation}
The next section introduces the collision operator $\Q(\F)$ for each of  four possible binary interactions between molecules of monatomic and polyatomic species.

\subsection{Collision operators for interaction between  monatomic gases} 
For  $i,j\in \left\{ 1, \dots, M \right\}$,   the collision operator in a bilinear form for distribution functions  $f(t, v)\geq0$ and $g(t,v)\geq 0$ describing species $\mathcal{A}_i$ and $\mathcal{A}_j$, respectively, reads
\begin{equation}\label{m-m Q}
Q_{ij}(f,g)(v) = \int_{\bR} \int_{\bS} \left( f(v')g(v'_*)-f(v)g(v_*) \right)  \mathcal{B}_{ij}(v,v_*,\sigma) \, \md \sigma \, \md v_*,
\end{equation}
where $v'$ and $v'_*$ are pre-collisional velocities expressed as functions of the post-collisional ones $v$, $v_*$ and a parameter $\sigma$ in the center-of-mass framework defined by   vectors of the center of mass velocity  $V$, relative velocity $u$ and  reduced mass $\mu_{ij}$,
\begin{equation}\label{V-u}
	V= \frac{m_i v + m_j v_*}{m_i + m_j}, \qquad u=v-v_*, \qquad 	\mu_{ij} = \frac{m_i m_j}{m_i + m_j},
\end{equation}
namely,
\begin{equation}\label{m-m primed velocity}
	v' = V + \frac{m_j}{m_i + m_j} |u| \sigma, \qquad 	v'_* = V - \frac{m_i}{m_i + m_j} |u| \sigma,
\end{equation}
while the collision kernel   $\mathcal{B}_{ij} $ satisfies the following microreversibility assumptions
\begin{equation}\label{m-m cross}
\mathcal{B}_{ij} := 	\mathcal{B}_{ij}(v, v_*, \sigma) = 	\mathcal{B}_{ij}(v', v'_*, \sigma')  =	\mathcal{B}_{ji}(v_*, v, -\sigma) \geq 0.
\end{equation}
Equations \eqref{m-m primed velocity} are $\sigma-$parametrization of the conservation laws of momentum and kinetic energy of a colliding pair of molecules, namely,
\begin{equation}\label{m-m coll CL}
m_i v' + m_j v'_* = m_i v + m_j v_*, \qquad \frac{m_i}{2} |v'|^2 + \frac{m_j}{2} |v'_*|^2 = \frac{m_i}{2} |v|^2 + \frac{m_j}{2} |v_*|^2,
\end{equation}
or equivalently
\begin{equation}\label{m-m coll CL 2}
	V' = V, \qquad |u'|=|u|.
\end{equation}
The kinetic energy can be represented in the center-of-mass framework, 
\begin{equation}\label{m-m en}
	E_{ij} = \frac{\mu_{ij}}{2} |u|^2,
\end{equation}
which is by \eqref{m-m coll CL 2} a conserved quantity.  

The collision operator weak form is carefully  explained in \cite{Pavan-B-G}, i.e. for any suitable test functions $\omega(v)$ and $\chi(v)$,
\begin{multline}\label{m-m weak form}
	\int_{\bR} Q_{ij}(f,g)(v) \, \omega(v) \, \md v +  \int_{\bR} Q_{ji}(g,f)(v) \, \chi(v) \, \md v\\
	=   \int_{(\bR)^2} \int_{\bS } \left\{  \omega(v') + \chi(v'_*) - \omega(v) - \chi(v_*) \right\} \, f(v) \, g(v_*)  \mathcal{B}_{ij}(v, v_*, \sigma) \, \md \sigma \,\md v_* \, \md v.
\end{multline}

\subsection{Collision operators  for interaction between  polyatomic gases} Let $i,j\in \left\{ M+1, \dots,  P \right\}$, i.e. we consider two colliding polyatomic molecules. Let one be of mass $m_i$ and velocity-internal energy $(v', I')$ and the another one of mass $m_j$ and velocity-internal energy $(v'_*, I'_*)$, that change to $(v, I)$ and $(v_*, I_*)$ (with same masses), respectively, after the collision. We assume that collisions are elastic in the sense that momentum and the total (kinetic + microscopic internal) energy are conserved during the collision,
 \begin{equation}\label{p-p coll CL}
 	m_i v' + m_j v'_*  = m_i v + m_j v_*, \qquad \frac{m_i}{2} |v'|^2 + I' +  \frac{m_j}{2} |v'_*|^2 + I'_* = \frac{m_i}{2} |v|^2 + I +  \frac{m_j}{2} |v_*|^2 + I_*.
 \end{equation}
These laws can be rewritten in the center-of-mass framework \eqref{V-u}, namely,
\begin{equation}\label{p-p V E}
	V' = V, \quad E'_{ij} = E_{ij}, 
\end{equation}  
with the energy
\begin{equation}\label{p-p en}
E_{ij} = \frac{\mu_{ij}}{2} |u|^2 + I +  I_*.
\end{equation}  
In order to express pre--collisional velocities and internal energies in the original particle framework, the so-called  Borgnakke-Larsen procedure \cite{Bor-Larsen, LD-Bourgat} is used. The idea is to parametrize equations \eqref{p-p V E} with the scattering direction $\sigma\in\bS$ and  energy exchange variables $R, r \in[0,1]$. More precisely,  first split the kinetic and internal energy part of the total energy with $R \in [0,1]$,
\begin{equation*}
\frac{\mu_{ij}}{2}  \left|u'\right|^2 = R E, \quad  I' + I'_* =  (1-R) E.
\end{equation*}
Then, the kinetic part is distributed among particles with $\sigma \in \bS$ and total internal energy is split on particles' internal energies with   $r\in [0,1]$,
\begin{equation}\label{p-p coll rules}
	\begin{alignedat}{2}
		v'  &= V + \frac{m_j}{m_i + m_j} \sqrt{\frac{2 \, R  \, E_{ij}}{\mu_{ij}}} \sigma,  \qquad
		I' &&=r (1-R) E_{ij}, \\
		v'_{*} &= V - \frac{m_i}{m_i + m_j} \sqrt{\frac{2 \,  R \,  E_{ij}}{\mu_{ij}}} \sigma,
		\qquad I'_* &&= (1-r)(1-R) E_{ij}. 
	\end{alignedat}
\end{equation}
For the convenience, we introduce the primed parameters as well, $\sigma' \in \bS$, $r', R' \in [0,1]$,
\begin{equation}\label{p-p primed param}
\sigma'=\hat{u}=\frac{u}{\left|u\right|}, \qquad R'=\frac{\mu_{ij} \left|u\right|^2}{2 E_{ij}}, \qquad	r'=\frac{I}{I+I_*} = \frac{I}{E_{ij}-\frac{\mu_{ij}}{2}\left| u \right|^2}.
\end{equation}
The transformation $\mathcal{T}_{pp}: (v, v_*, I, I_*, r, R, \sigma) \mapsto (v', v'_*, I', I'_*, r', R', \sigma')$ is an involution and its Jacobian, which will deeply influence the structure of collision operator, is computed in \cite{DesMonSalv},
\begin{equation}\label{p-p Jac}
J_{\mathcal{T}_{pp}} : = \left| \frac{\pa(v', v'_{*},  I', \sigma', r' R')}{\pa (v,v_*,I_*,\sigma, r, R)  } \right| = \frac{(1-R) \sqrt{R}}{(1-R') \sqrt{R'}}.
\end{equation}
The collision operator for  distributions functions $f:= f(t,v,I)\geq 0$ and $g:= g(t,v,I)\geq0$  describing species $\mathcal{A}_i$ and $\mathcal{A}_j$, respectively,  used in this paper is a variant of the one proposed in \cite{DesMonSalv}, 
\begin{multline}\label{p-p coll operator}
	Q_{ij}(f,g)(v,I) = \int_{\bRfp} \int_{\bS \times [0,1]^2} \left\{ f(v',I')g(v'_*,I'_*) \left(\frac{I }{I' }\right)^{\alpha_i } \left(\frac{ I_*}{ I'_*}\right)^{ \alpha_j}   - f(v,I)g(v_*, I_*) \right\} \\ \times \mathcal{B}_{ij}(v, v_*, I, I_*, \sigma, r,  R) \, \dpo  \, \md \sigma \, \md r \, \md R \, \md v_* \, \md I_*,
\end{multline}
where the pre-collisional quantities $v'$, $v'_*$, $I'$ and $I'_*$  are defined in \eqref{p-p coll rules}, $\alpha_i, \alpha_j >-1$ are constants related to the specific heats of the polyatomic gas \cite{MPC-Dj-T}, parameters' function given by
\begin{equation}\label{fun r R}
\dpo =	  r^{\alpha_i}(1-r)^{\alpha_j} \, (1-R)^{\alpha_i + \alpha_j+1} \, \sqrt{R},
\end{equation}
and the collision kernel  $	\mathcal{B}_{ij}$ satisfies the microreversibility assumptions 
\begin{equation}\label{p-p Bij micro}
\mathcal{B}_{ij} :=	\mathcal{B}_{ij}(v, v_*, I, I_*, \sigma, r, R) = 	\mathcal{B}_{ij}(v', v'_*, I', I'_*, \sigma', r', R') = \mathcal{B}_{ji}(v_*, v, I_*, I, -\sigma, 1-r, R) \geq 0,
\end{equation}
corresponding to the interchange of pre and post-collisional molecules, and the interchange of colliding  molecules. It is important to notice that the choice of functions \eqref{fun r R} depends on the weight factor in the collision operator \eqref{p-p coll operator} because it ensures its invariance, since the factor
\begin{equation*}
\left(	r (1-R) I \right)^{\alpha_i}  \left(	(1-r)(1-R) I_* \right)^{\alpha_j}  
\end{equation*}
is invariant with respect to the changes described in \eqref{p-p Bij micro}. Together with the Jacobian \eqref{p-p Jac}, this implies the invariance of the measure 
\begin{equation*}
I^{\alpha_i} I_*^{\alpha_j} \dpo \, \md \sigma \, \md r \, \md R \, \md v_* \, \md I_* \, \md v \, \md I,
\end{equation*}
with respect to changes from \eqref{p-p Bij micro}. These considerations yield the well defined weak form, 
\begin{multline}\label{p-p weak form}
		 \int_{\bRfp} Q_{ij}(f,g)(v, I) \, \omega(v,I) \, \md v\, \md I +  \int_{\bRfp} Q_{ji}(g,f)(v, I) \, \chi(v,I) \, \md v\, \md I  \\
	  =  \int_{(\bRfp)^2} \int_{\bS \times [0,1]^2} \left\{  \omega(v', I') + \chi(v'_*, I'_*) - \omega(v, I) - \chi(v_*, I_*) \right\} \, f(v, I) \, g(v_*, I_*) 
	\\  \times \mathcal{B}_{ij}(v, v_*, I, I_*, r,  \sigma, R) \, \dpo \, \md \sigma \, \md r \, \md R \, \md v_* \, \md I_* \, \md v \, \md I,
\end{multline}
for any suitable test functions $\omega(v,I)$ and $\chi(v,I)$. 

\subsection{Collision operators  for interaction between  monatomic and polyatomic gases} \label{Sec: mon-poly interaction}
Let   $\mathcal{A}_j$,  $j\in \left\{ 1, \dots, M \right\}$, be a  monatomic component of the mixture described with the distribution function $g(t,v) \geq 0$ and     $\mathcal{A}_i$, $i\in \left\{ M+1, \dots, P \right\}$, the polyatomic component characterized by the distribution function $f(t,v,I) \geq 0$. 

Since  molecules differ in nature, the corresponding study of molecular collisions will depend on  which molecule the internal energy is associated to -- whether to the molecule of interest or the partner in collision. This raises the definition of two different collision operators, 
\begin{enumerate}
	\item[(i)] collision operator $Q_{ij}(f,g)(v,I)$ describing   the influence of a monatomic component  $\mathcal{A}_j$  on the polyatomic one $\mathcal{A}_i$ and acting on $(v,I)$ pair (poly-mono interaction),
	\item[(ii)] collision operator $Q_{ji}(g,f)(v)$  describing   the influence of a polyatomic component $\mathcal{A}_i$ on the monatomic one $\mathcal{A}_j$ and acting on $v$ only (mono-poly interaction).
\end{enumerate}
\subsubsection{Case (i) -- study of poly-mono interaction} 
We consider a pair of colliding molecules, one  polyatomic  molecule of mass $m_i$ and velocity-internal energy $(v',  I')$ which collides with the  monatomic  partner of mass $m_j$ and velocity $v'_*$.  After the collision, they belong to the same species so masses do not change, but the velocity--internal energy pair of a polyatomic molecule become $(v, I)$, while the velocity of the monatomic molecule changes to $v_*$. This collision is assumed elastic, in the sense that momentum and total energy are preserved during the collision, 
\begin{equation}\label{p-m coll CL}
	m_i v' + m_j v'_*  = m_i v + m_j v_*, \qquad \frac{m_i}{2} |v'|^2 + I' +  \frac{m_j}{2} |v'_*|^2  = \frac{m_i}{2} |v|^2 + I +  \frac{m_j}{2} |v_*|^2.
\end{equation}
For this setting we introduce the center-of-mass framework \eqref{V-u} together with the energy 
\begin{equation}\label{p-m en}
 E_{ij} =  \frac{\mu_{ij}}{2} |u|^2 + I.
\end{equation}  
Then, laws \eqref{p-m coll CL} are equivalent to
\begin{equation}\label{p-m V E}
	V'= V, \quad   E'_{ij} = E_{ij}.
\end{equation}  
These laws are parametrized with the angular parameter $\sigma \in \bS$ and a parameter $R\in[0,1]$,
\begin{equation}\label{p-m coll rules}
	\begin{split}
	v'  = V + \frac{m_j}{m_i + m_j} \sqrt{\frac{2 \, R  \, E_{ij}}{\mu_{ij}}} \sigma,  \qquad
	v'_{*} = V - \frac{m_i}{m_i + m_j} \sqrt{\frac{2 \,  R \,  E_{ij}}{\mu_{ij}}} \sigma,\qquad
	I' = (1-R)\, E_{ij}\,.
	\end{split}
\end{equation}
We complement these equations with the definition of primed parameters $\sigma' \in \bS$,  $ R' \in [0,1]$,
\begin{equation}\label{p-m coll rules par}
	\sigma' = \hat{u} = \frac{u}{|u|}, \qquad	R' = \frac{\mu_{ij}|u|^2}{2 \, E_{ij}}.
\end{equation}

\subsubsection{Case (ii) --  study of mono-poly interaction}  Let consider the   counterpart problem for the \textit{Case (i)}. Now we fix the monatomic molecule with mass $m_j$ and velocity $w'$ which changes to $v$ due to  a collision with the  polyatomic molecule partner   of mass $m_i$ and velocity-internal energy pair $(w'_*, I'_*)$ that changes to $(v_*, I_*)$ after the collision. During the collision, the following conservation laws hold
\begin{equation}\label{m-p coll CL}
	m_j w' + m_i w'_*  = m_j v + m_i v_*, \qquad \frac{m_j}{2} |w'|^2 + \frac{m_i}{2} |w'_*|^2 + I'_*= \frac{m_j}{2} |v|^2 + \frac{m_i}{2} |v_*|^2 + I_*.
\end{equation}
Introducing the center-of-mass reference framework with \eqref{V-u} and the center of mass velocity  $W$ which differs from  $V$ by the mass  interchange $m_i \leftrightarrow m_j$,
\begin{equation}\label{W}
	W= \frac{m_j v + m_i v_*}{m_i + m_j},
\end{equation}
\eqref{m-p coll CL} can be rewritten  as
\begin{equation}\label{m-p V E}
	W' = W, \quad E'_{ji} = E_{ji}:= \frac{\mu_{ji}}{2} |u|^2 + I_*.
\end{equation}  
Similarly as in \eqref{p-m coll rules}, these equations are parametrized with $\sigma \in \bS$ and   $R\in[0,1]$,
\begin{equation}\label{m-p coll rules}
		w' = W + \frac{m_i}{m_i + m_j} \sqrt{\frac{2 \, R  \, E_{ji}}{\mu_{ji}}} \sigma, \qquad
		w'_{*} = W  - \frac{m_j}{m_i + m_j} \sqrt{\frac{2 \,  R \,  E_{ji}}{\mu_{ji}}} \sigma, \qquad
		I'_{*} = (1-R)\, E_{ji}.
\end{equation}
For convenience, we also express primed parameters $(\sigma', R')$ in terms of non-primed quantities, 
\begin{equation}\label{m-p coll rules par}
	\sigma' = \hat{u} =  \frac{u}{|u|}, \qquad	R' = \frac{\mu_{ji}|u|^2}{2 \, E_{ji}}.
\end{equation}

\begin{lemma}\label{Lemma mono-poly invariance}
Let $i\in \left\{ M+1, \dots, P \right\}$  and	$j\in \left\{ 1, \dots, M \right\}$. Let $\alpha_i > -1$ and define
\begin{equation}\label{psi m-p}
 \dm = (1-R)^{\alpha_i} \, \sqrt{R}.
\end{equation}
Consider  transformations 
\begin{alignat}{2}
&\mathcal{T}_{pm}: (v,v_*,I,\sigma,R) & \rightarrow & \ (v', v'_{*},  I', \sigma', R'), \quad \text{defined by \eqref{p-m coll rules}--\eqref{p-m coll rules par}}, \label{coll tr p-m}\\
&\mathcal{T}_{mp}: (v,v_*,I_*,\sigma,R) \ &\rightarrow & \ (w', w'_*,  I'_{*}, \sigma', R'), \quad \text{defined by \eqref{m-p coll rules}--\eqref{m-p coll rules par}}. \label{coll tr m-p}
\end{alignat}
\begin{description}
	\item[\textit{Part (i)}]  Transformations $\mathcal{T}_{pm}$ and $\mathcal{T}_{mp}$ are involutions and theirs  Jacobians, respectively,  $J_{\mathcal{T}_{pm}} : = \left| \frac{\pa(v', v'_{*},  I', \sigma',R')}{\pa (v,v_*,I_*,\sigma,R)  } \right|$ and $J_{\mathcal{T}_{mp}} : = \left| \frac{\pa(w', w'_{*},  I'_{*}, \sigma', R')}{\pa (v,v_*,I_*,\sigma,R)}  \right|$, are
	\begin{equation}\label{m-p Jac}
	J_{\mathcal{T}_{pm}} = J_{\mathcal{T}_{mp}}  = \frac{\sqrt{R}}{\sqrt{R'}},
	\end{equation}
	 where $R'$ is understood  as \eqref{p-m coll rules par} for $\mathcal{T}_{pm}$ and as     \eqref{m-p coll rules par} for $\mathcal{T}_{mp}$. 
	 	\item[\textit{Part (ii.a)}] The  following measure is invariant under the collision transformation $\mathcal{T}_{pm}$,
	 \begin{equation}
	   I^{\alpha_i}  \, \dm \, \md R \,  \md \sigma  \, \md v_*  \, \md v  \, \md I,
	 \end{equation}
		\item[\textit{Part (ii.b)}] The  following measure is invariant under the collision transformation $\mathcal{T}_{mp}$,
		\begin{equation}
		I_*^{\alpha_i} \, \dm \, \md R \,  \md \sigma  \, \md v_* \,  \md I_* \, \md v.
		\end{equation}
\end{description}

\end{lemma}

\subsubsection{Interchange of the collision reference}
The aforementioned considerations concern relations connecting pre-- and post--collisional quantities.  It remains  to study a transformation describing the interchange of a collision reference in the case of mixed poly-mono and mono-poly interactions. Take a monatomic molecule of mass $m_j$, $ j\in \{1,\dots,M\}$,  and velocity $v_*$ and  a polyatomic molecule of mass $m_i$, $   i\in \{M+1,\dots,P\}$, and velocity-internal energy $(v,I)$, so the poly-mono interaction. Collision reference  interchange is a transformation constructed such that  if  $(v_*, v, I) \leftrightarrow (v, v_*, I_*)$ then the same change should hold before the collision i.e. $(v'_*, v', I') \leftrightarrow (w', w'_*, I'_*)$. 
Thus, we are led to consider the transformation
\begin{equation}\label{interchange of particles}
	\mathcal{I}: (v,v_*,I,\sigma,R)  \leftrightarrow (v_*,v,I_*, -\sigma,R).
\end{equation}
Then,   for the fixed $i \in\{M+1,\dots,P\}$ and $j \in\{1,\dots,M\}$,  $  \mathcal{T}_{pm}(	\mathcal{I}(v,v_*,I,\sigma,R)) \allowbreak = (w'_*, w', I'_*, -\sigma', R')$ and  $ \mathcal{T}_{mp}(	\mathcal{I}(v,v_*,I,\sigma,R)) \allowbreak = (v'_*, v', I', -\sigma', R')$. 

\subsubsection{Case (i) -- collision operator for poly-mono interaction}  The influence of a monatomic component  $\mathcal{A}_j$ described by distribution function $g(t, v) \geq 0$  on the polyatomic one $\mathcal{A}_i$ with distribution function $f(t,v,I) \geq 0$ is captured with the collision operator
\begin{multline}\label{p-m coll operator}
	Q_{ij}(f,g)(v,I) = \int_{\bR} \int_{\bS \times [0,1]} \left\{ f(v',I')g(v'_*) \left(\frac{I}{I'}\right)^{\alpha_i}  - f(v,I)g(v_*) \right\} \\ \times \mathcal{B}_{ij}(v, v_*, I, \sigma, R) \,  \dm \, \md \sigma \, \md R \, \md v_*,
\end{multline}
where the measure $\dm$ is given in \eqref{psi m-p} and the collision kernel  satisfies the following micro-reversibility properties
\begin{equation}\label{p-m cross}
	\mathcal{B}_{ij} :=	\mathcal{B}_{ij}(v, v_*, I, \sigma, R) = 	\mathcal{B}_{ij}(v', v'_*, I', \sigma', R') = \mathcal{B}_{ji}(v_*, v, I, -\sigma, R) \geq 0,
\end{equation}
with $v'$, $v'_*$, $I'$ as defined in \eqref{p-m coll rules}.

\subsubsection{Case (ii) -- collision operator for mono-poly interaction}  The influence of a polyatomic component $\mathcal{A}_i$ characterized via distribution function $f(t, v, I) \geq 0$ on the monatomic gas component $\mathcal{A}_j$ with distribution function $g(t, v) \geq 0$  is described  with the collision operator
\begin{multline}\label{m-p coll operator}
	Q_{ji}(g,f)(v) = \int_{\bRfp} \int_{\bS \times [0,1]} \left\{ g(w')f(w'_*,I'_*) \left(\frac{I_*}{I'_*}\right)^{\alpha_i} - g(v)f(v_*,I_*) \right\} 
	\\ \times  \mathcal{B}_{ji}(v,v_*,I_*, \sigma,R) \, \dm \, \md \sigma \, \md R \, \md v_* \, \md I_*,
\end{multline}
where  the measure $\dm$ is given in \eqref{psi m-p}  and the collision kernel  is assumed to satisfy the micro-reversibility properties
\begin{equation}\label{m-p micro rev}
	\mathcal{B}_{ji}:=	\mathcal{B}_{ji}(v, v_*, I_*, \sigma, R) = 	\mathcal{B}_{ji}(w', w'_*, I'_*, \sigma', R'), \ \text{and} \
		\mathcal{B}_{ji}(v_*, v, I, -\sigma, R) =	\mathcal{B}_{ij}(v, v_*, I, \sigma, R),
\end{equation}
and 
$w'$, $w'_*$, $I'_*$ are given in \eqref{m-p coll rules}.

\subsubsection{Weak form of the collision operators  for poly-mono \& mono-poly interactions}
Lemma \ref{Lemma mono-poly invariance} ensures a well-defined weak form of the collision operators  \eqref{p-m coll operator} and \eqref{m-p coll operator}.  The conservative form of the weak formulation is obtained when all operators in the interaction of two species are considered simultaneously. Indeed, for any suitable test functions $ \omega(v,I) $ and $\chi(v)$, the following weak form is obtained \cite{LD-ch-ens},
\begin{align}
&	 \int_{\bRfp} Q_{ij}(f,g)(v, I) \, \omega(v,I) \, \md v\, \md I + \int_{\bR} Q_{ji}(g,f)(v) \, \chi(v) \md v \nonumber \\
&  = \int_{\bR \times \bR \times \bRp} \int_{\bS \times [0,1]} \left\{  \omega(v', I') + \chi(v'_*) - \omega(v, I) - \chi(v_*) \right\} \, f(v, I) \, g(v_*) \label{weak vs v I} \\
& \qquad \qquad \qquad \qquad \qquad \qquad \qquad \times \mathcal{B}_{ij}(v, v_*, I, \sigma, R) \, \dm \, \md \sigma \, \md R \, \md v_* \, \md v \, \md I \nonumber \\ 
&  = \int_{\bR \times \bR \times \bRp} \int_{\bS \times [0,1]} \left\{ \chi(w') + \omega(w'_*, I'_*)  - \chi(v) -  \omega(v_*, I_*)  \right\} \, f(v_*, I_*) \, g(v) \label{weak vs v Is} \\
& \qquad \qquad \qquad \qquad \qquad \qquad \qquad \times \mathcal{B}_{ji}(v, v_*, I_*, \sigma, R) \, \dm \, \md \sigma \, \md R \, \md v_* \, \md v \, \md I_*, \nonumber 
\end{align}
where the involved quantities are detailed throughout the Section \ref{Sec: mon-poly interaction}.  

\subsection{Weak form of the vector valued collision operator} 
Take a suitable vector valued test function $X$
\begin{equation}\label{test}
	X = \left[ \begin{matrix}  \Big[\chi_i(v)\Big]_{i=1,\dots,M} \\  \Big[\chi_i( v, I)\Big]_{i=M+1,\dots,P} \end{matrix} \right].
\end{equation}
Collecting collision operator weak forms for each pair of species, stated above in \eqref{m-m weak form}, \eqref{weak vs v I} and \eqref{p-p weak form}, the conservative weak form of the vector valued collision operator defined in \eqref{F Q vector} reads
\begin{equation}\label{weak form general}
	\begin{split}
&\sum_{i=1}^M   \int_{\bR} \left[ \Q(\F)(v) \right]_i \, \chi_i(v) \md v +  \sum_{i=M+1}^P  \int_{\bRfp} \left[ \Q(\F)(v, I) \right]_i  \, \chi_i(v, I) \, \md v \, \md I
		\\
		& =  \sum_{i, j =1}^M  \int_{\bR} Q_{ij}(f_i, f_j)(v) \,  \chi_i(v) \md v  +    \sum_{i, j=M+1}^P  \int_{\bRfp} Q_{ij}(f_i, f_j)(v, I) \, \chi_i(v, I) \, \md v \, \md I
		\\
		& +  \sum_{j=1}^M  \sum_{i=M+1}^P \left( \int_{\bR} Q_{ji}(f_j, f_i)(v) \,  \chi_j(v)\,\md v +   \int_{\bRfp} Q_{ij}(f_i, f_j)(v, I) \,  \chi_i(v, I) \, \md v \, \md I \right).
	\end{split}
\end{equation}
Note that  conservation laws of particles' energies \eqref{m-m coll CL}, \eqref{p-p coll CL} and \eqref{p-m coll CL} imply for the choice $\chi_{i}(\cdot) = \la \cdot \ra_i^2$, 
\begin{equation}\label{weak form energy}
\sum_{i=1}^M   \int_{\bR} \left[ \Q(\F)(v) \right]_i \,  \la v \ra_i^2 \, \md v +  \sum_{i=M+1}^P  \int_{\bRfp} \left[ \Q(\F)(v, I) \right]_i  \,  \la v, I \ra_i^2 \, \md v \, \md I=0.
\end{equation}
If $\chi_{i}(\cdot) =1$,  then the following conservation per each species holds
\begin{equation}\label{weak form mass}
  \int_{\bR} \left[ \Q(\F)(v) \right]_i \,   \md v  =0, \ i=1,\dots,M, \ \text{and} \   \int_{\bRfp} \left[ \Q(\F)(v, I) \right]_i \, \md v \, \md I=0, \ i=M+1,\dots,P.
\end{equation}

\section{Notation and Functional spaces}\label{Sec: notation}

A natural space to solve the space homogeneous Boltzmann equation is the space of integrable functions appropriately weighed. In the present setting, we define Lebesgue brackets 
	\begin{equation}\label{brackets}
	\begin{split}
		\la v \ra_i &= \sqrt{1 + \frac{m_i}{2 \, \mM}  |v|^2 }, \quad \text{for} \ i =1,\dots, M, \\
		\la v, I \ra_i &= \sqrt{1 + \frac{m_i}{2\, \mP}  |v|^2  + \frac{1}{ \mP}  I  },  \quad \text{for} \ i =M+1,\dots, P,
\qquad \text{with} \quad	\mP =	\sum_{\ell=1}^{P} m_\ell.
	\end{split}
\end{equation}
The Banach space associated to the mixture constituent $\mathcal{A}_i$ is defined by
 \begin{equation}\label{Lqi}
 \begin{split}
\Lqi &= \left\{ f: 
\left\| f \right\|_{\Lqi} := \int_{\bR} \left| f(v) \right| \la v \ra_i^q  \ \md v < \infty   \right\}, \quad i \in \left\{ 1, \dots, M \right\}, \\
\Lqi &= \left\{ f: 
\left\| f \right\|_{\Lqi} := \int_{\bRfp} \left| f(v,I) \right| \la v,I \ra_i^q  \ \md v \, \md I < \infty   \right\}, \quad i \in \left\{ M+1, \dots, P \right\},
 \end{split}
 \end{equation}
for any $q\geq 0$. The same notion is introduced for the vector valued Banach spaces associated to the whole mixture,
  \begin{equation}\label{Lq}
 		\Lq = \left\{ \F= \left[f_i \right]_{i=1,\dots,P}:  \left\| \F \right\|_{\Lq} := 	 \sum_{i=1}^P  \left\| f_i \right\|_{\Lqi}   < \infty   \right\}.
 \end{equation}
A closely related concept is the one of polynomial moments. 
\begin{definition}
	 The $i$-th polynomial moment of order $q\geq 0$  for a suitable function  $f(t,v)$ for $i =1,\dots, M$, and $f(t,v,I)$ for $i =M+1,\dots, P$, is defined   by
	\begin{equation}\label{poly moment i}
			\begin{split}
			\mathfrak{m}^i_q[f](t) & = \int_{\bR} f(t,v) \, \la v \ra_i^{q} \, \md v,  \quad \text{for} \ i =1,\dots, M, \\
			\mathfrak{m}^i_q[f](t) & = \int_{ \bRfp} f(t,v,I) \, \la v,I \ra_i^{q} \, \md v \, \md I,  \quad \text{for} \ i =M+1,\dots, P\,.
		\end{split}
	\end{equation}
\end{definition} 
Note that when $f \geq 0$, the notion of polynomial moment coincides with $L^1$ norms. Moreover, for any $i$ whether $i \in \left\{ 1, \dots, M \right\}$ or $i \in \left\{ M + 1, \dots, P \right\}$, the monotonicity property holds
\begin{equation}\label{monotonicity of norm}
	\mathfrak{m}^i_{q_1}[f](t) \leq	\mathfrak{m}^i_{q_2}[f](t) , \text{whenever} \ 0\leq q_1 \leq q_2.
\end{equation}
\begin{definition}\label{def moment}
	{The  polynomial moment} of order $q\geq 0$ for a suitable vector valued  function \eqref{F Q vector}
	 is defined with
\begin{equation}\label{moment vector}
	\begin{split}
		\mathfrak{m}_q[\F](t)  = \sum_{i=1}^P 	\mathfrak{m}^i_{q}[f_i](t)   =	 \sum_{i=1}^M \int_{\bR} f_i(t, v) \la v \ra_i^{q} \, \md v +  \sum_{i=M+1}^P \int_{\bRfp} f_i(t, v,I) \la v, I \ra_i^{q} \, \md v \, \md I\,.
	\end{split}
\end{equation}
\end{definition} 
The $i$-th polynomial  moment of order 0, $\mathfrak{m}_0^i[f]$, has a physical intuition of  $\mathcal{A}_i$-th species number density described with the distribution function $f\geq 0$, whereas  polynomial moment of the second order, $\mathfrak{m}_2[\F]$, is physically interpreted as the sum of   number density and  total specific energy  of the mixture described with the vector valued distribution function $\F\geq0$. Note that by conservative properties of the weak form \eqref{weak form mass} and \eqref{weak form energy} these quantities are conserved for the Boltzmann flow.

\section{Assumptions on the collision kernels }\label{Sec: Assump Bij}
In this section, we summarize the assumptions we impose on the collision kernels  $\mathcal{B}_{ij}$, $i,j\in\left\{1, \dots, P \right\}$, that depend on the nature of the interactions. Our aim is to cover as many models as possible, and thus assumptions are formulated in a rather general manner involving upper and lower bounds on $\mathcal{B}_{ij}$. 

The main reason for this approach is to build a flexible strategy valid for a wide range of collision kernels suitable for interactions involving polyatomic gases.  Our motivation comes from the analysis of the Boltzmann equation modelling a single polyatomic gas \cite{MPC-IG-poly}, which successfully found its application in engineering modelling, as in \cite{MPC-Dj-T} for gas transport coefficients that match experimentally measured values, suggesting that the collision kernel model is appropriate. Of course, for interactions involving only monatomic molecules, our assumption reduces to  the frequently  used model of hard potentials. \\

For any pair of indices $i, j \in \left\{ 1, \dots, P \right\}$, the rate $\gamma_{ij}$ is supposed to have the following properties
\begin{equation}\label{gamma ij assumpt}
	\gamma_{ij} = \gamma_{ji},\quad	\gamma_{ij} \in [0,2] \ \ \text{and, additionally,} \ \ \forall \ i \in \left\{ 1, \dots, P \right\} \quad  \max_{1 \leq j \leq P } \gamma_{ij} =: \gi >0.  
\end{equation}
We also denote
\begin{equation}\label{gamma w}
\gw := \min_{1 \leq i \leq P } \gi   >0, \qquad \gm :=  \max_{1 \leq i \leq P } \gi   >0.
\end{equation} 

\subsection{Mono-mono interactions}\label{m-m cross ass}
The collision kernels  $\mathcal{B}_{ij}(v, v_*, \sigma)$, $i, j \in \left\{ 1, \dots, M \right\}$ introduced in \eqref{m-m cross} describing interactions between monatomic molecules are  assumed to take the following form 
\begin{equation}\label{m-m ass B}
	\mathcal{B}_{ij}(v, v_*, \sigma) = b_{ij}(\hat{u}\cdot \sigma)  \tilde{\mathcal{B}}_{ij}(v, v_*),
\end{equation}
where the angular part is assumed non-negative, integrable and symmetric with respect to the interchange $i\leftrightarrow j$, 
\begin{equation}\label{bij integrable}
b_{ij}(\hat{u}\cdot \sigma)   \in L^1(\bS), \quad b_{ij}(\hat{u}\cdot \sigma) = b_{ji}(\hat{u}\cdot \sigma) \geq 0, \quad u = v-v_*, \quad \hat{u}= \frac{u}{|u|},
\end{equation}
and $\tilde{\mathcal{B}}_{ij}(v, v_*)$ is the usual model of hard potentials (up to a multiplicative constant),
\begin{equation}\label{m-m ass B tilde}
	\tilde{\mathcal{B}}_{ij}(v, v_*)  
	= \left( \frac{\mu_{ij}}{2 \, \mP}  \right)^{\gamma_{ij}/2} |v-v_*|^{\gamma_{ij}},
\end{equation}
with $\gamma_{ij}$ from \eqref{gamma ij assumpt}.
\subsection{Poly-mono \& mono-poly interactions} \label{m-p cross ass}
For $i \in \left\{ M+1, \dots, P \right\}$ and $ j \in \left\{ 1, \dots, M \right\}$, the collision kernels  $	\mathcal{B}_{ij}(v, v_*, I, \sigma, R)$ defined by \eqref{p-m cross} are supposed to satisfy the following bounds
\begin{equation}\label{p-m ass B}
 b_{ij}(\hat{u}\cdot \sigma) \,  \tilde{b}_{ij}^{lb}(R) \,  \tilde{\mathcal{B}}_{ij}(v, v_*,I) \leq 	\mathcal{B}_{ij}(v, v_*, I, \sigma, R) \leq   b_{ij}(\hat{u}\cdot \sigma) \,  \tilde{b}_{ij}^{ub}(R) \,  \tilde{\mathcal{B}}_{ij}(v, v_*,I),
\end{equation}
where the angular part $ b_{ij}(\hat{u}\cdot \sigma)$ is assumed to be  as in \eqref{bij integrable}, and non-negative functions $ \tilde{b}_{ij}^{lb}(R)$,  $ \tilde{b}_{ij}^{ub}(R)$   are assumed to have the following integrability 
properties,
\begin{equation}\label{bij R integrable}
\tilde{b}_{ij}^{lb}(R), \ \tilde{b}_{ij}^{ub}(R)  \in L^1([0,1]; \dm\, \md R), 
\end{equation}
where $d_i(R) $ is from \eqref{psi m-p}.
The velocity-internal energy part is assumed to have the following form
\begin{equation}\label{p-m ass B tilde}
	\tilde{\mathcal{B}}_{ij}(v, v_*,I)  
	= \left( \frac{1}{\mP} \left( \frac{\mu_{ij}}{2} |v-v_*|^2 + I   \right) \right)^{\gamma_{ij}/2},
\end{equation}
with $\gamma_{ij}$ from \eqref{gamma ij assumpt}. 

\subsection{Poly-poly interactions} \label{p-p cross ass} 
For indices  $i, j \in \left\{ M+1, \dots, P \right\}$, which describe interactions between  polyatomic molecules  solely,
the collision kernels   $\mathcal{B}_{ij}(v, v_*, I, I_*, \sigma, r,  R)$ are assumed to satisfy the following bounds 
\begin{equation}\label{p-p ass B} 
	b_{ij}(\hat{u}\cdot \sigma) \,  \tilde{b}_{ij}^{lb}(r, R) \,  \tilde{\mathcal{B}}_{ij}(v, v_*,I, I_*) \leq 	\mathcal{B}_{ij}(v, v_*, I, I_*, \sigma, r,  R) \leq   b_{ij}(\hat{u}\cdot \sigma) \,  \tilde{b}_{ij}^{ub}(r, R) \,  \tilde{\mathcal{B}}_{ij}(v, v_*,I, I_*),
\end{equation}
where the angular part $b_{ij}(\hat{u}\cdot \sigma)$ is supposed  as in \eqref{bij integrable}, and non-negative  functions $\tilde{b}_{ij}^{lb}(r, R)$, $\tilde{b}_{ij}^{ub}(r, R)$ are assumed to have the following 
integrability properties
\begin{equation}\label{bij r R integrable}
	\tilde{b}_{ij}^{lb}(r, R), \ \tilde{b}_{ij}^{ub}(r, R)  \in L^1([0,1]^2; \,  \dpo \, \md r \, \md R),
\end{equation}
where the function $\dpo$ was introduced in \eqref{fun r R}. The velocity-internal energy part $\tilde{\mathcal{B}}_{ij}(v, v_*,I, I_*)$ takes the following form
\begin{equation}\label{p-p ass B tilde} 
 \tilde{\mathcal{B}}_{ij}(v, v_*,I, I_*) 
 = \left( \frac{1}{\mP} \left( \frac{\mu_{ij}}{2} |v-v_*|^2 + I + I_*  \right) \right)^{\gamma_{ij}/2},
\end{equation}
with $\g$ from \eqref{gamma ij assumpt}. 

\subsection{Remarks}\label{Remark tr prob}

Note that all three assumptions \eqref{m-m ass B tilde}, \eqref{p-m ass B tilde}  and \eqref{p-p ass B tilde} on the form of $\tilde{\mathcal{B}}_{ij}$ can be written concisely as
\begin{equation}\label{tilde Bij 3}
	\tilde{\mathcal{B}}_{ij}  = \left( \frac{E_{ij}}{\mP} \right)^{\gamma_{ij}/2},
\end{equation}
where the energy $E_{ij}$ is to be understood as \eqref{m-m en}, \eqref{p-m en} or \eqref{p-p en} depending on indices $i, j$. Moreover, such a form of $	\tilde{\mathcal{B}}_{ij}$ is micro-reversible itself, by \eqref{m-m coll CL 2}, \eqref{p-m V E} and \eqref{p-p V E}.

\smallskip
\noindent
Integrability properties \eqref{bij integrable} and \eqref{bij R integrable} led us to define constants
\begin{equation}\label{kappa p-m}
\left(	\begin{matrix}
		\kappa_{ij}^{lb} \\[5pt]
			\kappa_{ij}^{ub}
	\end{matrix} \right)
 = 	\int_{\bS \times[0,1]} b_{ij}(\hat{u}\cdot \sigma) \,   
\left( \begin{matrix}
 	\blR \\[5pt]
 \buR
 \end{matrix} \right)
\, \dm \, \md \sigma \, \md R, \quad i \in \left\{ M+1, \dots, P \right\},  j \in \left\{ 1, \dots, M \right\}\,,
\end{equation}
and by \eqref{bij r R integrable},
\begin{equation}\label{kappa p-p}
	\left(	\begin{matrix}
		\kappa_{ij}^{lb} \\[5pt]
		\kappa_{ij}^{ub}
	\end{matrix} \right)
	= 	\int_{\bS \times[0,1]^2} b_{ij}(\hat{u}\cdot \sigma) \,   
	\left( \begin{matrix}
		\blRr \\[5pt]
		\buRr
	\end{matrix} \right)
\dpo \, \md \sigma \, \md r \, \md R, \quad i, j \in \left\{ M+1, \dots, P \right\}.
\end{equation}
For convenience, we also introduce the constant  for monatomic interaction which actually reduces to the $L^1$ norm od the angular part, i.e.
\begin{equation}\label{kappa norm b}
	\kappa_{ij}^{lb} = 	\kappa_{ij}^{ub} = \nb, \quad \text{when} \  i, j \in \left\{ 1, \dots, M \right\}.
\end{equation}

Note that discrepancy in constants 	$\kappa_{ij}^{lb}$  and 	$\kappa_{ij}^{ub}$   is due to estimates on the parts of the collision kernels that is concerned with parameters  $R$ and $r, R$ for mono-poly \& poly-mono and poly-poly interactions, respectively. This difference disappears when, for instance, $\tilde{b}_{ij}^{lb} = \tilde{b}_{ij}^{ub}  =1$ i.e. for the choice $\mathcal{B}_{ij}=b_{ij}(\hat{u}\cdot\sigma) \, \tilde{\mathcal{B}}_{ij}$, which is a possible choice due to  micro-reversibility properties of  $\tilde{\mathcal{B}}_{ij}$ defined by \eqref{tilde Bij 3} .

\subsection{An example of the collision kernel } 
Besides the model $\mathcal{B}_{ij}=b_{ij}(\hat{u}\cdot\sigma) \, \tilde{\mathcal{B}}_{ij}$, with $ \tilde{\mathcal{B}}_{ij}$ from \eqref{tilde Bij 3}, mentioned in the above Remark \ref{Remark tr prob}, which for a single species polyatomic gas  corresponds  to the Model 1 in \cite{MPC-IG-poly}, we bring an another example corresponding to   Models 2 and 3 in \cite{MPC-IG-poly}, showed  to be successful in providing physical intuition of the single Boltzmann model for a polyatomic gas \cite{MPC-Dj-S, MPC-Dj-T}.  For $i \in \left\{ M+1, \dots, P \right\}$ and $ j \in \left\{ 1, \dots, M \right\}$, consider 
\begin{equation*}
	\mathcal{B}_{ij}(v, v_*, I, \sigma, R) 
	\\ = b_{ij}(\hat{u}\cdot\sigma) \left( R^{\gamma_{ij}/2}|v-v_*|^{\gamma_{ij}} + \left(\frac{(1-R) I}{\mP}\right)^{\gamma_{ij}/2}  \right),
\end{equation*}
while for  $i, j \in \left\{ M+1, \dots, P \right\}$ take
\begin{multline*}
	\mathcal{B}_{ij}(v, v_*, I, I_*, \sigma, r, R) 
	\\ = b_{ij}(\hat{u}\cdot\sigma) \left( R^{\gamma_{ij}/2}|v-v_*|^{\gamma_{ij}} + \left(\frac{r(1-R) I}{\mP}\right)^{\gamma_{ij}/2} + \left(\frac{(1-r)(1-R) I_*}{\mP}\right)^{\gamma_{ij}/2} \right).
\end{multline*}
Then, assumptions  \eqref{p-m ass B},  \eqref{p-p ass B}  are satisfied, for example,   by choosing
\begin{equation*}
	b_{ij}^{lb}(R) = \min\left\{  \frac{2\, \mP}{\mu_{ij}} R , 1-R \right\}^{\g/2}, \qquad 	b_{ij}^{ub}(R)  = 2^{1-\g/2}\max\left\{  \frac{2\, \mP}{\mu_{ij}} R , 1-R \right\}^{\g/2},
\end{equation*}
for  $i \in \left\{ M+1, \dots, P \right\}$ and $ j \in \left\{ 1, \dots, M \right\}$, and
	\begin{equation*}
		\begin{split}
	b_{ij}^{lb}(r, R) & = \min\left\{  \frac{2\, \mP}{\mu_{ij}} R, r(1-R), (1-r)(1-R) \right\}^{\g/2}, \\ 	
	b_{ij}^{ub}(r, R) &= 3^{1-\g/2} \max\left\{  \frac{2\, \mP}{\mu_{ij}} R, r(1-R), (1-r)(1-R) \right\}^{\g/2},
\end{split}
\end{equation*}
for $i, j \in \left\{ M+1, \dots, P \right\}$.

\section{Estimates on the collision operator}\label{Sec: Est coll op}

The first step in the well-posedness proof is to show dissipative character of the gain operator's $k$-th moment  reached by averaging post-collisional velocities or velocity-energy pairs. The original  technique was introduced by Bobylev in \cite{Bob-Moment-ineq}, which uses decomposition of the post-collisional velocities in the center-of-mass framework and relies on symmetries built in the model of a single monatomic gas.   This idea was lately used for different frameworks, as for inelastic collisions \cite{IG-Bob-Panf-inelastic}, granular gases \cite{Alonso-Lods-Granular}, more general cross sections \cite{IG-Panf-Vill, Lu-Mouhot-1, IG-Alonso-Task-Pavl, MPC-MT-Kac}. Recently, the authors developed an averaging tool for monatomic gas mixtures \cite{MPC-IG-ARMA, Erica-pspde, Alonso-Orf} and single polyatomic gases \cite{MPC-IG-poly}, which relies on the representation of post-collisional velocities (and internal energies for polyatomic interactions) in a convex combination form of the energies in the center-of-mass framework. 

Exploiting this idea, we will first represent post-collisional quantities in a suitable convex combination form   in the upcoming Section \ref{Section: en id}, which will be the basis for the averaging over the space of parameters in Section \ref{Section: Av lemma}, crucial to show dissipation  of the gain operator. 

\subsection{Energy Identities and Estimates}\label{Section: en id}

We first introduce notation that will be used in this section. The parameter $s_{ij} \in (0,1)$ convexly splits the sum of masses $m_i + m_j$,
\begin{equation}\label{sij}
	s_{ij} = \frac{m_i}{m_i + m_j} \quad \Rightarrow \quad 1- s_{ij} = \frac{m_j}{m_i + m_j} = s_{ji}.
\end{equation}
Its minimum will play an important role in computations,
\begin{equation}\label{sij min}
	\bar{s}_{ij} = \min\left\{ s_{ij}, s_{ji}\right\} \in (0,1).
\end{equation}

For the two colliding molecules, the total energy (kinetic or kinetic+internal) during collision, which is a conserved quantity, will be written in brackets form of \eqref{brackets}, 
\begin{equation}\label{E tot gen}
	\Etot_{ij} = \la \cdot \ra_i^2  + \la \cdot \ra_j^2 = \la \cdot' \ra_i^2  + \la \cdot' \ra_j^2,
\end{equation}
where the argument in the brackets can be either velocity or velocity-internal energy pair depending on whether $i$ or $j$ belong to $\left\{ 1, \dots, M \right\}$ or $\left\{ M+ 1, \dots, P \right\}$. This total energy in brackets form can be written in the center-of-mass reference framework \eqref{V-u}, 
\begin{equation}\label{en tot}
	\Etot_{ij} = 2 + \frac{m_i + m_j}{2 \,  \mP } |V|^2 + \frac{E_{ij}}{ \mP},
\end{equation}
where  $E_{ij}$ is \eqref{m-m en} for $i, j \in \left\{ 1, \dots, M \right\}$, \eqref{p-p en} for $i, j \in \left\{ M+ 1, \dots, P \right\}$ and \eqref{p-m en} if $i \in \left\{ M+ 1, \dots, P \right\}$ and $j \in \left\{ 1, \dots, M \right\}$. 

The goal of the upcoming lemma  is to express the total energy \eqref{en tot} in convex components in order to represent each primed bracket $\la \cdot' \ra_i^2$ and $ \la \cdot' \ra_j^2$ separately in terms of non-primed quantities.

\begin{lemma}[Energy Identity Lemma]\label{Lemma: En Iden}
The following energy identities and estimates hold, depending on the nature of particle interactions,
\begin{description}
\item[\textit{Part (i) - mono-mono interactions}] Let $i, j \in \left\{ 1, \dots, M \right\}$ and the primed velocities be defined in \eqref{m-m primed velocity}.  Then, there exist non-negative functions $p_{ij}$, $q_{ij}$ and $\lambda_{ij}$  which depend only on velocities $v$, $v_*$ and parameter $s_{ij}$ from \eqref{sij} and satisfy $p_{ij} + q_{ij} = 1 $, such that the following representation holds,
	\begin{equation}\label{m-m en id lemma}
		\la v' \ra_i^2 = \Etot_{ij} \,  p_{ij} + \lambda_{ij} \Vs, \qquad 	\la v'_* \ra_j^2 = \Etot_{ij} \,  q_{ij} - \lambda_{ij} \Vs, 
	\end{equation}
and the following estimate 
\begin{equation}\label{m-m en id lemma est}
\la v' \ra_i^2,  \	\la v'_* \ra_j^2  \leq   \left( 1 - \bs  \left(  1 - |\Vs| \right)  \right)  \Etot_{ij},
\end{equation}
where $\Etot_{ij} =  \la v \ra_i^2  + \la v_* \ra_j^2 $ and $\bs$ is given in \eqref{sij min}.
\item[\textit{Part (ii) - poly-poly interactions}]  Let $i,j \in \left\{M+1, \dots, P \right\}$ and the primed velocity - internal energy pairs be defined  in \eqref{p-p coll rules}. There exist non-negative functions $\tilde{p}_{ij}$, $\tilde{q}_{ij}$,	$\tilde{t}_{ij}$ and $\lambda_{ij}$ which depend on velocities $v$, $v_*$, internal energies $I$, $I_*$, energy exchange  variable $R$ and mass ratio $s_{ij}$ from \eqref{sij} such that the following convexity property holds $	\tilde{p}_{ij} + 	\tilde{q}_{ij} + 	\tilde{t}_{ij} =1$ and representation of the primed velocity - internal energy pairs,
\begin{equation}\label{p-p en id lemma}
	\la v', I' \ra_i^2 = \Etot_{ij} \left( \tilde{p}_{ij} + r \, \tilde{t}_{ij}  \right) + \lambda_{ij} \Vs,   \qquad 	\la v'_*, I'_* \ra_j^2 = \Etot_{ij} \left( \tilde{q}_{ij} + (1-r) \tilde{t}_{ij}  \right) - \lambda_{ij} \Vs,
\end{equation}
where $\Etot_{ij} = \la v, I \ra_i^2  + \la v_*, I_* \ra_j^2$. Moreover, the following estimate holds
\begin{equation}\label{p-p en id lemma est}
	\begin{split}
		\la v', I' \ra_i^2 & \leq \Etot_{ij} \left( 1- \tilde{q}_{ij} (1-|\Vs|) - \tilde{t}_{ij} (1-r) \right),  \\
		\la v'_*, I'_* \ra_j^2 & \leq   \Etot_{ij} \left( 1- \tilde{p}_{ij} (1-|\Vs|) - \tilde{t}_{ij} r \right),
	\end{split}
\end{equation}
where the involved terms satisfy
\begin{equation}\label{p-p en id lemma est 2}
	( \tilde{p}_{ij} + \tilde{t}_{ij} ), \ 	( \tilde{q}_{ij} + \tilde{t}_{ij} )  \geq \bs,
\end{equation}
with $\bs$ from \eqref{sij min}.
\item[\textit{ Part (iii) - poly-mono \& mono-poly interactions }]  Consider $i\in \left\{ M+1, \dots, P \right\}$ and $   j \in \left\{ 1,  \allowbreak \dots,  \allowbreak M \right\}$.  Let the primed velocity-internal energy pair $(v', I')$ and the primed velocity  $v'_*$ be defined as in \eqref{p-m coll rules}. There exist non-negative functions  $\tilde{p}_{ij}$, $\tilde{q}_{ij}$,	$\tilde{t}_{ij}$ and $\lambda_{ij}$ depending on $v$, $v_*$, $I$, $R$ and $s_{ij}$ from \eqref{sij} such that it holds the convexity property $	\tilde{p}_{ij} + 	\tilde{q}_{ij} + 	\tilde{t}_{ij} =1$ and the representation 
\begin{equation}\label{p-m en id lemma}
	\la v', I' \ra_i^2 =  \left( \tilde{p}_{ij} +  \tilde{t}_{ij} \right)  \Etot_{ij} + \lambda_{ij} \Vs ,  \qquad 	\la v'_* \ra_j^2 =   \tilde{q}_{ij} \, \Etot_{ij}  - \lambda_{ij} \Vs,
\end{equation}
with $\Etot_{ij} = \la v, I \ra_i^2  + \la v_* \ra_j^2$. Moreover, the following estimate holds
\begin{equation}\label{p-m en id lemma est}
	\la v', I' \ra_i^2, \  \la v'_* \ra_j^2  \leq   \left( 1 - \bs R \left(  1 - |\Vs| \right)  \right)  \Etot_{ij},
\end{equation}
with $\bs$ from \eqref{sij min}.
\end{description}

\end{lemma}

\begin{proof}
First define function $\Theta_{ij}$ which depends on velocities or velocity-internal energy pairs solely, 
\begin{equation}
	\Theta_{ij} \Etot_{ij} = 1 +  \frac{m_i + m_j}{2 \, \mP } |V|^2 \quad \Rightarrow \quad \left( 1 -	\Theta_{ij}  \right)\Etot_{ij} = 1 + \frac{1}{ \mP} E_{ij}.
\end{equation}
The part $ \left( 1 -	\Theta_{ij}  \right)\Etot_{ij} $ is further split depending on the type of molecular interaction. 
\subsubsection*{Mono-mono interactions} This type of interactions was already studied in \cite{MPC-IG-ARMA}, Lemma 4.1. Defining 
\begin{equation}\label{m-m split}
	\begin{split}
		{p}_{ij} &= s_{ij} \Theta_{ij} + (1-s_{ij})  \left( 1 -	\Theta_{ij}  \right), \\
		{q}_{ij} &= (1-s_{ij}) \Theta_{ij} + s_{ij}   \left( 1 -	\Theta_{ij}  \right), \\	
		\lambda_{ij}  &= 2 \sqrt{s_{ij}(1-s_{ij})}  \sqrt{\Theta_{ij}  \Etot_{ij} -1  } \sqrt{ \left( 1 -	\Theta_{ij}  \right) \Etot_{ij} -1},
	\end{split}
\end{equation}
the representation \eqref{m-m en id lemma} and estimate \eqref{m-m en id lemma est} hold.

\subsubsection*{Poly-poly interactions} For $i,j \in \left\{M+1, \dots, P \right\}$, in order to split the term $\left( 1 -	\Theta_{ij} \ \right)\Etot_{ij}$, we introduce function $\Sigma_{ij}$ that, except on velocities and internal energies like for $\Theta_{ij} $, additionally depends on the  parameter $R$, 
\begin{equation}
	\Sigma_{ij} \left( 1 -	\Theta_{ij}  \right) \Etot_{ij} = 1 + \frac{1}{ \mP} R E_{ij} \quad \Rightarrow \quad  \left( 1 -	\Sigma_{ij}  \right) \left( 1 -	\Theta_{ij}  \right)\Etot_{ij} = \frac{1}{ \mP} (1-R) E_{ij}.
\end{equation}
We then define the following functions in terms of convex combination functions,
\begin{equation}\label{p-p split}
	\begin{split}
		\tilde{p}_{ij} &= s_{ij} \Theta_{ij} + (1-s_{ij}) \Sigma_{ij} \left( 1 -	\Theta_{ij}  \right), \\
		\tilde{q}_{ij} &= (1-s_{ij}) \Theta_{ij} + s_{ij} \Sigma_{ij} \left( 1 -	\Theta_{ij}  \right), \\	
		\tilde{t}_{ij} &= \left( 1 -	\Sigma_{ij}  \right) \left( 1 -	\Theta_{ij}  \right), \\
		\lambda_{ij}  &= 2 \sqrt{s_{ij}(1-s_{ij})} \sqrt{\Sigma_{ij} \left( 1 -	\Theta_{ij}  \right) \Etot_{ij} -1  }\sqrt{\Theta_{ij}  \Etot_{ij} -1}.
	\end{split}
\end{equation}
Taking the square of the primed velocities and internal energies from \eqref{p-p coll rules} and combining them into the bracket form \eqref{brackets},  for the  term $	\la v', I' \ra_i^2 $ it follows
\begin{equation*}
	\begin{split}
	\la v', I' \ra_i^2 & = \frac{m_i}{m_i + m_j} \left( 1 + \frac{m_i + m_j }{2 \mP} |V|^2 \right) + \frac{m_j}{m_i + m_j} \left( 1 + \frac{R E_{ij}}{ \mP} \right) +  \frac{r (1-R) E_{ij}}{ \mP} \\ &+ \frac{m_i m_j}{(m_i + m_j) \mP} \sqrt{\frac{2 R E_{ij}}{\mu_{ij}}} |V| \Vs,
	\end{split}
 \end{equation*}
and similarly for the counterpart $	\la v'_*, I'_* \ra_i^2$. Expressing 
\begin{equation*}
	\begin{split}
 &\sqrt{\frac{(m_i+m_j)}{2 \mP}} |V| = \sqrt{\Theta_{ij} \Etot_{ij} - 1}, \qquad \sqrt{\frac{\mu_{ij}}{m_i+m_j}} = \sqrt{s_{ij}(1-s_{ij})}, \\ 
 &\sqrt{\frac{R E_{ij}}{\mP}} = \sqrt{\Sigma_{ij} \left( 1 -	\Theta_{ij}  \right) \Etot_{ij} -1},
 \end{split}
\end{equation*} 
 the representation \eqref{p-p en id lemma} is obtained. Note that for a single polyatomic gas corresponding to the case $s_{ij}=1/2$, the representation  \eqref{p-p en id lemma} coincides with the one introduced in \cite{MPC-IG-poly}. 

To prove the second part, note that  by Young's inequality the following estimates on $\lambda_{ij}$ hold
\begin{equation}\label{lambda est}
	 \lambda_{ij}  \leq \tilde{p}_{ij} \Etot_{ij}, \qquad 	 \lambda_{ij}  \leq \tilde{q}_{ij}\Etot_{ij}.
\end{equation}
Then, from  \eqref{p-p en id lemma} the following estimates are straightforward, 
\begin{equation*}
	\begin{split}
		\la v', I' \ra_i^2 & \leq \Etot_{ij} \left( \tilde{p}_{ij} + r \, \tilde{t}_{ij} + \tilde{q}_{ij} |\Vs| \right) = \Etot_{ij} \left( 1- \tilde{q}_{ij} (1-|\Vs|) - \tilde{t}_{ij} (1-r) \right),  \\ 	
		\la v'_*, I'_* \ra_j^2 & \leq \Etot_{ij} \left( \tilde{q}_{ij} + (1-r) \tilde{t}_{ij} + \tilde{p}_{ij} |\Vs| \right) = \Etot_{ij} \left( 1- \tilde{p}_{ij} (1-|\Vs|) - \tilde{t}_{ij} r \right),
		\end{split}
\end{equation*}
yielding \eqref{p-p en id lemma est}. Estimate \eqref{p-p en id lemma est 2} follows from
\begin{equation*}
	( \tilde{p}_{ij} + \tilde{t}_{ij} ), \ 	( \tilde{q}_{ij} + \tilde{t}_{ij} )  \geq \min\left\{ s_{ij}, 1- s_{ij}, 1 \right\} \left(   \Theta_{ij} +  \Sigma_{ij} \left( 1 -	\Theta_{ij}  \right)  +  \left( 1 -	\Sigma_{ij}  \right) \left( 1 -	\Theta_{ij}  \right) \right) = \bs,
\end{equation*}
by \eqref{sij min}.
\subsubsection*{Poly-mono \& mono-poly interactions } For $i\in \left\{ M+1, \dots, P \right\}$ and $j\in \left\{ 1, \dots, M \right\}$,  functions $\tilde{p}_{ij}$, $\tilde{q}_{ij}$ and $\tilde{t}_{ij}$ are of the same form as \eqref{p-p split} for poly-poly interactions, except that the total energy in the center-of-mass framework $E_{ij}$ is understood as \eqref{p-m en}.

For the second part, we use the estimate on $\lambda_{ij}$ as in \eqref{lambda est}. Then from the representation \eqref{p-m en id lemma}, the following estimates hold, 
\begin{equation*}
	\begin{split}
	\la v', I' \ra_i^2 & \leq \Etot_{ij} \left(  \tilde{p}_{ij}  + \tilde{t}_{ij} + \tilde{q}_{ij} \, | \Vs | \right)  =\Etot_{ij} \left(  1 -  \tilde{q}_{ij} \left(  1 - | \Vs |  \right) \right), \\
		\la v'_*  \ra_j^2 & \leq \Etot_{ij} \left(  \tilde{q}_{ij}    + \tilde{p}_{ij} \, | \Vs | \right)  \leq   \Etot_{ij} \left(  1 -  \tilde{p}_{ij} \left(  1 - | \Vs |  \right) \right)\,.
	\end{split}
\end{equation*}
Note that
\begin{equation*}
	\tilde{p}_{ij}, 	\tilde{q}_{ij}  \geq \bs \left(  \Theta_{ij} +   \Sigma_{ij} \left( 1 -	\Theta_{ij}  \right) \right) = \bs \left(   1 -	\tilde{t}_{ij}  \right) \geq \bs \, R,
\end{equation*}
where the last inequality is due to
\begin{equation*}
	1 -	\tilde{t}_{ij}   = \frac{1}{\Etot_{ij}} \left( \Etot_{ij} - (1-R) \frac{E_{ij}}{\mM} \right) \geq R.
\end{equation*}
This implies \eqref{p-m en id lemma est}, which concludes the proof.
\end{proof}

\subsection{Compact Manifold  Averaging Lemma}\label{Section: Av lemma}
The next goal is to show that estimates on the primed quantities $\la \cdot' \ra_i^2$, $ \la \cdot' \ra_j^2$, proved in the previous lemma yield    decay properties of their $k-$th power $\la \cdot' \ra_i^k$, $ \la \cdot' \ra_j^k$, with respect to $k$ when averaged over a suitable compact domain, such as angular transitions and partition functions,  describing transition probability rates depending on the particles' interaction nature.  

To this end, we prove the following key lemma, highlighting  a novel  method of proof  flexible enough to conveniently adapt to all type of interactions satisfying the general pairwise energy transfer  identity structure of Lemma \ref{Lemma: En Iden}.

\begin{lemma}[Compact Manifold Averaging Lemma]\label{Lemma: Averaging} 
	With the notation of  Lemma \ref{Lemma: En Iden} and the assumptions on collision kernels  stated in Section \ref{Sec: Assump Bij}, there exist non-negative constants $\Cpov$ decreasing in $k\geq0$ and with $\lim_{k\rightarrow\infty}\Cpov=0$ or more precisely 
	\begin{equation}\label{C pov}
\Cpov = \smallO(k^{-\frac a2}), \quad \text{for all} \ \ a\in(0,1),
\end{equation}
	 such that  the following  estimates hold, depending on the nature of particle interactions. 
	\begin{description}
		\item[\textit{Part (i) - mono-mono interactions}] For $i, j \in \left\{ 1, \dots, M \right\}$,
		\begin{equation}\label{m-m av lemma}
		\int_{\bS}  \left( 1 - \bs  \left(  1 - |\Vs| \right)  \right)^k b_{ij}(\hat{u}\cdot\sigma) \, \md \sigma \leq \Cpov.
		\end{equation}
		\item[\textit{Part (ii) - poly-poly interactions}]  When $i,j \in \left\{M+1, \dots, P \right\}$,
		\begin{equation}\label{p-p av lemma}
			\begin{split}
			& \int_{\bS \times[0,1]^2}\left( 1- \tilde{q}_{ij} (1-|\Vs|) - \tilde{t}_{ij} (1-r) \right)^k b_{ij}(\hat{u}\cdot \sigma)  \, \buRr \,  \dpo \, \md \sigma \, \md r \, \md R \leq \Cpov\,, \\
			&\text{and}\\
			& \int_{\bS \times[0,1]^2} \left( 1- \tilde{p}_{ij} (1-|\Vs|) - \tilde{t}_{ij} r \right)^k b_{ij}(\hat{u}\cdot \sigma) \,  \buRr \,  \dpo \, \md \sigma \, \md r \, \md R \leq \Cpov\,. 
			\end{split}
		\end{equation}
		\item[\textit{Part (iii) - poly-mono \& mono-poly interactions}]  Consider $i\in \left\{ M+1, \dots, P \right\}$ and $   j \in \left\{ 1,  \allowbreak \dots,  \allowbreak M \right\}$.  Then, 
		\begin{equation}\label{p-m av lemma}
		\int_{\bS \times[0,1]}  \left( 1 - \bs R \left(  1 - |\Vs| \right)  \right) ^k b_{ij}(\hat{u}\cdot \sigma) \, \buR \, \dm \, \md \sigma \, \md R \leq \Cpov\,.
		\end{equation}
	\end{description}
As a consequence of the estimates \eqref{m-m av lemma},\eqref{p-p av lemma}, and \eqref{p-m av lemma}, there exists $\ks$, depending only on the angular part $b_{ij}(\hat{u}\cdot\sigma)$ and functions $b_{ij}^{ub}$ of  energy exchange variables, such that for all $i,j\in\left\{1, \dots, P \right\}$,
\begin{equation}\label{ks povzner}
	 \mathcal{C}_k^{ij} < \frac{ \kappa_{ij}^{lb} }{2}, \qquad \text{for} \;\; k \geq \ks,
\end{equation}	
with $\kappa_{ij}^{lb}$ provided in \eqref{kappa norm b}, \eqref{kappa p-p} and \eqref{kappa p-m}, respectively for each type of particles' interactions.
\end{lemma}

\begin{proof} We prove each type of interaction separately. The idea of the proof is to split the domain of integration into sub-domains, one sub-domain, $A_\eps^0$, on which the term raised on power $k$ is strictly less than 1, guarantying the power decay in $k$, and its complement  whose measure will be  $\smallO(\eps)$. A suitable choice  of $\eps$ in terms of $k$ will allow to conclude the proof. 
	
\subsubsection*{Mono-mono interactions} Split the sphere $\bS$ into two sub-regions: 
\begin{equation*}
		A_\eps^0 = \left\{ \sigma \in \bS: |\Vs| \leq 1-\eps \right\}, \qquad
		A_\eps^1 = \left\{ \sigma \in \bS: |\Vs| >  1-\eps   \right\}, \quad  \eps>0.
\end{equation*}
In $A_\eps^0 $ the following inequality holds
\begin{equation*}
	1 - \bs  \left(  1 - |\Vs| \right) \leq  1 - \bs \eps  < 1, 
\end{equation*}
since both $\bs, \eps > 0$. Therefore, the averaging over this domain will ensure the power decay in $k$, i.e.  the left-hand side of \eqref{m-m av lemma} is estimated  as
\begin{equation}\label{m-m av lemma:e1}
\int_{\bS}  \left( 1 - \bs  \left(  1 - |\Vs| \right)  \right)^k b_{ij}(\hat{u}\cdot\sigma) \, \md \sigma \leq   \left( 1 - \bs \eps \right)^k \nb +  \int_{\bS}     \mathds{1}_{\!A_\eps^1 }\,   b_{ij}(\hat{u}\cdot \sigma)  \, \md \sigma. 
\end{equation}
It remains to show that the last term is of order $\smallO(\eps)$. To that end, consider the family of measurable functions $\Phi_\eps(\hat{V}, \hat{u}): \bS \times \bS \rightarrow \bRp$, for $\eps>0$ defined as 
\begin{equation}\label{m-m av lemma:e2}
0 \leq \Phi_\eps(\hat{V}, \hat{u}) :=  \int_{\bS} \mathds{1}_{| \Vs | >  1-\eps } \, b_{ij}(\hat{u}\cdot \sigma) \, \md \sigma\,.
\end{equation}
Let us prove that $\Phi_\eps(\hat{V}, \hat{u})$ converges uniformly (in the variables $\hat{V}, \hat{u}$) to zero as $\eps\rightarrow 0$.  To this end, fix $\delta>0$ and note that for any $K>0$
\begin{align*}
\Phi_\eps(\hat{V}, \hat{u}) &=  \int_{\bS} \mathds{1}_{|\Vs| >  1-\eps } \, b_{ij}(\hat{u}\cdot \sigma)\,\mathds{1}_{b_{ij}(\hat{u}\cdot \sigma) > K } \, \md \sigma + \int_{\bS} \mathds{1}_{ |\Vs| >  1-\eps } \, b_{ij}(\hat{u}\cdot \sigma)\,\mathds{1}_{b_{ij}(\hat{u}\cdot \sigma) \leq K } \, \md \sigma \\
&=:  \Phi^1_\eps(\hat{V}, \hat{u}) + \Phi^2_\eps(\hat{V}, \hat{u})\,.
\end{align*}
For $\Phi^1_\eps(\hat{V}, \hat{u})$ we use polar coordinates setting $\cos(\theta)=\hat{u}\cdot\sigma$ with $\hat{u}$ arbitrary but fixed.  Then, due to the monotone convergence theorem, there exists a sufficiently large $K:=K(\delta,b_{ij})$ such that
\begin{equation*}
0 \leq \Phi^1_\eps(\hat{V}, \hat{u}) \leq  |\mathbb{S}^{d-2}|\int^{1}_{-1} b_{ij}(\cos(\theta))\,\mathds{1}_{b_{ij}(\cos(\theta)) > K } \,\sin^{d-2}(\theta) \md \theta \leq \frac{\delta}{2}\,.
\end{equation*}
As for $\Phi^2_\eps(\hat{V}, \hat{u})$,
\begin{equation*}
0\leq\Phi^2_\eps(\hat{V}, \hat{u})\leq  K\int_{\bS} \mathds{1}_{|\Vs| >  1-\eps } \md \sigma = |\mathbb{S}^{d-2}|\,K\,\int^{1}_{-1}\mathds{1}_{|s| >  1-\eps } \md s \leq 2\,|\mathbb{S}^{d-2}|\,K\,\eps\,.
\end{equation*}
Thus, choosing $\eps<\frac{\delta}{4|\mathbb{S}^{d-2}|K}=:\eps_*(\delta,b_{ij})$ it holds that $\Phi^2_\eps(\hat{V}, \hat{u})\leq\frac{\delta}{2}$.  Consequently,
\begin{equation*}
0\leq \sup_{\hat{V}, \hat{u}}\Phi_\eps(\hat{V}, \hat{u}) \leq \delta\qquad \text{for any}\;\eps<\eps_*(\delta,b_{ij})\,.
\end{equation*}
In the sequel we simply write that $\sup_{\hat{V}, \hat{u}}\Phi_\eps(\hat{V}, \hat{u})=\smallO(\eps)$.  Returning to \eqref{m-m av lemma:e1}, we conclude that
\begin{equation}\label{m-m av lemma:e3}
\int_{\bS}  \left( 1 - \bs  \left(  1 - |\Vs| \right)  \right)^k b_{ij}(\hat{u}\cdot\sigma) \, \md \sigma \leq   \left( 1 - \bs \eps \right)^k \nb +  \smallO(\eps). 
\end{equation}
With the choice $\eps=k^{-a}$, with $a\in(0,1)$, it holds that
\begin{equation*}
\lim_{k\rightarrow\infty}\left( 1 - \bs\, \eps \right)^k \sim  \lim_{k\rightarrow\infty} e^{-\bs k^{1-a}} = 0\,,
\end{equation*}
and therefore,
\begin{equation}\label{m-m av lemma:e4}
\sup_{\hat{V}, \hat{u}}\int_{\bS}  \left( 1 - \bs  \left(  1 - |\Vs| \right)  \right)^k b_{ij}(\hat{u}\cdot\sigma) \, \md \sigma = \smallO\Big(k^{-a}\Big),\qquad a\in(0,1).
\end{equation}

\subsubsection*{Poly-poly interactions} For the first estimate in \eqref{p-p av lemma}, split the region $\bS \times[0,1]$ into three sub-regions: 
\begin{equation*}
	\begin{split}
		A_\eps^0 &= \left\{ (\sigma, r) \in \bS \times[0,1]: |\Vs| \leq 1-\eps,  \ r \leq  1- \eps  \right\},\\
		A_\eps^1 &= \left\{ (\sigma, r) \in \bS \times[0,1]: |\Vs| >  1-\eps   \right\}, \\
		A_\eps^2 &= \left\{ (\sigma, r) \in \bS \times[0,1]:   |\Vs| \leq 1-\eps, \ r >  1- \eps  \right\}.
	\end{split}
\end{equation*}
In $A_\eps^0$, thanks to \eqref{p-p en id lemma est 2},
\begin{equation*}
	\left( 1- \tilde{q}_{ij} (1-|\Vs|) - \tilde{t}_{ij} (1-r) \right) \leq  	\left( 1- \eps \left( \tilde{q}_{ij} + \tilde{t}_{ij} \right) \right) \leq 1- \bs \eps.
\end{equation*}
Denoting 
\begin{equation*}
	\rho_{ij}^{ub} = \int_{[0,1]^2}  \buRr \,  \dpo  \, \md r \, \md R,
\end{equation*}
the left-hand side of \eqref{p-p av lemma} becomes
\begin{multline*}
		 \int_{\bS \times[0,1]^2}\left( 1- \tilde{q}_{ij} (1-|\Vs|) - \tilde{t}_{ij} (1-r) \right)^k b_{ij}(\hat{u}\cdot \sigma) \,  \buRr \,  \dpo \, \md \sigma \, \md r \, \md R   \\
		\leq \left( 1- \bs \eps \right)^k \kappa_{ij}^{ub} + \rho_{ij}^{ub}  \Phi_\eps(\hat{V}, \hat{u}) 
		+ \nb \int_{[0,1]^2} \mathds{1}_{r >  1-\eps }  \,  \buRr \,  \dpo \, \md r \, \md R \\
		\leq  \left( 1- \bs \,\eps \right)^k \kappa_{ij}^{ub}  +  \smallO(\eps) \,,
\end{multline*} 
where for the latter integral we invoked the monotone convergence theorem.  The estimate follows, as in the previous case, choosing $\eps=k^{-a}$, with $a\in(0,1)$.

\smallskip
\noindent
For the second estimate in \eqref{p-p av lemma} we proceed similarly by considering the regions 
\begin{equation*}
	\begin{split}
		A_\eps^0 &= \left\{ (\sigma, r) \in \bS \times[0,1]: |\Vs| \leq 1-\eps,  \ r \geq  \eps  \right\},\\
		A_\eps^1 &= \left\{ (\sigma, r) \in \bS \times[0,1]: |\Vs| >  1-\eps   \right\},\\
		A_\eps^2 &= \left\{ (\sigma, r) \in \bS \times[0,1]: |\Vs| \leq 1-\eps,  \  r <   \eps  \right\}.
	\end{split}
\end{equation*}
Then, for $A_\eps^0$ the estimate \eqref{p-p en id lemma est 2} yields
\begin{equation*}
		 \left( 1- \tilde{p}_{ij} (1-|\Vs|) - \tilde{t}_{ij} r \right) \leq  	\left( 1- \eps \left( \tilde{p}_{ij} + \tilde{t}_{ij} \right) \right) \leq 1- \bs \eps,
\end{equation*}
and therefore
	\begin{multline*}
		 \int_{\bS \times[0,1]^2} \left( 1- \tilde{p}_{ij} (1-|\Vs|) - \tilde{t}_{ij} r \right)^k b_{ij}(\hat{u}\cdot \sigma) \,  \buRr \,  \dpo \, \md \sigma \, \md r \, \md R \\
		 	\leq \left( 1- \bs \eps \right)^k \kappa_{ij}^{ub} + \rho_{ij}^{ub} \Phi_\eps(\hat{V}, \hat{u}) 
		 + \nb \int_{[0,1]^2} \mathds{1}_{r <  \eps }  \,  \buRr \,  \dpo \, \md r \, \md R \\
		 \leq \left( 1- \bs \,\eps \right)^k \kappa_{ij}^{ub}  +  \smallO(\eps)\,. 
\end{multline*}
To conclude choose $\eps=k^{-a}$, with $a\in(0,1)$.

\subsubsection*{Poly-mono \& mono-poly interactions }  Split the region $\bS \times[0,1]$ into three sub-regions: 
\begin{equation*}
\begin{split}
A_\eps^0 &= \left\{ (\sigma, R) \in \bS \times[0,1]: |\Vs| \leq 1-\eps, \ R\geq \eps  \right\},\\
A_\eps^1 &= \left\{ (\sigma, R) \in \bS \times[0,1]: |\Vs| >  1-\eps   \right\},\\
A_\eps^2 &= \left\{ (\sigma, R) \in \bS \times[0,1]: |\Vs| \leq 1-\eps, \ R < \eps  \right\}.
\end{split}
\end{equation*}
Since in $A_\eps^0$  the following bound holds,
\begin{equation*}
	1 - \bs R \left(  1 - |\Vs| \right)  \leq 1 - \bs \, \eps^2,
\end{equation*}
the left-hand side of \eqref{p-m av lemma} can be estimated as 
\begin{multline*}
	\int_{\bS \times [0,1]}  \left( 1 - \bs R \left(  1 - |\Vs| \right)  \right)^k b_{ij}(\hat{u}\cdot \sigma) \, \buR \, \dm \, \md \sigma \, \md R \leq  \left( 1 - \bs \, \eps^2 \right)^k \kappa_{ij}^{ub} \\
	+ \ru \,\Phi_\eps(\hat{V}, \hat{u})  + \nb  \int_{[0,1]}  \mathds{1}_{R <  \eps } \,  \buR \, \dm \, \md R =  \left( 1 - \bs \, \eps^2 \right)^k \kappa_{ij}^{ub} + \smallO(\eps)\,,
\end{multline*}
where we denoted 
\begin{equation*}
	\rho_{ij}^{ub} = 	\int_{0}^1  \buR \, \dm \,  \md R.
\end{equation*}
We conclude taking $\eps=k^{-\frac a2}$, with $a\in(0,1)$.
\end{proof}
\begin{remark} A particular important case is when the kernels are all bounded $$b_{ij}(\hat{u}\cdot \sigma); \quad \tilde{b}^{ub}_{ij}(R) \ \text{and} \  \dm; \quad \tilde{b}^{ub}_{ij}(r, R) \ \text{and} \  \dpo \in L^{\infty}.$$
In such case there is an explicit rate for $ \mathcal{C}_k^{ij}$, namely, $\mathcal{C}_k^{ij}  \leq \frac{C}{\sqrt{k}}$.  Indeed, we analyse the most restrictive case of poly-mono \& mono-poly interactions. One splits $\bS \times[0,1]$ into the subregions
\begin{equation*}
\begin{split}
A_\eps &= \left\{ (\sigma, R) \in \bS \times[0,1]:  R\geq \eps  \right\},\\
A_\eps^c &= \left\{ (\sigma, R) \in \bS \times[0,1]: R < \eps  \right\}.
\end{split}
\end{equation*}
In $A_\eps$ 
\begin{equation*}
	1 - \bs R \left(  1 - |\Vs| \right)  \leq 1 - \bs \, \eps\left(  1 - |\Vs| \right)\,,
\end{equation*}
so that 
\begin{align*}
	\int_{A_\eps}  \left( 1 - \bs R \left(  1 - |\Vs| \right)  \right)^k b_{ij}(\hat{u}\cdot \sigma) \, \buR \, \psi_i(R) \, \sqrt{R} \, \md \sigma \, \md R \\
	 \leq C\int^{1}_{0}\left( 1 - \bs \, \eps\,s \right)^k {\rm d}s \leq \frac{C}{k\,\bs\,\eps}\,.
\end{align*}
Whereas, in $A_\eps^c$
\begin{align*}
\int_{A_\eps^c}  \left( 1 - \bs R \left(  1 - |\Vs| \right)  \right)^k b_{ij}(\hat{u}\cdot \sigma) \, \buR \, \psi_i(R) \, \sqrt{R} \, \md \sigma \, \md R \leq C\eps\,.
\end{align*}
The result follows minimising in $\eps$, that is choosing $\eps\sim\frac{1}{\sqrt{k}}$.

\end{remark}

\begin{corollary}\label{Lemma: G} With the notation of  lemmas \ref{Lemma: En Iden} and  \ref{Lemma: Averaging}  and the assumptions on collision kernels stated in Section \ref{Sec: Assump Bij}, the following estimates hold,
	\begin{description}
		\item[\textit{Part (i) - mono-mono interactions}] For $i, j \in \left\{ 1, \dots, M \right\}$ 
		\begin{equation}\label{m-m G}
		\mathcal{G}^+_{ij}(v, v_*) := \int_{\bS}  \left(  \la v' \ra_i^{k} + \la v'_*\ra_j^{k}  \right)  b_{ij}(\hat{u}\cdot\sigma) \, \md \sigma \leq 2 \, \Cpov  \left(  \la v \ra_i^{2} +   \la v_* \ra_j^{2}   \right)^{k/2}.
		\end{equation}
		\item[\textit{Part (ii) - poly-poly interactions}]  For $i,j \in \left\{M+1, \dots, P \right\}$ 
		\begin{equation}\label{p-p G}
			\begin{split}
					\mathcal{G}^+_{ij}(v, v_*, I, I_*) & :=  \int_{\bS \times[0,1]^2} \left(  \la v', I' \ra_i^{k} + \la v'_*, I'_* \ra_j^{k}  \right)   b_{ij}(\hat{u}\cdot \sigma) \,  \buRr \,  \dpo \, \md \sigma \, \md r \, \md R   \\
					& \leq 2 \, \Cpov  \left(  \la v, I \ra_i^{2} + \la v_*, I_* \ra_j^{2}  \right)^{k/2}.
			\end{split}		
		\end{equation}
		\item[\textit{Part (iii) - poly-mono \& mono-poly interactions}]  Consider $i\in \left\{ M+1, \dots, P \right\}$ and $   j \in \left\{ 1,  \allowbreak \dots,  \allowbreak M \right\}$.  
		\begin{equation}\label{p-m G}
			\begin{split}
			\mathcal{G}^+_{ij}(v, v_*, I) &:=	\int_{\bS \times[0,1]}  \left(  \la v', I' \ra_i^{k} + \la v'_*\ra_j^{k}  \right) b_{ij}(\hat{u}\cdot \sigma) \, \buR \, \dm \, \md \sigma \, \md R 
			\\
			& \leq 2 \, \Cpov  \left(  \la v, I \ra_i^{2} + \la v_*\ra_j^{2}  \right)^{k/2},
			\end{split}
		\end{equation}
	\end{description}
where the constant 	$\Cpov$ is characterized in Lemma \ref{Lemma: Averaging}.		
\end{corollary}

\begin{proof}
The estimates follow from lemmas \ref{Lemma: En Iden} and \ref{Lemma: Averaging}.
\end{proof}

\subsection{Estimates on the $k$-moments of the collision operator in a bi-linear form}\label{Sec: bi-linear form coll op}
In this section, we will consider moments of the collision operator written in a bi-linear form.  Namely, depending on $i$ and $j$, whether they refer to monatomic or polyatomic species, there are three bi-linear forms,\\
\textit{(i)}  $i, j \in \left\{ 1, \dots, M \right\}$ 
\begin{equation}\label{bi-linear coll op m-m}
  \mathcal{Q}_k^{ij}[f,g] :=	\int_{\bR} Q_{ij}(f,g)(v) \, \la v \ra_i^{k} \, \md v +  \int_{\bR} Q_{ji}(g,f)(v) \,  \la v \ra_j^{k} \, \md v,
\end{equation}
\noindent \textit{(ii)}  $i\in \left\{ M+1, \dots, P \right\}$ and $   j \in \left\{ 1,  \allowbreak \dots,  \allowbreak M \right\}$
\begin{equation}\label{bi-linear coll op p-m}
\mathcal{Q}_k^{ij}[f,g]  := \int_{\bRfp} Q_{ij}(f,g)(v, I) \, \la v, I \ra_i^{k} \, \md v\, \md I + \int_{\bR} Q_{ji}(g,f)(v) \, \la v \ra_j^{k}  \, \md v,
\end{equation}
\noindent \textit{(iii)} $i,j \in \left\{M+1, \dots, P \right\}$ 
\begin{equation}\label{bi-linear coll op p-p}
\mathcal{Q}_k^{ij}[f,g] :=	\int_{\bRfp} Q_{ij}(f,g)(v, I) \, \la v, I \ra_i^{k} \, \md v\, \md I +  \int_{\bRfp} Q_{ji}(g,f)(v, I) \,  \la v, I \ra_j^{k} \, \md v\, \md I.
\end{equation}
\ \\

The goal of this section is to estimate these bi-linear forms in terms of suitable statistical moments of the input functions $f$ and $g$ with the help of the Compact Manifold Average Lemma, as presented in Lemma \ref{Lemma: Averaging} or Corollary \ref{Lemma: G}, and the following pointwise lemma which provides, in the sequel, a new estimate valid in the whole range $k>2$ for propagation of moments estimates.  
\begin{lemma}[p-Binomial inequality] \label{PI}
	Assume $p>1$.  For all $x, y>0$, the following inequality holds
\begin{equation}\label{polynomial 1}
	\left( x + y \right)^{p}  
	\leq x^{p} + y^{p} + 2^{p+1} \big(x y^{p- 1} \; 1_{y\geq x} + x^{p- 1} y\,1_{x\geq y} \big)\,.\end{equation}
\end{lemma}
\begin{proof}
Consider the function $\varphi(z)=z^{p} + 1 + 2^{p+1}z - (z+1)^{p}$ in $z\in[0,1]$. Note that $\varphi(0)=0$ and $\varphi'(0)>0$, that is $\varphi$ is nonnegative in some $[0,z_*]$.  Certainly $z_*\geq z_{c}$ with $z_{c}$ the smallest critical point of $\varphi$ in $(0,1)$.   Now, $\varphi'(z)=p\,z^{p-1} + 2^{p+1} - p\,(z+1)^{p-1}$, therefore $\varphi'(z_c)=0$ implies
\begin{equation*}
\frac{z_c + 1}{z_c} = \Big(1 + \frac{2^{p+1}}{pz^{p-1}_{c}}\Big)^{\frac{1}{p-1}}\geq \Big(\frac{4}{p}\Big)^{\frac{1}{p-1}}\frac{2}{z_{c}}\,.
\end{equation*}
That is, $z_{c}\geq \big(\frac{4}{p}\big)^{\frac{1}{p-1}}\,2-1$.  It is not difficult to check that $\big(\frac{4}{p}\big)^{\frac{1}{p-1}}\geq \frac34$, { for $p>1$}, which implies that $z_{c}\geq\frac12$.  Consequently, $\varphi\geq0$ in the interval $[0,\frac12]$.  As for the interval $(\frac12,1]$ note that in such interval
\begin{equation*}
\varphi(z)\geq \min_{z\in(\frac12,1]}\Big(z^{p} + 1 + 2^{p+1}z - (z+1)^{p}\Big)\geq \frac{1}{2^{p}}+1 + 2^{p} - 2^{p}\geq0.
\end{equation*}
This proves that $\varphi$ is nonnegative in $[0,1]$, that is
\begin{equation*}
(z+1)^{p}\leq z^{p} + 1 + 2^{p+1}z \,,\qquad z\in[0,1]\,.
\end{equation*}
Now, for $x, y>0$ write
\begin{equation*}
(x+y)^{p} = y^{p}\Big(\frac{x}{y}+1\Big)^{p}1_{y\geq x} + x^{p}\Big(\frac{y}{x}+1\Big)^{p}1_{x\geq y}
\end{equation*}
and, then, conclude applying the aforementioned inequality with $z=\frac{x}{y}$ and $z=\frac{y}{x}$ respectively. 
\end{proof}

\begin{lemma}\label{Lemma bi-linear form}
	Let non-negative functions $f, g \in L^1_{k+\gm}$, with $k \geq \ks$, $\ks$ is from \eqref{ks povzner}, $\gm$  from \eqref{gamma w}. There exist non-negative constants $ A_\star^{ij}, \eps >0$ and ${B}^{ij}_{k}$, such that the following estimate holds on the bi-linear forms \eqref{bi-linear coll op m-m}--\eqref{bi-linear coll op p-p} with the collision kernel  satisfying assumptions stated in Section \ref{Sec: Assump Bij},
	\begin{multline}\label{bi-linear coll op estimate}
			 \mathcal{Q}_k^{ij}[f,g] 	\leq
			-  {A}_{\star}^{ij} \, \mfj_0[g] \, \mfi_{k+\gamma_{ij}}[f]  -{A}_{\star}^{ji} \,  \mfi_0[f] \, \mfj_{k+\gamma_{ij}}[g] \\
		+4 \, \eps  \left( \mfj_0[g] \, \mfi_{k+\gi}[f]  +  \mfi_0[f] \, \mfj_{k+\gj}[g] \right)   
			+
			{B}_k^{ij}[f,g] + {B}_k^{ji}[g,f],
	\end{multline}
where  expressions are to be understood depending on indices  $i$ and $j$ according to \eqref{poly moment i}. Constants  are explicitly computed in the proof with the final expression given in Remark \ref{Remark const}.
\end{lemma}

\begin{proof} We firstly present the proof for poly-poly interactions, so taking $i,j \in \left\{M+1, \dots, P \right\}$.  By the weak form \eqref{p-p weak form} and assumptions on the collision kernel,
\begin{align}
	&  \mathcal{Q}_k^{ij}[f,g] 
		=  \int_{(\bRfp)^2} \int_{\bS \times [0,1]^2} f(v, I) \, g(v_*, I_*) \left\{  \la v', I' \ra_i^{k} +  \la v'_*, I'_* \ra_j^{k} -  \la v, I \ra_i^{k}  -  \la v_*, I_* \ra_j^{k} \right\} \nonumber   \\
	&		\qquad \qquad  	\qquad   \qquad \qquad \qquad \qquad   \times \mathcal{B}_{ij}(v, v_*, I, I_*, r,  \sigma, R) \, \dpo \, \md \sigma \, \md r \, \md R \, \md v_* \, \md I_* \, \md v \, \md I \label{p-p bilinear proof 0}  \\
	&	\leq   \int_{(\bRfp)^2} \int_{\bS \times [0,1]^2} f(v, I) \, g(v_*, I_*)  \tilde{\mathcal{B}}_{ij}(v, v_*, I, I_*) \nonumber \\
	& \qquad \qquad \qquad \qquad  \qquad \qquad \times   \left\{  \mathcal{G}^+_{ij}(v, v_*, I, I_*) - \kappa_{ij}^{lb} \left( \la v, I \ra_i^{k}  +  \la v_*, I_* \ra_j^{k} \right) \right\}  \, \md v_* \, \md I_* \, \md v \, \md I. \label{p-p bilinear proof}  
\end{align}
The term $ \mathcal{G}^+_{ij}(v, v_*, I, I_*) $  introduced in \eqref{p-p G} can be split by exploiting the estimate \eqref{polynomial 1}, given in Lemma \ref{PI}, yielding 
\begin{multline}\label{pomocna 3}
\left(  \la v, I \ra_i^{2} + \la v_*, I_* \ra_j^{2}  \right)^{k/2} \leq  \la v, I \ra_i^{k}  +  \la v_*, I_* \ra_j^{k} 
\\
+  \tilde{c}_k \left( \left\langle v, I \right\rangle_i^2 \left\langle v_*, I_* \right\rangle_j^{k- 2} \mathds{1}_{ \left\langle v, I \right\rangle_i \leq \left\langle v_*, I_* \right\rangle_j } + \left\langle v, I \right\rangle_i^{k- 2} \left\langle v_*, I_* \right\rangle_j^2  \mathds{1}_{   \left\langle v_*, I_* \right\rangle_j \leq \left\langle v, I \right\rangle_i} \right),
\end{multline}
where the constant $\tilde{c}_k$ is
\begin{equation}\label{c tilda k}
  \tilde{c}_k= 2^{\frac{k}2+1}.
\end{equation}
This estimate implies
\begin{multline*}
 \mathcal{Q}_k^{ij}[f,g] 	\leq   \int_{(\bRfp)^2} \int_{\bS \times [0,1]^2} f(v, I) \, g(v_*, I_*)  \tilde{\mathcal{B}}_{ij}(v, v_*, I, I_*) \\
	\times   \left\{   - \left( \kappa_{ij}^{lb} - 2 \, \mathcal{C}_k^{ij} \right) \left( \la v, I \ra_i^{k}  +  \la v_*, I_* \ra_j^{k} \right)  
	\right.
	\\
	\left.	  
	+ 2 \, \mathcal{C}_k^{ij}  \tilde{c}_k \left( \left\langle v, I \right\rangle_i^2 \left\langle v_*, I_* \right\rangle_j^{k- 2} \mathds{1}_{ \left\langle v, I \right\rangle_i \leq \left\langle v_*, I_* \right\rangle_j }  + \left\langle v, I \right\rangle_i^{k- 2} \left\langle v_*, I_* \right\rangle_j^2  \mathds{1}_{   \left\langle v_*, I_* \right\rangle_j \leq \left\langle v, I \right\rangle_i} \right)   \right\}  \, \md v_* \, \md I_* \, \md v \, \md I.
\end{multline*}
Since,   by the assumption of this Lemma, $k \geq \ks$ with $\ks$ from \eqref{ks povzner},   the constant in front of the highest order term is strictly positive. Moreover, as $\mathcal{C}_k^{ij} $ is decreasing in $k$, we choose
\begin{equation}\label{A star tilde}
0<\tilde{A}_{\star}^{ij}	:=	 \kappa_{ij}^{lb} - 2 \, \mathcal{C}_{\ks}^{ij}  \leq  \kappa_{ij}^{lb} - 2 \, \mathcal{C}_{k}^{ij}, \quad \text{for any} \ k\geq \ks.
\end{equation}
Using the lower bound for the negative term and the upper bound for the positive term from \eqref{p-p bounds on Bij tilde} in the  Appendix Lemma \ref{Lemma bounds on Bij tilde}, 
\begin{multline*}
  \mathcal{Q}_k^{ij}[f,g] 	\leq   \int_{(\bRfp)^2}  f(v, I) \, g(v_*, I_*)     \left\{   - \tilde{A}_{\star}^{ij} L_{ij} \left( \la v, I \ra_i^{k+\gamma_{ij}}  +  \la v_*, I_* \ra_j^{k+\gamma_{ij}} \right) 
		\right.
	\\
	\left.	  
	+ \tilde{A}_{\star}^{ij}  \left( \la v, I \ra_i^{k} \la v_*, I_* \ra_j^{\gamma_{ij}}  +  \la v, I \ra_i^{\gamma_{ij}} \la v_*, I_* \ra_j^{k} \right)
	\right.
	\\
	\left.	  
		+ 4 \, \mathcal{C}_k^{ij}  \tilde{c}_k  
	\left( \left\langle v, I \right\rangle_i^{2} \left\langle v_*, I_* \right\rangle_j^{k- 2+\gamma_{ij}} \mathds{1}_{ \left\langle v, I \right\rangle_i \leq \left\langle v_*, I_* \right\rangle_j }   
	\right.
		\right.
	\\
	\left.	  
		\left.	  
			+ \left\langle v, I \right\rangle_i^{k- 2+\gamma_{ij}} \left\langle v_*, I_* \right\rangle_j^2  \mathds{1}_{  \left\langle v_*, I_* \right\rangle_j \leq \left\langle v, I \right\rangle_i  }    \right)   \right\}  \, \md v_* \, \md I_* \, \md v \, \md I.
\end{multline*}
Majorizing the indicator functions and switching to the moment notation \eqref{poly moment i}, we rewrite the previous inequality as
\begin{multline}\label{pomocna 1}
 \mathcal{Q}_k^{ij}[f,g] 	\leq    - \tilde{A}_{\star}^{ij} L_{ij} \left( \mfj_0[g]\, \mfi_{k+\gamma_{ij}}[f]  + \mfi_0[f] \,  \mfj_{k+\gamma_{ij}}[g] \right) 
	\\
	+\tilde{A}_{\star}^{ij}  \left( \mfi_k[f] \, \mfj_{\gamma_{ij}}[g]  +  \mfi_{\gamma_{ij}}[f] \, \mfj_k[g] \right)
{	+ 4 \, \mathcal{C}_k^{ij}  \tilde{c}_k 
	\left(
\mfi_{2}[f] \, \mfj_{k- 2 +\gamma_{ij}}[g] +	 \mfi_{k- 2+\gamma_{ij}}[f] \,  \mfj_2[g]   \right).}
\end{multline}
The monotonicity of moments \eqref{monotonicity of norm} allows to bound positive terms in \eqref{pomocna 1} related to the moments of the order containing $\gamma_{ij}$ by  moments of the same order but involving $\gi$ or $\gj$ instead of $\gamma_{ij}$, see \eqref{gamma ij assumpt}.  This will be an essential step for the upcoming moments interpolation which requires strict positivity of  rates $\gi, \gj >0$. Thus, \eqref{pomocna 1} becomes
\begin{multline}\label{pomocna 6}
\mathcal{Q}_k^{ij}[f,g] 	\leq    - \tilde{A}_{\star}^{ij} L_{ij} \left( \mfj_0[g]\, \mfi_{k+\gamma_{ij}}[f]  + \mfi_0[f] \,  \mfj_{k+\gamma_{ij}}[g] \right) 
	\\
	+\tilde{A}_{\star}^{ij}  \left( \mfi_k[f] \, \mfj_{\gj}[g]  +  \mfi_{\gi}[f] \, \mfj_k[g] \right)
{	+ 4 \, \mathcal{C}_k^{ij}  \tilde{c}_k 
	\left( \mfi_{2}[f] \, \mfj_{k- 2 +\gj}[g] +
	\mfi_{k- 2+\gi}[f] \,  \mfj_2[g]  \right). }
\end{multline}

Next, we invoke  arguments of   \cite{IG-Alonso-BAMS} that involve moment interpolation formulas 
\begin{equation}\label{mom interpolation}
	\mathfrak{m}^\ell_\lambda \leq  (\mathfrak{m}^\ell_{\lambda_1})^\tau \,  (\mathfrak{m}^\ell_{\lambda_2})^{1-\tau}, \quad 0\leq \lambda_1 \leq \lambda \leq \lambda_2, \quad 0< \tau < 1, \quad \lambda=\tau \lambda_1 + (1-\tau) \lambda_2,
\end{equation}
where $\ell$ will be either $i$ or $j$. For simplicity of the notation, we drop for the moment the reference to the distribution function, i.e. we shorten $\mfi := \mfi[f]$.

Thus, since $\gl >0$, \eqref{mom interpolation} yields
\begin{equation}\label{mk interpolation}
\mfl_k  \leq (\mfl_2)^{\frac{\gl}{k-2+\gl}} (\mfl_{k+\gl})^{\frac{k-2}{k-2+\gl}}, 
\qquad \mfl_{k-2+\gl}  \leq (\mfl_0)^\frac{2}{k+\gl} (\mfl_{k+\gl})^\frac{k+\gl-2}{k+\gl}, \quad \ell \in \left\{ i, j\right\}.
\end{equation}
Incorporating these estimates into \eqref{pomocna 6} and using moment monotonicity \eqref{monotonicity of norm} for $0 < \gi, \gj \leq 2$,
\begin{align}
\mathcal{Q}_k^{ij}[f,g] &	\leq    - \tilde{A}_{\star}^{ij} L_{ij}  \left( \mfj_0[g] \, \mfi_{k+\g}[f]  +  \mfi_0[f] \, \mfj_{k+\g}[g] \right)  \nonumber
	\\
& \quad	+  K_1^{ij}[f,g] \, (\mfi_{k+\gi}[f])^{\frac{k-2}{k-2+\gi}}  + K_1^{ji}[g,f] \, (\mfj_{k+\gj}[g])^\frac{k-2}{k-2+\gj} \nonumber
	\\ 
& \quad 	+ K_2^{ij}[f,g] \, (\mfi_{k+\gi}[f])^{\frac{k+\gi-2}{k+\gi}}
		  +  K_2^{ji}[g,f] \, (\mfj_{k+\gj}[g])^{\frac{k+\gj-2}{k+\gj}}, \nonumber
		  \\
& =:    - \tilde{A}_{\star}^{ij} L_{ij}  \left( \mfj_0[g]  \, \mfi_{k+\g}[f]  + \mfi_0[f] \,  \mfj_{k+\g}[g] \right)  + T_1 + T_2, \label{bi-linear Qij}
\end{align}
with constants
\begin{equation*}
K_1^{ij}[f,g] =  \tilde{A}_{\star}^{ij}  (\mfi_2[f])^{\frac{\gi}{k-2+\gi}}   \mfj_{2}[g], \quad
K_2^{ij}[f,g] = {4} \, \mathcal{C}_k^{ij}  \tilde{c}_k   \,  (\mfi_0[f])^{\frac{2}{k+\gi}}  \,  \mfj_2[g].
\end{equation*}
In order to factorize terms of the order $k+\gi$ and $k+\gj$, we use Young's inequality, 
\begin{equation}\label{Y}
	\left| a \, b\right| \leq \frac{1}{p \, \eps^{p/p'} } \left|a\right|^p + \frac{\eps}{p'} \left|b\right|^{p'}, \quad \text{for} \ \  \eps>0 \ \ \text{and} \ \ \frac{1}{p} + \frac{1}{p'} =1,
\end{equation}
and proceed separately for each term.
\subsubsection*{Term $T_1$.} For the first term of $T_1$, Young's inequality \eqref{Y} implies
\begin{align*}
K_1^{ij}[f,g] \, (\mfi_{k+\gi}[f])^{\frac{k-2}{k-2+\gi}} & \leq  \frac{(\mfj_0[g])^{-p/p'}}{p \,\eps^{p/p'}}  \left(  K_1^{ij}[f,g] \right)^p   + \eps \,  \mfj_0[g] \, (\mfi_{k+\gi}[f])^{\frac{k-2}{k-2+\gi}p'}\\
&= \tilde{K}_1^{ij}[f,g] + \eps \,  \mfj_0[g]  \, \mfi_{k+\gi}[f],
\end{align*}
with the notation 
\begin{equation}\label{K1 tilde}
	\tilde{K}_1^{ij}[f,g]  = \frac{(\mfj_0[g])^{-p/p'}}{p \,\eps^{p/p'}}  \left(  K_1^{ij}[f,g] \right)^p,  
\end{equation}
and the choice 
\begin{equation}\label{Y poly r}
p= \frac{k-2+\gi}{\gi} \ \Rightarrow \ p'= \frac{k-2+\gi}{k-2}.
\end{equation}
The very same computations imply for the counterpart, 
\begin{align*}
	K_1^{ji}[g,f] \, (\mfj_{k+\gj}[g])^{\frac{k-2}{k-2+\gj}} & \leq
 \tilde{K}_1^{ji}[g,f] + \eps \,  \mfi_0[f]  \, \mfj_{k+\gj}[g].
\end{align*}
Thus, term $T_1$ is estimated as follows,
\begin{equation}\label{T1}
	T_1 \leq \tilde{K}_1^{ij}[f,g] + \tilde{K}_1^{ji}[g,f] + \eps \,  \mfj_0[g]  \, \mfi_{k+\gi}[f] + \eps \,  \mfi_0[f]  \, \mfj_{k+\gj}[g],
\end{equation}
where the involved terms are specified in \eqref{K1 tilde} with \eqref{Y poly r}.
\subsubsection*{Term $T_2$.}  
In a similar fashion, for the first part of term $T_2$, Young's inequality \eqref{Y} implies
\begin{align*}
	K_2^{ij}[f,g] \, (\mfi_{k+\gi}[f])^{\frac{k-2+\gi}{k+\gi}} &\leq \frac{(\mfj_0[g])^{-q/q'}}{q \,\eps^{q/q'}}  \left(  K_2^{ij}[f,g] \right)^q   + \eps \,  \mfj_0[g] \, (\mfi_{k+\gi}[f])^{\frac{k-2+\gi}{k+\gi}q'}\\
	&= \tilde{K}_2^{ij}[f,g] + \eps \,  \mfj_0[g]  \, \mfi_{k+\gi}[f],
\end{align*}
with the choice 
\begin{equation}\label{Y poly s}
	q = \frac{k+\gi}{2} \ \Rightarrow \ q'=\frac{k+\gi}{k-2+\gi},
\end{equation}
and notation 
\begin{equation}\label{K3 tilde}
	\tilde{K}_2^{ij}[f,g] = \frac{(\mfj_0[g])^{-q/q'}}{q \,\eps^{q/q'}}  \left(  K_2^{ij}[f,g] \right)^q.
\end{equation}
Thus, for $T_2$ we conclude
\begin{equation}\label{T3}
	T_2 \leq \tilde{K}_2^{ij}[f,g]  + \tilde{K}_2^{ji}[g,f] +  \eps \, \mfj_0[g]  \, \mfi_{k+\gi}[f]  +   \eps \, \mfi_0[f] \, 	\mfj_{k+\gj}[g].
\end{equation}
Gathering \eqref{T1} and \eqref{T3},  the bi-linear form $\mathcal{Q}_{ij}$ becomes 
\begin{equation*}
	\begin{split}
\mathcal{Q}_k^{ij}[f,g] & 	\leq - \tilde{A}_{\star}^{ij} L_{ij}    \left( \mfj_0[g] \, \mfi_{k+\g}[f]  +  \mfi_0[f] \, \mfj_{k+\g}[g] \right)  \\
		& +  {2} \, \eps  \left( \mfj_0[g] \, \mfi_{k+\gi}[f]  +  \mfi_0[f] \, \mfj_{k+\gj}[g] \right)  
	 +	{B}^{ij}[f,g] + {B}^{ji}[g,f], 
	\end{split}
\end{equation*}
where all constants are merged into the one,
\begin{equation}\label{const Bij}
		{B}^{ij}_k[f,g] =  \tilde{K}_1^{ij}[f,g]  +   \tilde{K}_2^{ij}[f,g],
\end{equation}
and we denoted
\begin{equation}\label{A star ij}
	{A}_{\star}^{ij}  =   {\tilde{A}_{\star}^{ij} L_{ij} },
\end{equation}
which concludes the proof for  poly-poly interaction when $i,j \in \left\{M+1, \dots, P \right\}$.\\

For monatomic gases, i.e. $i, j \in \left\{ 1, \dots, M \right\}$,   the weak form \eqref{m-m weak form} together with the appropriate assumptions on the collision kernel  from Section \ref{Sec: Assump Bij} and  Lemma \ref{Lemma: G}  imply,
\begin{align*}
	& 	\mathcal{Q}_k^{ij}[f,g]	=   \int_{(\bR)^2} \int_{\bS  } f(v) \, g(v_*) \left\{  \la v' \ra_i^{2k} +  \la v'_* \ra_j^{2k} -  \la v \ra_i^{2k}  -  \la v_*  \ra_j^{2k} \right\}  \mathcal{B}_{ij}(v, v_*,  \sigma) \, \md \sigma \, \md v_*  \, \md v  \\
	&	\leq  \int_{(\bR)^2} \int_{\bS} f(v) \, g(v_*) \,  \tilde{\mathcal{B}}_{ij}(v, v_*)   \left\{  \mathcal{G}^+_{ij}(v, v_*) - \kappa_{ij}^{lb} \left( \la v \ra_i^{2k}  +  \la v_* \ra_j^{2k} \right) \right\}  \, \md v_* \,   \md v.
\end{align*}
Therefore, the very same computation as the one after the line \eqref{pomocna 3} leads to the desired estimate \eqref{bi-linear coll op estimate}. Similarly, for poly-mono interaction when  $i \in \left\{ M+1, \dots, P \right\}$ and  $ j \in \left\{ 1, \dots, M \right\}$, the weak form  \eqref{weak vs v I}, assumptions on the collision kernel  from Section \ref{Sec: Assump Bij} and Lemma \ref{Lemma: G} yield
\begin{align*}
	&	 \mathcal{Q}_k^{ij}[f,g] = \int_{\bR \times \bR \times \bRp} \int_{\bS \times [0,1]}  f(v, I) \, g(v_*) \left\{  \la v', I' \ra_i^{2k} + \la v'_*\ra_j^{2k}  -  \la v, I \ra_i^{2k} - \la v_* \ra_j^{2k}  \right\}  \\
	& \qquad \qquad \qquad \qquad \qquad \qquad \qquad \times \mathcal{B}_{ij}(v, v_*, I, \sigma, R) \, \dm \, \md \sigma \, \md R \, \md v_* \, \md v \, \md I  \\
	&  \leq  \int_{\bR \times \bR \times \bRp}  f(v, I) \, g(v_*)  \, \tilde{\mathcal{B}}_{ij}(v, v_*, I) \left\{  \mathcal{G}^+_{ij}(v, v_*, I)   - \kappa_{ij}^{lb}  \left( \la v, I \ra_i^{2k} + \la v_* \ra_j^{2k} \right) \right\} \, \md v_* \, \md v \, \md I,
\end{align*}
and so proceeding in the same fashion as for poly-poly interaction after \eqref{pomocna 3} leads to \eqref{bi-linear coll op estimate}, which concludes the proof.
\end{proof}

\begin{remark}\label{Remark const}
	The constant $B^{ij}_{k}[f,g]$ is explicit and depends on $k$, $\eps$, $\gi$, $\gj$ and on $f$ and $g$ through their zero and second order species moment. Its formula is given in \eqref{const Bij} after gathering expressions for  $\tilde{K}_1^{ij}[f,g]$ and $ \tilde{K}_2^{ij}[f,g]$,
\begin{equation*}
	\begin{split}
	\tilde{K}_1^{ij}[f,g] & = \eps^{-\frac{k-2}{\gi}} \frac{\gi}{k-2+\gi} \left( \tilde{A}_\star^{ij}  \right)^{\frac{k-2+\gi}{\gi}} \mathfrak{m}^i_2[f] \, \mathfrak{m}^j_0[g]^{-\frac{k -2}{\gi}} \, \mathfrak{m}^j_2[g]^{\frac{k-2+\gi}{\gi}},  \\
	\tilde{K}_2^{ij}[f,g] & =  \eps^{- \frac{1}{2+\gj-\gi} \left(  (k-2+\gi)  + \frac{(k+\gj)\gi}{k-2} \right) }  \left(\frac{2+\gj}{k+\gj}\right)^{\frac{(2+\gj)(k-2+\gi)}{(2+\gj-\gi)(k-2)}} \left( {4} \, \Cpov \, \tilde{c}_k \right)^{\frac{(k+\gj)(k-2+\gi)  }{(k-2)(2+\gj-\gi)}}  \\ 
	& \qquad \qquad \times \mathfrak{m}^j_0[g] \, \mathfrak{m}^i_0[f]^{-\frac{k -2+\gi}{2+\gj-\gi}} \, \mathfrak{m}^i_2[f]^{\frac{k+\gj}{2+\gj-\gi}},
	\end{split}
\end{equation*}	
where $\tilde{A}_\star^{ij}$ is from \eqref{A star tilde},  $\Cpov$ as in \eqref{C pov} and $\tilde{c}_k$ is specified in \eqref{c tilda k}.
\end{remark}

{
\begin{lemma}\label{Lemma bi-linear form below ks}
	Let non-negative functions $f, g \in L^1_{k}$, with $k > 2$.  The following inequality holds for the bi-linear forms \eqref{bi-linear coll op m-m}--\eqref{bi-linear coll op p-p} with the collision kernel  satisfying assumptions stated in Section \ref{Sec: Assump Bij},
	\begin{equation}\label{bi-linear coll op estimate below ks}
	\mathcal{Q}_k^{ij}[f,g] 	\leq
	2 \, \kappa_{ij}^{ub} \, \tilde{c}_k 
	\left( 
	\mfi_{2}[f] \, \mfj_{k }[g] +	 \mfi_{k }[f] \,  \mfj_2[g] 	 \right),
	\end{equation}
with $ \kappa_{ij}^{ub} $ from \eqref{kappa p-p} and $\tilde{c}_k$ from \eqref{c tilda k}.
\end{lemma}

\begin{proof}
As in the previous Lemma \ref{Lemma bi-linear form}, we  present the proof for poly-poly interactions, i.e. for $i,j \in \left\{ M+1, \dots, P \right\}$ keeping in mind that other types of interactions follow the same strategy. 

Start with the bi-linear form \eqref{p-p bilinear proof 0},
\begin{align}
	& \mathcal{Q}_k^{ij}[f,g] 
	=  \int_{(\bRfp)^2} \int_{\bS \times [0,1]^2} f(v, I) \, g(v_*, I_*) \left\{  \la v', I' \ra_i^{k} +  \la v'_*, I'_* \ra_j^{k} -  \la v, I \ra_i^{k}  -  \la v_*, I_* \ra_j^{k} \right\} \nonumber   \\
	&		\qquad \qquad  	\qquad   \qquad \qquad \qquad \qquad   \times \mathcal{B}_{ij}(v, v_*, I, I_*, r,  \sigma, R) \, \dpo \, \md \sigma \, \md r \, \md R \, \md v_* \, \md I_* \, \md v \, \md I. \nonumber
\end{align}
Then the primed quantities are estimated in terms of non-primed by using the conservation of energy \eqref{p-p coll CL}. Indeed, for $k > 2$, 
\begin{equation*}
	 \la v', I' \ra_i^{k} +  \la v'_*, I'_* \ra_j^{k} \leq   \left(\la v', I' \ra_i^{2} +  \la v'_*, I'_* \ra_j^{2}  \right)^{k/2}  = \left(\la v, I \ra_i^{2} +  \la v_*, I_* \ra_j^{2}   \right)^{k/2}.
\end{equation*}
The last term is estimated via  \eqref{pomocna 3}, which in combination with the loss term implies cancellation of the highest order moment,
\begin{align*}
	&\mathcal{Q}_k^{ij}[f,g] 
	\leq \tilde{c}_k  \int_{(\bRfp)^2} \int_{\bS \times [0,1]^2} f(v, I) \, g(v_*, I_*) \, \mathcal{B}_{ij}(v, v_*, I, I_*, r,  \sigma, R) \, \dpo \\ 
&	\times \left( \left\langle v, I \right\rangle_i^2 \left\langle v_*, I_* \right\rangle_j^{k- 2} \mathds{1}_{ \left\langle v, I \right\rangle_i \leq \left\langle v_*, I_* \right\rangle_j } + \left\langle v, I \right\rangle_i^{k- 2} \left\langle v_*, I_* \right\rangle_j^2  \mathds{1}_{   \left\langle v_*, I_* \right\rangle_j \leq \left\langle v, I \right\rangle_i} \right)  \md \sigma \, \md r \, \md R \, \md v_* \, \md I_* \, \md v \, \md I.  \\
\end{align*}
Then, the assumption on the collision kernel \eqref{p-p ass B}  allows to integrate over $(\sigma, r, R)$ implying, with notation \eqref{kappa p-p},
\begin{equation*}
	\begin{split}
&  \mathcal{Q}_k^{ij}[f,g] 
	\leq  \kappa_{ij}^{ub} \,  \tilde{c}_k   \int_{(\bRfp)^2}  f(v, I) \, g(v_*, I_*) \, \tilde{\mathcal{B}}_{ij}(v, v_*, I, I_*)
\\& \times	 \left( \left\langle v, I \right\rangle_i^2 \left\langle v_*, I_* \right\rangle_j^{k- 2} \mathds{1}_{ \left\langle v, I \right\rangle_i \leq \left\langle v_*, I_* \right\rangle_j } + \left\langle v, I \right\rangle_i^{k- 2} \left\langle v_*, I_* \right\rangle_j^2  \mathds{1}_{   \left\langle v_*, I_* \right\rangle_j \leq \left\langle v, I \right\rangle_i} \right)  \,   \md v_* \, \md I_* \, \md v \, \md I.
\end{split}
\end{equation*}	 
The last step consists in exploiting the upper bound \eqref{p-p bounds on Bij tilde} and conveniently using the indicator functions, as in the previous lemma, to get
	 \begin{align*}
	  \mathcal{Q}_k^{ij}[f,g] 
	 &\leq 2 \, \kappa_{ij}^{ub} \, \tilde{c}_k 
	\left( 
	\mfi_{2}[f] \, \mfj_{k- 2 +\gamma_{ij}}[g] +	 \mfi_{k- 2+\gamma_{ij}}[f] \,  \mfj_2[g] 	 \right)
		 \\
	& \leq 2 \, \kappa_{ij}^{ub} \, \tilde{c}_k 
	\left( 
	\mfi_{2}[f] \, \mfj_{k }[g] +	 \mfi_{k }[f] \,  \mfj_2[g] 	 \right),
\end{align*}
where the last estimate is due to monotonicity of moments \eqref{monotonicity of norm}.

\end{proof}

}
\subsection{Estimates on the moments of the vector valued collision operator}\label{Sec: estimates vector coll op}
In this section, we consider the vector valued distribution function $\F$ and the collision operator $\Q(\F)$ as introduced in \eqref{F Q vector} and provide estimates on the collision operator moments in terms of the moments of the distribution function $\F$. In other words, we switch from the bi-linear forms presented in Section \ref{Sec: bi-linear form coll op} to a vector valued form that combines all possible interactions among species $\mathcal{A}_i$ and $\mathcal{A}_j$, for $i,j=1,\dots,P$. For this Section, recall Definition \ref{def moment} of polynomial moments for vector valued  function $\F$.

\begin{lemma}\label{Lemma: coll op estimate}
	{Take the vector valued collision operator $\Q(\F)$ defined in \eqref{F Q vector} with the collision kernel  satisfying assumptions stated in Section \ref{Sec: Assump Bij}.  The following estimates on  $\mathfrak{m}_k[\Q(\F)]$ hold,}
	\begin{enumerate}
		\item {For $k>2$, assuming that each component $f_i$ of the vector valued distribution function $\F$ satisfy assumptions of Lemma \ref{Lemma bi-linear form below ks} 
			\begin{equation}\label{mk of Q below ks}
				\mathfrak{m}_k[\Q(\F)] 
				\leq   D_k \, \mathfrak{m}_{k}[\F],
			\end{equation}
		with $D_k$ explicitly given in \eqref{Dk}.
		}
		\item For $k \geq \ks$, assuming that each component $f_i$ of the vector valued distribution function $\F$ satisfy assumptions of Lemma \ref{Lemma bi-linear form} and additionally $\mfi_0[f_i]>0$	for each $i=1,\dots,P$, there exist positive constants  ${A}_\star$  and $B_k$ related to those in Lemma \ref{Lemma bi-linear form} such that
		\begin{equation}\label{mk of Q}
			\mathfrak{m}_k[\Q(\F)] 
			\leq  - 	{A}_\star \, \mathfrak{m}_2[\F]^{-\frac{\gw}{k-2}} \ \mathfrak{m}_{k}[\F]^{1+\frac{\gw}{k-2}} +  B_k,
		\end{equation}
	with $\gw$ defined in \eqref{gamma w}.  An explicit choice for $A_\star$ is given in \eqref{As} and asymptotic expression of $B_{k}$ in \eqref{asympBk}. 
	\end{enumerate}

\end{lemma}

\begin{proof} Starting from the definition of moments \eqref{poly moment i} for a vector valued function, the first goal is to write $\mathfrak{m}_k$-th moment of the collision operator $\Q(\F)$ in terms of bi-linear forms {\eqref{bi-linear coll op m-m}, \eqref{bi-linear coll op p-m} and \eqref{bi-linear coll op p-p}}. To that end, we use the weak form \eqref{weak form general} for the test function $\psi_i(\cdot) = \la \cdot \ra_i^k$,
\begin{equation}\label{moment of vector Q}
\mathfrak{m}_k[\Q(\F)] 
= \frac{1}{2} \sum_{i, j =1}^M \mathcal{Q}_{ij}[f_i, f_j]  +  \frac{1}{2}  \sum_{i, j=M+1}^P  \mathcal{Q}_{ij}[f_i, f_j]  
 +  \sum_{j=1}^M  \sum_{i=M+1}^P   \mathcal{Q}_{ij}[f_i, f_j].
\end{equation}
{For the first part of the statement valid for  any $k>2$, plugging the estimate \eqref{bi-linear coll op estimate below ks} 	on each  bi-linear form and rearranging the sum,
	\begin{equation*} 
				\mathfrak{m}_k[\Q(\F)] 
		\leq	2\, \tilde{c}_k \left( \max_{1 \leq i, j  \leq P} \kappa_{ij}^{ub} \right)   \mathfrak{m}_2[\F] \, \mathfrak{m}_{k}[\F],
	\end{equation*}
	which leads  to the statement  \eqref{mk of Q below ks}  with the constant
	\begin{equation}\label{Dk}
D_k = 2\, \tilde{c}_k \left( \max_{1 \leq i, j  \leq P} \kappa_{ij}^{ub} \right) \, \mathfrak{m}_2[\F].
	\end{equation}

The second part valid for $k\geq \ks$ starts with \eqref{moment of vector Q}.   Estimates   \eqref{bi-linear coll op estimate} on bi-linear forms and rearrangement of the sum imply } 
\begin{multline}
		\mathfrak{m}_k[\Q(\F)]  \leq  \sum_{i, j =1}^M \left( - 	A_\star^{ij}  \mfj_0 \, \mfi_{k+\g} + {2} \, \eps \,   \mfj_0 \, \mfi_{k+\gi}   + B^{ij}_k   \right)    \\
		 +   \sum_{i, j=M+1}^P  \left( - 	A_\star^{ij}  \mfj_0 \, \mfi_{k+\g}  + {2} \, \eps \,   \mfj_0 \, \mfi_{k+\gi}  + B^{ij}_k   \right)    \\
		 +  \sum_{j=1}^M  \sum_{i=M+1}^P \left(   	-  {A}_{\star}^{ij} \, \mfj_0 \, \mfi_{k+\gamma_{ij}}
		+ {2} \, \eps \,   \mfj_0 \, \mfi_{k+\gi} +	{B}^{ij}_k 
		-{A}_{\star}^{ji} \,  \mfi_0 \, \mfj_{k+\gamma_{ij}} + 	{2} \, \eps  \,  \mfi_0[f_i] \, \mfj_{k+\gj}[f_j] 
		+ {B}^{ji}_k \right)  \\
		 = \sum_{i, j =1}^P \left( - 	A_\star^{ij}  \mfj_0 \, \mfi_{k+\g}  + {2} \, \eps \,   \mfj_0 \, \mfi_{k+\gi}  + B^{ij}_k   \right), \label{pomocna 4}
\end{multline}
where for the simplicity of notation, we denoted $  \mfi_k := \mfi_k[f_i]$ and ${B}^{ij}_k:= {B}^{ij}_k[f_i,f_j]$, since the reference to the distribution function is clear in this case.

The first goal is to absorb the  term multiplying a small constant $\eps$ by a negative term. To that end, we use the lower bound for the negative term.  Since $A_\star^{ij}>0$ and by assumption $\mfj_0 >0$ for every $j$, we denote
\begin{equation}\label{As}
\As= \frac{1}{2} \min_{1 \leq i, j  \leq P} \left( A_\star^{ij}  \mfj_0\right) >0.
\end{equation}
On the other side,
\begin{equation*}
	\sum_{i, j =1}^P    \mfi_{k+\g} \geq \sum_{i=1}^P  \mfi_{k+\gi}, 
\end{equation*}
with $\gi$  introduced in \eqref{gamma ij assumpt}. Thus, \eqref{pomocna 4} becomes
\begin{equation*}
	\mathfrak{m}_k[\Q(\F)]  \leq 
	 - \left( 2 \As	  -  {2} \, \eps \,   \mathfrak{m}_0  \right) \sum_{i =1}^P	\, \mfi_{k+\gi}    +  \sum_{i, j =1}^P B^{ij}_k.
\end{equation*}
Therefore, the choice
\begin{equation}\label{eps}
	\eps = \frac{	\As }{2 \, \mathfrak{m}_0 }, 
\end{equation}
ensures the  negative sign of the coefficient of the highest order term.   Estimating 
\begin{equation*}
\sum_{i=1}^P  \mfi_{k+\gi}[f_i]  \geq \mathfrak{m}_{k+\gw}[\F],
\end{equation*}
with  $\gw$  from \eqref{gamma w}, we get the bound on the vector valued collision operator in terms of the polynomial moments of the vector valued distribution function,
\begin{equation*}
	\mathfrak{m}_k[\Q(\F)]  \leq 
	-  	\As \mathfrak{m}_{k+\gw}[\F]    +  \sum_{i, j =1}^P B^{ij}_k.
\end{equation*}
Finally,  the species-moment interpolation \eqref{mk interpolation} for $\mfi_k$ 
implies the moment interpolation for the $k-$th moment of the vector valued $\F$ by H\"{o}lder's inequality,
\begin{equation*}
	\mathfrak{m}_{k}[\F] = \sum_{i=1}^P 	\mfi_{k}[f_i]   \leq  \sum_{i=1}^P  (\mfi_2[f_i])^{\frac{\gw}{k-2+\gw}}   (\mfi_{k+\gw}[f_i])^{\frac{k-2}{k-2+\gw}} \leq  (	\mathfrak{m}_2[\F])^{\frac{\gw}{k-2+\gw}}    (	\mathfrak{m}_{k+\gw}[\F])^{\frac{k-2}{k-2+\gw}},
\end{equation*}
which gives the lower bound for $ \mathfrak{m}_{k + \gw}$,
\begin{equation*}
	\mathfrak{m}_{k + \gw }[\F]  \geq \mathfrak{m}_2[\F]^{-\frac{\gw}{k-2}}\, \mathfrak{m}_{k}[\F]^{1+\frac{\gw}{k-2}},
\end{equation*}
and therefore for \eqref{pomocna 4} implies
\begin{equation*}
	\mathfrak{m}_k[\Q(\F)] 
	\leq  - \As  \mathfrak{m}_2[\F]^{-\frac{\gw}{k-2}}\, \mathfrak{m}_{k}[\F]^{1+\frac{\gw}{k-2}} +  B_k, \qquad \text{with} \quad B_k :=  \sum_{i, j =1}^P B^{ij}_k,
\end{equation*}
which is exactly the desired estimate \eqref{mk of Q}.  

In addition, {constants $\tilde{K}_1^{ij}[f_i,f_j]$ and $\tilde{K}_2^{ij}[f_i,f_j]$ in Remark \ref{Remark const}  can be made explicit by plugging the choice of $\eps$ from \eqref{eps}, namely,  
\begin{equation*}
	\begin{split}
		\tilde{K}_1^{ij} & = \left( \frac{	\As }{ 2  \, \mathfrak{m}_0[\F] }\right)^{-\frac{k-2}{\gi}} \frac{\gi}{k-2+\gi} \left( \frac{{A}_\star^{ij}}{L_{ij}} \right)^{\frac{k-2+\gi}{\gi}} \mathfrak{m}^i_2[f_i] \, \mathfrak{m}^j_0[f_j]^{-\frac{k -2}{\gi}} \, \mathfrak{m}^j_2[f_j]^{\frac{k-2+\gi}{\gi}},  \\
		\tilde{K}_2^{ij} & =   \left( \frac{	\As }{ 2 \, \mathfrak{m}_0[\F] }\right)^{- \frac{1}{2+\gj-\gi} \left(  (k-2+\gi)  + \frac{(k+\gj)\gi}{k-2} \right) }  \left(\frac{2+\gj}{k+\gj}\right)^{\frac{(2+\gj)(k-2+\gi)}{(2+\gj-\gi)(k-2)}} \left( 8 \, \Cpov \, 2^{\frac k2} \right)^{\frac{(k+\gj)(k-2+\gi)  }{(k-2)(2+\gj-\gi)}}  \\ 
		& \qquad \qquad \times \mathfrak{m}^j_0[f_j] \, \mathfrak{m}^i_0[f_i]^{-\frac{k -2+\gi}{2+\gj-\gi}} \, \mathfrak{m}^i_2[f_i]^{\frac{k+\gj}{2+\gj-\gi}}.
	\end{split}
\end{equation*}	 } A simple comparison of the constants  shows that the one driving the asymptotic behaviour in $B_{k}$, as $k$ increases, is the latter.  Thus, 
\begin{equation}\label{asympBk}
B_{k} \sim \max_{ij}\tilde{K}_2^{ij}[f_i,f_j] \sim \max_{ij}\mathfrak{m}^j_0[f_j]\bigg(\frac{{16}\,  \Cpov\,\mathfrak{m}^i_2[f_i]  \mathfrak{m}_0[\F]}{\mathfrak{m}^i_0[f_i]\,\As}\bigg)^{\frac{k}{2+\gj-\gi}}  \left(\frac{2+\gj}{k}\right)^{\frac{2+\gj}{2+\gj-\gi}}\,2^{\frac{k^2}{2(2+\gj-\gi)}}\,.
\end{equation}

\end{proof}

\section{Polynomial moments a priori estimates on the solution to the Boltzmann system}\label{Sec: poly mom}

In this section, we prove a priori estimates on  polynomial moments for the  solution of the system of Boltzmann equations \eqref{BS vector}.

\begin{proposition}[Moment ODI]\label{Lemma moment ODI}
Let $\F=[f_i]_{1\leq i \leq P} \geq 0$ be a solution of the system of Boltzmann equations \eqref{BS vector} having the collision kernel  satisfying assumptions listed in Section \ref{Sec: Assump Bij}. The $k$-th polynomial moment of the system solution $\F$ satisfies the {following} ODIs,
\begin{enumerate}
	\item {For $k>2$ and $\F \in L^1_k$, 
		\begin{equation}\label{moment ODI below ks}
		\frac{\md \mathfrak{m}_k[\F]}{\md t} \leq D_k \,    \mathfrak{m}_{k}[\F],
	\end{equation}
}
	\item For $k \geq \ks$, $\ks$ from \eqref{ks povzner},  and for  $\F \in L^1_{k+\gm}$, $\gm$  from \eqref{gamma w},  with additionally  $\mfi_0[f_i]>0$	for each $i=1,\dots,P$. Then for any $\gw>0$ defined in \eqref{gamma w}, 
	\begin{equation}\label{moment ODI}
		\frac{\md \mathfrak{m}_k[\F]}{\md t} \leq - 	\As \, \mathfrak{m}_2[\F]^{-\frac{\gw}{k-2}} \    \mathfrak{m}_{k}[\F]^{1+\frac{\gw}{k-2}}  +  B_k,
	\end{equation}
\end{enumerate}
where explicit form of constants ${A}_\star>0$ and $B_k, {D_k}\geq 0$  is given in Lemma \ref{Lemma: coll op estimate}.
\end{proposition}

\begin{proof}
The proof follows by taking $\mathfrak{m}_k$-th moment of the Boltzmann system \eqref{BS vector}, \begin{equation}\label{pomocna 5}
	\frac{\md \mathfrak{m}_{k}[\F] }{\md t} =  \mathfrak{m}_{k}[\Q(\F)],
\end{equation}
and applying the estimate on the collision operator polynomial moment \eqref{mk of Q}.
\end{proof}

\begin{theorem}[Generation and propagation of polynomial moments] 
		Let $\F$ satisfy assumptions of the Proposition \ref{Lemma moment ODI} , then the following estimates hold,
		\begin{enumerate}
			\item (Polynomial moments generation estimate.) Define
			\begin{equation}\label{Ek}
	\Ek   = \mathfrak{m}_2[\F]^{\frac{\gw}{k-2 + \gw}} \left( \frac{B_{k}}{\As} \right)^{\frac{k-2}{k-2+\gw}}.
			\end{equation}
			Then for any $k \geq \ks$, 
			\begin{equation}\label{poly gen}
				\mathfrak{m}_{k}[\F](t) \leq {E}_{k}   + \mathfrak{m}_2[\F] \left( \frac{k-2}{\gw \As } \right)^\frac{k-2}{\gw} t^{-\frac{k-2}{\gw}}, \qquad \forall t>0,
			\end{equation}
		whereas for $2< k < \ks$, 
	\begin{equation}\label{Ek1}
			\mathfrak{m}_k[\F](t) 	\leq \mathfrak{m}_2[\F]^\frac{\ks - k+1}{\ks-1}   \left( {E}_{\ks+1}\right)^{\frac{k-2}{\ks-1}}   + \mathfrak{m}_2[\F]^{\frac{1}{\ks-1}} \left( \frac{\ks-1}{\gw \As } \right)^\frac{k-2}{\gw} t^{-\frac{k-2}{\gw}},  \qquad \forall t>0.
		\end{equation}			
			\item (Polynomial moments propagation estimate.) Moreover, {assume} $	\mathfrak{m}_{k}[\F](0) < \infty$.  Then for any $k  \geq \ks$,  the following estimate holds
			\begin{equation}\label{poly prop}
				\mathfrak{m}_{k}[\F](t) \leq \max \left\{ 	\Ek,	\mathfrak{m}_{k}[\F](0) \right\}, \qquad \forall t\geq 0.
			\end{equation}
		{Define 
			\begin{equation}
			\tilde{E}_k =	\mathfrak{m}_2^\frac{\ks - k+1}{\ks-1}   \left( {E}_{\ks+1}\right)^{\frac{k-2}{\ks-1}}   + \mathfrak{m}_2^{\frac{1}{\ks-1}} \left( \frac{\ks-1}{\gw \As } \right)^\frac{k-2}{\gw} D_k^{\frac{k-2}{\gw}},
			\end{equation}
		with $D_k$ from \eqref{Dk}. Then for 
	 $2<k<\ks$,
		 \begin{equation}\label{prop below ks}
\mathfrak{m}_{k}[\F]{(t)} \leq \max \left\{  \tilde{E}_k, e \, \mathfrak{m}_{k}[\F](0) \right\},   \qquad \forall t{\geq }0.
		\end{equation}
	 
	}
		\end{enumerate}
\end{theorem}
\begin{proof}
The proof relies on the ODE comparison principle, already used in \cite{IG-Alonso-BAMS, MPC-IG-poly}. First note that ${E}_k$ is the equilibrium solution of the associated upper ODE problem
	\begin{equation}\label{aux ODE}
	\left\{ 
	\begin{split}
	y'(t) &=	- \As \, \mathfrak{m}_2^{-c} \,  y(t)^{1+c} +  B_k, \qquad  c:= \gw/(k-2), \\
	y(0) &=\mathfrak{m}_{k}(0),
	\end{split}	\right.
\end{equation}	
where we abbreviated $\mathfrak{m}_{k}:= \mathfrak{m}_{k}[\F]$.
If $y(0) =  \mathfrak{m}_{k}(0) < \infty$, then the propagation estimate \eqref{poly prop} follows, since
\begin{equation*}
	\mathfrak{m}_{k}(t)  \leq y(t) \leq \max \left\{ 	\Ek,	\mathfrak{m}_{k}(0) \right\},  \qquad \forall t\geq 0.
\end{equation*}
The generation estimate is proven by constructing the function as in \cite{MPC-IG-poly},
\begin{equation}
	z(t) = \Ek + \mathfrak{m}_2 \left( c \, A_\star \right)^{-1/c} t^{-1/c}, 
\end{equation}
and applying the comparison principle for ODEs stated in the Lemma \ref{Lemma: comparison} proved in \cite{IG-Alonso-BAMS} that yields the estimate on the solution of the moment ODI \eqref{moment ODI},
\begin{equation}\label{mk gen}
	\mathfrak{m}_{k}(t)  \leq z(t) \leq   \Ek+ \mathfrak{m}_2 \left( c \, A_\star \right)^{-1/c} t^{-1/c},  \qquad \forall t > 0,
\end{equation}
for $k>\ks$. For $2< k < \ks$, firstly the interpolation argument is used,
\begin{equation*}
\mathfrak{m}_k \leq \mathfrak{m}_2^\frac{\ks - k+1}{\ks-1} \ \mathfrak{m}_{\ks +1}^{\frac{k-2}{\ks-1}}.
\end{equation*}
Then, for $\mathfrak{m}_{\ks +1}$ we apply the inequality above \eqref{mk gen} and get, since $(k-2)/(\ks -1)<1$,
\begin{equation*}
	\begin{split}
\mathfrak{m}_k 	&\leq \mathfrak{m}_2^\frac{\ks - k+1}{\ks-1} \left( \left( {E}_{\ks+1}\right)^{\frac{k-2}{\ks-1}}   + \mathfrak{m}_2^{\frac{k-2}{\ks-1}} \left( \frac{\ks-1}{\gw \As } \right)^\frac{k-2}{\gw} t^{-\frac{k-2}{\gw}} \right).
\end{split}
\end{equation*} 
{The propagation result for $2<k<\ks$ firstly uses the ODI   \eqref{moment ODI below ks}  which for short time implies 
 \begin{equation*}
 	\mathfrak{m}_k(t) 
 	\leq  \, e \, \mathfrak{m}_{k}(0) , \quad 0 < t \leq \frac{1}{D_k}.
 \end{equation*}
Then, for $t>\frac{1}{D_k}$,  use generation of moments \eqref{Ek1} that yields
\begin{equation*}
		\mathfrak{m}_k(t) 	\leq \mathfrak{m}_2^\frac{\ks - k+1}{\ks-1}   \left( {E}_{\ks+1}\right)^{\frac{k-2}{\ks-1}}   + \mathfrak{m}_2^{\frac{1}{\ks-1}} \left( \frac{\ks-1}{\gw \As } \right)^\frac{k-2}{\gw} D_k^{\frac{k-2}{\gw}}.
\end{equation*}
Taking the maximum of the two constants }
  completes the proof.
\end{proof}

\section{Existence and Uniqueness theory}\label{Sec: exi uni}
In this section, we prove the existence and uniqueness theory for the Cauchy problem \eqref{BS vector}. To that end, we first define the set of initial data.  \\

For $i=1,\dots, P$, fix constants $C^i_0, C_2, \Cs >0$, with 
\begin{equation}\label{ks and C's for Omega}
	\Cs \geq \Eks + B_{\kss} =: \mathfrak{h}_{\kss}, \quad \text{with} \ 	\kss=\max\{2+2\gm, \ks\},
\end{equation}
where $\gm$ is from \eqref{gamma w} and $\ks$ is introduced in \eqref{ks povzner}. 
Then define the set $\Omega  \subseteq L^1_2$
\begin{equation}\label{Omega}
	\Omega = \left\{ \F \in L^1_2: \ \F \geq 0, \  \mathfrak{m}^i_0[\F] = C^i_0, \ \mathfrak{m}_2[\F] = C_2, \
\mathfrak{m}_{\kss}[\F] \leq  \Cs  \right\}.
\end{equation}
The following theorem holds.
\begin{theorem}\label{cauchy-1}
	Let the collision kernel  satisfy assumptions stated in Section \ref{Sec: Assump Bij} with $\gamma_{ij}$ satisfying \eqref{gamma ij assumpt}--\eqref{gamma w}. Assume that $\F_0 \in \Omega$. Then the Cauchy problem \eqref{BS vector} has a unique solution in $\mathcal{C}([0,\infty), \Omega)  \cap \mathcal{C}^1((0,\infty), L^1_2)$.
\end{theorem}
{Our final result uses Theorem \ref{cauchy-1} to find solutions in the bigger space
\begin{equation}\label{Omega tilde}
	\Omega\subset\tilde\Omega = \left\{ \F \in L^1_2: \F\geq 0\,, 0< \mathfrak{m}^i_0[\F]  <\infty,  \,\mathfrak{m}_{(2 + \gm - \gw)^+}[\F] <\infty \right\}\subset L^1_2.
\end{equation}

\begin{theorem}\label{Th 2}
	Let the collision kernel  satisfy assumptions stated in Section \ref{Sec: Assump Bij} with $\gamma_{ij}$ satisfying \eqref{gamma ij assumpt}--\eqref{gamma w}. Assume that $\F_0 \in \tilde\Omega$. Then the Cauchy problem \eqref{BS vector} has a unique solution in $\mathcal{C}([0,\infty), \tilde\Omega)  \cap \mathcal{C}^1((0,\infty), L^1_2)$.
\end{theorem}
}
\begin{proof}[Proof of Theorem \ref{cauchy-1}]
	The goal is apply Theorem \ref{Theorem general} from the general ODE theory. For the collision operator $\Q$ understood  as  a mapping $\Q: \Omega \rightarrow L_{2}^1$, we will prove the following conditions for $\F, \G \in \Omega$,
\begin{enumerate}
	\item H\"{o}lder continuity condition
	\begin{equation}\label{Holder}
		\left\| \Q(\F) - \Q(\G) \right\|_{L^1_{2}} \leq C_H   \left\| \F-\G \right\|_{L^1_{2}}^{1/2},
	\end{equation}
	with the constant $C_H  = 6 \, \Cs^{3/2} \max_{1 \leq i, j  \leq P}(\kappa_{ij}^{ub})$,
	\item Sub-tangent condition
	\begin{equation*}
		\lim\limits_{h\rightarrow 0+} \frac{\text{dist}\left(\F + h \Q(\F), \Omega \right)}{h} =0, \quad \text{with} \quad 	\text{dist}\left(\G,\Omega\right)=\inf_{\omega \in \Omega} \left\| \G-\omega \right\|_{L_{2}^1}.
	\end{equation*}
	\item One-sided Lipschitz condition
	\begin{equation}\label{Lipschitz}
		\left[  \F- \G, \Q(\F) - \Q(\G) \right] \leq C_L  \left\| \F-\G \right\|_{L^1_{2}},
	\end{equation}
	with the constant $C_L = { 2} \, \Cs \,  \max_{1 \leq i, j  \leq P}(\kappa_{ij}^{ub})$, where  brackets are defined in \eqref{Lipschitz bra}.
\end{enumerate}
	
	\subsubsection*{H\"{o}lder continuity condition} 
	Firstly, we introduce some notation, namely,
	\begin{equation}\label{H L}
		\mathbb{H} =  \F-\G, \quad \mathbb{L}=  \F + \G.  
	\end{equation}
	Then, we notice that the bi-linear structure of all collision operators allows to write
	\begin{equation}\label{bi-linear}
		Q_{ij}(f_i, f_j) - Q_{ij}(g_i, g_j) =\frac{1}{2} \left(  Q_{ij}(h_i, l_j)  + Q_{ij}(l_i, h_j) \right), 
	\end{equation}
	for any possible combination of $i, j \in \left\{ 1, \dots, P \right\}$, and with the notation \eqref{H L}, \eqref{F Q vector}. Therefore, the left-hand side of \eqref{Holder} becomes 
	\begin{align}
	\mathcal{I}_H :=	\left\| \Q(\F) - \Q(\G) \right\|_{L^1_2} & \leq \sum_{i, j =1}^P \left\| Q_{ij}(f_i, f_j) - Q_{ij}(g_i, g_j)  \right\|_{L^1_{2,i}} \nonumber \\
		& \leq \frac{1}{2}  \sum_{i, j =1}^P  \left(    \left\| Q_{ij}(h_i, l_j)  \right\|_{L^1_{2,i}} +  \left\| Q_{ij}(l_i, h_j) \right\|_{L^1_{2,i}}  \right). \label{lhs holder}
	\end{align}
	For each combination of $i$ and $j$ we proceed separately. 
	\subsubsection*{Case (i): mono-mono interactions}  Let $i, j\in \left\{1, \dots, M\right\}$ and take generic real-valued functions $f$ and $g$. The norm of the collision operator \eqref{m-m Q} is estimated as follows
	\begin{multline}\label{pomocna 7}
		\left\| Q_{ij}(f, g)  \right\|_{L^1_{2,i}} \leq \int_{(\bR)^2} \int_{\bS} \left( |f(v')| |g(v'_*)|  +|f(v)| |g(v_*)|  \right) \la v \ra_i^2 \, \mathcal{B}_{ij} \, \md \sigma \, \md v_* \, \md v
		=: T_1 + T_2.
	\end{multline}
	For the first term $T_1$, coming from the gain part, for the weight we use the estimate \eqref{m-m upp bounds} and for the collision kernel  $  \mathcal{B}_{ij}:=\mathcal{B}_{ij}(v, v_*, \sigma)$ we use the microreversibility property \eqref{m-m cross} together with the assumption on its form \eqref{m-m ass B}, \eqref{m-m ass B tilde} and again \eqref{m-m upp bounds},
	\begin{equation*}
		\la v \ra_i^2 \, \mathcal{B}_{ij} \leq 	\la v' \ra_i^2 	\la v'_* \ra_j^2   \, \mathcal{B}'_{ij} \leq \la v' \ra_i^{2+\gamma_{ij}} 	\la v'_* \ra_j^{2+\gamma_{ij}}   \,  b_{ij}(\hat{u}'\cdot \sigma').
	\end{equation*}
	It remains to  interchange pre-post variables and obtain
	\begin{multline}\label{pomocna 8}
	T_1 \leq \int_{(\bR)^2} \int_{\bS}   |f(v)|   \la v \ra_i^{2+\gamma_{ij}}  |g(v_*)|	\la v_* \ra_j^{2+\gamma_{ij}}   \,  b_{ij}(\hat{u}\cdot \sigma)  \, \md \sigma \, \md v_* \, \md v 
	\\	= \nb \left\| f \right\|_{L^1_{2+ \gamma_{ij},i}} \left\| g \right\|_{L^1_{2+ \gamma_{ij},j}}.
	\end{multline}
	For the second term $T_2$ in \eqref{pomocna 7} coming from the loss term, we only make use of the collision kernel form \eqref{m-m ass B}, \eqref{m-m ass B tilde} together with the estimate \eqref{m-m upp bounds}, which together lead to 
	\begin{equation*}
		\mathcal{B}_{ij} \leq  \la v \ra_i^{\gamma_{ij}}  \la v_* \ra_j^{\gamma_{ij}}   b_{ij}(\hat{u}\cdot \sigma),   
	\end{equation*}
	and 
	\begin{equation}\label{pomocna 9}
		T_2 \leq \nb \left\| f \right\|_{L^1_{2+ \gamma_{ij},i}} \left\| g \right\|_{L^1_{\gamma_{ij},j}}.
	\end{equation}
	Therefore, gathering \eqref{pomocna 8} and \eqref{pomocna 9}, and using monotonicity of norms in \eqref{pomocna 9}, \eqref{pomocna 7} becomes
	\begin{equation*}
		\left\| Q_{ij}(f, g)  \right\|_{L^1_{2,i}} \leq 2  \nb \left\| f \right\|_{L^1_{2+ \gamma_{ij},i}} \left\| g \right\|_{L^1_{2+ \gamma_{ij},j}}.
	\end{equation*}
	For one part of the sum in \eqref{lhs holder}, this implies
	\begin{multline}\label{norm m-m}
		\frac{1}{2}  \sum_{i, j =1}^M  \left(    \left\| Q_{ij}(h_i, l_j)  \right\|_{L^1_{2,i}} +  \left\| Q_{ij}(l_i, h_j) \right\|_{L^1_{2,i}}  \right) 
		\\
		\leq   \sum_{i, j =1}^M   \kappa_{ij}^{ub} \left( \left\| h_i \right\|_{L^1_{2+ \gamma_{ij},i}} \left\| l_j \right\|_{L^1_{2+ \gamma_{ij},j}} + \left\| l_i \right\|_{L^1_{2+ \gamma_{ij},i}} \left\| h_j \right\|_{L^1_{2+ \gamma_{ij},j}} \right),
	\end{multline}
	in the light of \eqref{kappa norm b}. 
	\subsubsection*{Case (ii):  poly-mono \& mono-poly interactions} From the sum in \eqref{lhs holder}, we split the terms involving different types of indices, 
	\begin{multline}\label{pomocna 14}
		\frac{1}{2}  \sum_{i =1}^M \sum_{j =M+1}^P \left(    \left\| Q_{ij}(h_i, l_j)  \right\|_{L^1_{2,i}} 
		+  \left\| Q_{ij}(l_i, h_j) \right\|_{L^1_{2,i}}  \right) 
		\\
		+ 	\frac{1}{2}  \sum_{i =M+1}^P \sum_{j =1}^M  \left(    \left\| Q_{ij}(h_i, l_j)  \right\|_{L^1_{2,i}} +  \left\| Q_{ij}(l_i, h_j) \right\|_{L^1_{2,i}}  \right) 
		\\
		=	\frac{1}{2}   \sum_{i =M+1}^P \sum_{j =1}^M   \left(    \left\| Q_{ji}(h_j, l_i)  \right\|_{L^1_{2,j}} +  \left\| Q_{ji}(l_j, h_i) \right\|_{L^1_{2,j}}  + \left\| Q_{ij}(h_i, l_j)  \right\|_{L^1_{2,i}} +  \left\| Q_{ij}(l_i, h_j) \right\|_{L^1_{2,i}}  \right) 	
	\end{multline}
	For the first two terms, using definition of the collision operator \eqref{m-p coll operator} describing mono-poly interaction, for some $f=f(t,v,I)$ and $g=g(t,v)$, its $L^1_{2,j}$ norm can be estimates as follows,
	\begin{multline*}
		\left\| 	Q_{ji}(g,f)\right\|_{L^1_{2,j}}  \leq  \int_{ \bR} \int_{\bRfp} \int_{\bS \times [0,1]} \left\{ |g(w')| |f(w'_*,I'_*)| \left(\frac{I_*}{I'_*}\right)^{\alpha_i} +  |g(v)| |f(v_*,I_*)| \right\}
		\\ \times   \la v\ra^2_j \, \mathcal{B}_{ji}(v,v_*,I_*, \sigma,R) \, \dm \, \md \sigma \, \md R \, \md v_* \, \md I_* \, \md v.
	\end{multline*}
	Interchanging the collision reference using the transformation \eqref{interchange of particles} and exploring micro-reversibility of the collision kernel  \eqref{m-p micro rev},
	\begin{multline}\label{pomocna 12}
		\left\| 	Q_{ji}(g,f)\right\|_{L^1_{2,j}}  \leq  \int_{ \bR} \int_{\bRfp} \int_{\bS \times [0,1]} \left\{ |g(v'_*)| |f(v',I')| \left(\frac{I}{I'}\right)^{\alpha_i} +  |g(v_*)| |f(v, I)| \right\}
		\\ \times   \la v_* \ra^2_j \, \mathcal{B}_{ij}(v,v_*,I, \sigma,R) \, \dm \, \md \sigma \, \md R \, \md v_* \, \md I \, \md v =: T_3 + T_4.
	\end{multline}
	For the first term $T_3$, we additionally interchange pre- and post-quantities with transformation $\mathcal{T}_{pm}$ and using Lemma \ref{Lemma mono-poly invariance},  after the use of  estimates \eqref{p-m ass B}, \eqref{p-m ass B tilde}  on the collision kernel  $\mathcal{B}_{ij}:=\mathcal{B}_{ij}(v, v_*, I, \sigma, R)$ and on the weight \eqref{p-m upp bounds} that yield
	\begin{equation*}
		\la v_* \ra_j^2 \, 	\mathcal{B}_{ij}   \leq    \la v', I' \ra_i^{2} \, \la v'_* \ra_j^{2}    \	\mathcal{B}'_{ij}
		\\ \leq  \la v', I' \ra_i^{2+\gamma_{ij}} \,  \la v'_* \ra_j^{2+\gamma_{ij}}  \,  b_{ij}(\hat{u}'\cdot \sigma')  \,  \tilde{b}_{ij}^{ub}(R').
	\end{equation*}
	Therefore,
	\begin{multline}\label{pomocna 10}
T_3 
		\leq \int_{ \bR} \int_{\bRfp} \int_{\bS \times [0,1]}  |g(v_*)|  \la v_* \ra_j^{2+\gamma_{ij}}  |f(v, I)|  \la v, I \ra_i^{2+\gamma_{ij}} 
		\\ \times   b_{ij}(\hat{u}\cdot \sigma)  \,  \tilde{b}_{ij}^{ub}(R)   \, \dm \, \md \sigma \, \md R \, \md v_* \, \md I \, \md v = \kappa_{ij}^{ub} \left\| g \right\|_{L^1_{2+\gamma_{ij},j}}  \left\| f \right\|_{L^1_{2+\gamma_{ij}, i}}. 
	\end{multline}
	For the second term $T_4$, we only use the estimate on $\mathcal{B}_{ij}$ from \eqref{p-m ass B}, \eqref{p-m ass B tilde} and  \eqref{p-m upp bounds},
	\begin{equation*}
		\mathcal{B}_{ij} \leq  \la v, I \ra_i^{\gamma_{ij}} \,  \la v_* \ra_j^{\gamma_{ij}}  \, b_{ij}(\hat{u}\cdot \sigma) \, \tilde{b}_{ij}^{ub}(R),
	\end{equation*}
	implying
	\begin{multline}\label{pomocna 11}
T_4		\leq	   \int_{ \bR} \int_{\bRfp} \int_{\bS \times [0,1]}   |g(v_*)| \,   \la v_* \ra^{2+\gamma_{ij}}_j  \, |f(v, I)|  
		\,  \la v, I \ra_i^{\gamma_{ij}} 
		\\ \times b_{ij}(\hat{u}\cdot \sigma) \, \tilde{b}_{ij}^{ub}(R) \, \dm \, \md \sigma \, \md R \, \md v_* \, \md I \, \md v = \kappa_{ij}^{ub} \left\| g \right\|_{L^1_{2+\gamma_{ij},j}}  \left\| f \right\|_{L^1_{\gamma_{ij}, i}}. 
	\end{multline}
	Therefore, \eqref{pomocna 10} and \eqref{pomocna 11} imply for \eqref{pomocna 12}, after the use of  norm monotonicity \eqref{monotonicity of norm},
	\begin{equation}\label{pomocna 13}
		\left\| 	Q_{ji}(g,f)\right\|_{L^1_{2,j}}  \leq  2 \, \kappa_{ij}^{ub} \left\| g \right\|_{L^1_{2+\gamma_{ij},j}}  \left\| f \right\|_{L^1_{2+\gamma_{ij}, i}}. 
	\end{equation}
	
	On the other side, for the last two terms of \eqref{pomocna 14}, we use the collision operator \eqref{p-m coll operator} that describes poly-mono interaction, and estimate its $L^1_{2,i}$ norm for some $f=f(t,v,I)$ and $g=g(t,v)$,
	\begin{multline*}
		\left\| Q_{ij}(f, g)  \right\|_{L^1_{2,i}} \leq  \int_{ \bRfp} \int_{\bR} \int_{\bS \times [0,1]} \left\{ |f(v',I')| |g(v'_*)| \left(\frac{I}{I'}\right)^{\alpha_i}  +  |f(v,I)| |g(v_*)| \right\}
		\\ \times   \la v, I \ra^2_i \,\mathcal{B}_{ij}  \, \dm \, \md \sigma \, \md R \, \md v_*\, \md v \,  \md I.
	\end{multline*}
	Incorporating the same arguments as for the counterpart  term \eqref{pomocna 12} since the same collision kernel  is used with assumptions \eqref{p-m ass B}, \eqref{p-m ass B tilde} and the same upper bounds apply to the weight $\la v_* \ra_j$ in that context, and to the weight $\la v, I \ra_i$ in the present one, by  virtue of \eqref{p-m upp bounds}, yields
	\begin{equation}\label{pomocna 15}
		\left\| Q_{ij}(f, g)  \right\|_{L^1_{2,i}} \leq  2 \, \kappa_{ij}^{ub}  \left\| f \right\|_{L^1_{2+\gamma_{ij}, i}} \, \left\| g \right\|_{L^1_{2+\gamma_{ij},j}}.
	\end{equation}
	Finally, \eqref{pomocna 13} and \eqref{pomocna 15} allow to conclude for  \eqref{pomocna 14}, 
	\begin{multline}\label{norm p-m}
		\frac{1}{2}   \sum_{i =M+1}^P \sum_{j =1}^M   \left(    \left\| Q_{ji}(h_j, l_i)  \right\|_{L^1_{2,j}} +  \left\| Q_{ji}(l_j, h_i) \right\|_{L^1_{2,j}}  + \left\| Q_{ij}(h_i, l_j)  \right\|_{L^1_{2,i}} +  \left\| Q_{ij}(l_i, h_j) \right\|_{L^1_{2,i}}  \right) 	
		\\
		\leq  2  \sum_{i =M+1}^P \sum_{j =1}^M \kappa_{ij}^{ub}  \left(      \left\| h_i \right\|_{L^1_{2+\gamma_{ij}, i}} \, \left\| l_j \right\|_{L^1_{2+\gamma_{ij},j}}  +     \left\| l_i \right\|_{L^1_{2+\gamma_{ij}, i}} \, \left\| h_j \right\|_{L^1_{2+\gamma_{ij},j}}    \right),
	\end{multline}
with $\kappa_{ij}$ from \eqref{kappa p-m}.	
	\subsubsection*{Case (iii): poly-poly interactions} For interactions involving only polyatomic molecules, $i, j \in \left\{ M+1, \dots, P \right\}$ the $L^1_{2,i}$ norm of the corresponding collision operator \eqref{p-p coll operator} for some real-valued functions $f, g$ can be estimated as follows
\begin{multline*}
	\left\| Q_{ij}(f, g)  \right\|_{L^1_{2,i}}  \leq \int_{(\bRfp)^2} \int_{\bS \times [0,1]^2} \Big\{ |f(v',I')| |g(v'_*,I'_*)| \left(\frac{I }{I' }\right)^{\!\!\alpha_i } \! \left(\frac{ I_*}{ I'_*}\right)^{\!\!\alpha_j}  
	\\
	 + |f(v,I)| |g(v_*, I_*)| \Big\} \,
  \la v, I \ra_i^2 \ \mathcal{B}_{ij} \
  \dpo \, \md \sigma \, \md r \, \md R \, \md v_* \, \md I_* \, \md v \, \md I.
\end{multline*}	
Using the same strategy as in the previous paragraphs which involves the assumption on the collision kernel  \eqref{p-p ass B} and \eqref{p-p ass B tilde} together with the upper bounds \eqref{p-p upp bounds}, we obtain
\begin{equation*}
	\left\| Q_{ij}(f, g)  \right\|_{L^1_{2,i}} \leq 2 \, \kappa_{ij}^{ub}  \left\| f \right\|_{L^1_{2+\gamma_{ij}, i}} \, \left\| g \right\|_{L^1_{2+\gamma_{ij},j}},
\end{equation*}	
after exploiting the monotonicity of norms \eqref{monotonicity of norm}. Thus, the part of the sum in \eqref{lhs holder} related to only polyatomic interaction becomes
\begin{multline}\label{norm p-p}
		\frac{1}{2}  \sum_{i, j =M+1}^P  \left(    \left\| Q_{ij}(h_i, l_j)  \right\|_{L^1_{2,i}} +  \left\| Q_{ij}(l_i, h_j) \right\|_{L^1_{2,i}}  \right) 
	\\	\leq 
		 \sum_{i, j =M+1}^P \kappa_{ij}^{ub} \left(   \left\| h_i \right\|_{L^1_{2+\gamma_{ij}, i}} \, \left\| l_j \right\|_{L^1_{2+\gamma_{ij},j}} + \left\| l_i \right\|_{L^1_{2+\gamma_{ij}, i}} \, \left\| h_j \right\|_{L^1_{2+\gamma_{ij},j}}   \right).
\end{multline}
Summarizing \eqref{norm m-m}, \eqref{norm p-m} and \eqref{norm p-p}, the left-hand side of the H\"{o}lder condition \eqref{lhs holder} becomes
\begin{multline}\label{lhs holder 2}
\mathcal{I}_H  \leq \max_{1 \leq i, j  \leq P}(\kappa_{ij}^{ub}) 	 \sum_{i, j =1}^P   \left( \left\| h_i \right\|_{L^1_{2+ \gm,i}} \left\| l_j \right\|_{L^1_{2+ \gm,j}} + \left\| l_i \right\|_{L^1_{2+ \gm,i}} \left\| h_j \right\|_{L^1_{2+ \gm,j}} \right)
\\
= 2 \max_{1 \leq i, j  \leq P}(\kappa_{ij}^{ub}) 	  \ \left\| \mathbb{H} \right\|_{L^1_{2+ \gm}} \left\| \mathbb{L} \right\|_{L^1_{2+ \gm}},
\end{multline}		
with $\gm$ from \eqref{gamma w}. 
Interpolating the $(2+\gm)$-th norm of $\mathbb{H}$,
\begin{equation*}
\left\| \mathbb{H} \right\|_{L^1_{2+ \gm}} = \sum_{i =1}^P \left\|  h_i \right\|_{L^1_{2+ \gm, i}} \leq \sum_{i =1}^P \left\|  h_i \right\|_{L^1_{2+ 2 \gm, i}}^{1/2}   \left\|  h_i \right\|_{L^1_{2, i}}^{1/2}  \leq \left\| \mathbb{H} \right\|_{L^1_{2+ 2 \gm}}^{1/2}  \left\| \mathbb{H} \right\|_{L^1_{2}}^{1/2}.
\end{equation*}
Since $\F, \G$ are non-negative, it follows $|h_i| \leq f_i + g_i = l_i$, and therefore 
\begin{equation*}
 \ \left\| \mathbb{H} \right\|_{L^1_{2+ \gm}} \left\| \mathbb{L} \right\|_{L^1_{2+ \gm}} \leq \left\| \mathbb{L} \right\|_{L^1_{2+ 2 \gm}}^{3/2}  \left\| \mathbb{H} \right\|_{L^1_{2}}^{1/2},
\end{equation*}
by monotonicity of norms. Since  $\F, \G \in \Omega$, it implies
\begin{equation*}
	 \left\| \mathbb{L} \right\|_{L^1_{2+ 2 \gm}} \leq 	 \left\| \F \right\|_{L^1_{2+ 2 \gm}} +  \left\| \G \right\|_{L^1_{2+ 2 \gm}} \leq 2 \, \Cs.
\end{equation*}
This allows to finally conclude for \eqref{lhs holder}
	\begin{equation*}
	\left\| \Q(\F) - \Q(\G) \right\|_{L^1_2}  \leq 6 \, \Cs^{3/2} \max_{1 \leq i, j  \leq P}(\kappa_{ij}^{ub}) 	 \left\| \F - \G \right\|_{L^1_{2}}^{1/2},
\end{equation*}
which is exactly \eqref{Holder} with the constant as announced.

	\subsubsection*{Sub-tangent condition} 
	
As shown in \cite{IG-Alonso-BAMS},	the sub-tangent condition follows from the lemma below.
	
	\begin{lemma}\label{Lemma sub tangent}
		Fix $ \F \in \Omega$. Then for any $\eps > 0$ there exists $h_*>0$ such that the set $B$ centered at $\F+h \, \Q(\F)$ with radius $h \, \eps$ denoted by
	$
			B(f+h Q(f,f), h \epsilon),
	$
		has a non-empty intersection with $\Omega$ 
		\begin{equation}\label{ball sub tangent}
			B(\F+h \, \Q(\F), h \, \eps) \cap \Omega \neq \emptyset, \quad \text{for any} \ 0<h<h_*.
		\end{equation}	
	\end{lemma}
	\begin{proof}
		For   $ \rho>0$ to be determined later, we define sets depending on the index $i$,
		\begin{equation*}
			\begin{split}
	B_i(\rho) & = \left\{ v \in \bR:   \left| v \right|   \leq \rho  \right\}, \quad  i \in \left\{ 1, \dots, M \right\}, \\		
	B_i(\rho) &= \left\{ (v, I) \in \bRfp: \sqrt{   \left| v \right|^2 + \frac{I}{\mP} } \leq \rho  \right\}, \quad  i \in \left\{ M+ 1, \dots, P \right\}.
	\end{split}
		\end{equation*}
For the fixed $\F \in \Omega$, we define its truncation function	
\begin{equation}\label{F trun}
	\Fr = \left[ \begin{matrix}  \Big[ f_i(t, v) \, \mathds{1}_{B_i(\rho)}(v)  \Big]_{i=1,\dots,M} \\  \Big[ f_i(t, v, I) \, \mathds{1}_{B_i(\rho)}(v, I)  \Big]_{i=M+1,\dots,P} \end{matrix} \right].
\end{equation}
In the spirit of \cite{IG-Alonso-BAMS}, we construct the function  of the form 
\begin{equation}\label{Wrho}
\Wr =  \F + h\, \Q(\Fr), \qquad 0<h\leq 1,
\end{equation}	
such that for a certain choice of $\rho$ and $h$, it is an element of the intersection \eqref{ball sub tangent}. 

Firstly, we prove that $\Wr $ is non-negative on a certain interval for $h$. To that end, for $i \in \left\{1, \dots, M \right\}$, we define  the collision frequency and develop its upper bound, by virtue of the assumption on the collision kernels  stated in Section \ref{Sec: Assump Bij} and Lemma \ref{Lemma bounds on Bij tilde},
\begin{align}\label{coll fr m}
	\left[ \nu(\F) \right]_i & = \sum_{j =1}^M \int_{ \bR} \int_{\bS} f_j(v_*) \mathcal{B}_{ij} \, \md \sigma \, \md v_* \nonumber  \\ 
	& \qquad \qquad +  \sum_{j = M+1}^P \int_{ \bRfp} \int_{\bS\times [0,1]} f_j(v_*,I_*) \mathcal{B}_{ij}  \, \dm\, \md \sigma \, \md R \, \md v_* \, \md I_* \nonumber
	\\
&	\leq \sum_{j =1}^M  \kappa_{ij}^{ub}  \int_{ \bR}  f_j(v_*) \left(  \la v\ra_i^{\g} +  \la v_*\ra_j^{\g} \right) \,   \md v_*  \nonumber
	\\
& \qquad \qquad 	+  \sum_{j = M+1}^P \kappa_{ij}^{ub}   \int_{ \bRfp}  f_j(v_*,I_*)  \left(  \la v\ra_i^{\g} +  \la v_*, I_*\ra_j^{\g} \right)  \, \md v_* \, \md I_* \nonumber
	\\
&	\leq 
\kub \left(  \mathfrak{m}_0[\F]  \la v\ra_i^{\gm} +  \mathfrak{m}_{\gm}[\F]  \right) 	
	\leq 
\frac{K}{2}  \left(  1 + \la v\ra_i^{\gm} \right),
\end{align}	
with the constant 
\begin{equation}\label{K}
	K =2 \kub \, \mathfrak{m}_{\gm}[\F].
\end{equation}
Since for $v \in B_i(\rho)$ is $\la v\ra_i^2 \leq 1 + \rho^2$, the estimate \eqref{coll fr m} combined  with the truncated function \eqref{F trun} yields
\begin{equation}\label{Fr mono}
	\left[ \Fr \right]_i 	\left[ \nu(\F) \right]_i \leq K  \left(  1 + \rho^{\gm} \right) \, f_i(t, v), \quad i=1,\dots,M.
\end{equation}
For $i \in \left\{ M+1, \dots, P \right\}$, the  same computations lead to the upper bound for the collision frequency,
\begin{multline*}
	\left[ \nu(\F) \right]_i = \sum_{j =1}^M \int_{ \bR} \int_{\bS\times [0,1]} f_j(v_*) \, \mathcal{B}_{ij}  \, \dm \, \md \sigma \, \md R \, \md v_* \\ +  \sum_{j = M+1}^P \int_{ \bRfp} \int_{\bS\times [0,1]^2} f_j(v_*,I_*) \, \mathcal{B}_{ij}  \, \dpo \, \md \sigma\, \md r \, \md R \, \md v_* \, \md I_*
	\leq 
\frac{K}{2}  \left(  1 + \la v, I\ra_i^{\gm} \right),
\end{multline*}	
which combined with the	truncated function \eqref{F trun} implies, by $\la v, I \ra_i \leq 1 + \rho^2$ for $(v, I) \in B_i(\rho)$,
\begin{equation}\label{Fr poly}
	\left[ \Fr \right]_i 	\left[ \nu(\F) \right]_i \leq K  \left(  1 + \rho^{\gm} \right) \, f_i(t, v, I), \quad i=M+1, \dots, P.
\end{equation}
For $\Wr$,  \eqref{Fr mono} and \eqref{Fr poly} yield positivity of $\Wr$  for certain $h$, 
\begin{equation}\label{Wr positive}
	\Wr \geq   \F - h\, \Fr \, \nu(\Fr) \geq  \F - h\, \Fr \, \nu(\F)  \geq \F \left(  1 - h K  \left(  1 + \rho^{\gm} \right)   \right) \geq 0,
\end{equation}	
for the choice of $h$ as follows, defining $h_*$ from the lemma's statement,
\begin{equation}\label{h}
 0< h   \leq  \frac{1}{K  \left(  1 + \rho^{\gm} \right)} =: h_*.
\end{equation}	
Conservative properties of the collision operator $\Q$ imply
\begin{equation}\label{Wr cons}
\mathfrak{m}_{0}^i[\Wr] = \mathfrak{m}_{0}^i[\F], \qquad \mathfrak{m}_{2}[\Wr] = \mathfrak{m}_{2}[\F].
\end{equation}
In order to show that the boundness of $\mathfrak{m}_{\kss}[\Wr]$, the bound on $	\mathfrak{m}_k[\Q(\F)] $   is recalled \eqref{mk of Q},
\begin{equation*}
	\mathfrak{m}_k[\Q(\F)] 
	\leq  - 	{A}_\star \, \mathfrak{m}_2[\F]^{-\frac{\gw}{k-2}} \ \mathfrak{m}_{k}[\F]^{1+\frac{\gw}{k-2}} +  B_k =: \mathcal{L}_k(\mathfrak{m}_{k}[\F]), \qquad  k\geq\ks.
\end{equation*}
Incorporating arguments of \cite{IG-Alonso-BAMS}, the map $\mathcal{L}_k: [0, \infty) \rightarrow \mathbb{R}$ has only one root given in \eqref{Ek}  at which it changes from positive to negative and its maximal value is achieved at $B_k$, i.e. $\mathcal{L}_k(x) \leq   B_k$, for any $x \in   [0, \infty)$.
In particular, for $k=\kss$,  and assuming $\mathfrak{m}_{\kss}[\F]\leq \Eks$, for $\Wr$ we get
\begin{equation*}
\mathfrak{m}_{\kss}[\Wr] = \mathfrak{m}_{\kss}[\F] + h 	\, \mathfrak{m}_{\kss}[\Q(\Fr)] \leq \Eks + h\, B_{\kss} \leq \Eks +  B_{\kss} =: \mathfrak{h}_{\kss}    \leq \Cs.
\end{equation*}
Otherwise, in the case 	$\mathfrak{m}_{\kss}[\F] > \Eks$, we take sufficiently large $\rho$ to ensure
\begin{equation}\label{rho l e}
\mathfrak{m}_{\kss}[\F] > \Eks \ \Rightarrow \ \mathfrak{m}_{\kss}[\Fr] \geq  \Eks \quad \text{ implying } \quad 	\mathfrak{m}_{\kss}[\Q(\Fr)]  \leq \mathcal{L}_{\kss}(\mathfrak{m}_{\kss}[\Fr]) \leq 0, 
\end{equation}
leading to
\begin{equation*}
 \mathfrak{m}_{\kss}[\Wr] \leq  \mathfrak{m}_{\kss}[\F] \leq \Cs, \quad \text{since} \ \F \in \Omega.
\end{equation*}
Thus, we conclude $	\mathfrak{m}_{\kss}[\Wr]  \leq \Cs$, which together with \eqref{Wr positive} and \eqref{Wr cons} implies $\Wr \in \Omega$ for sufficiently large $\rho$ to have \eqref{rho l e} and $h$ as defined in \eqref{h}. \\

On the other side, for this element $\Wr \in \Omega$, by the H\"{o}lder property \eqref{Holder},
	\begin{equation}\label{sub tangent 2}
h^{-1}    \left\| \F + h  \, \Q(\F) - \Wr \right\|_{L_{2}^1} =     \left\|   \Q(\F) - \Q(\Fr) \right\|_{L_{2}^1} \leq C_H  \left\|   \F - \Fr \right\|_{L_{2}^1}^{1/2} \leq \eps,
\end{equation}
for sufficiently large $\rho$.\\

Therefore, we conclude that $\Wr$ as defined in \eqref{Wrho} for $\rho$ sufficiently large to ensure both \eqref{rho l e} and \eqref{sub tangent 2}  and the corresponding  $h$ from \eqref{h}    is an intersection element of $\Omega$ and $	B(\F+h \, \Q(\F), h \, \eps)$, proving \eqref{ball sub tangent} and concluding this lemma and the sub-tangent condition. 
	\end{proof}

\subsubsection*{One-sided Lipschitz condition } 	
Denote the vector valued  $\chi$ such that $\chi_i(\cdot) = \text{sign}\left( (f_i - g_i)(\cdot) \right) \la \cdot \ra_i^2$, where the argument is either $v$ for $i\in \left\{1, \dots,M \right\}$ or $(v, I)$ for $i\in \left\{M+1, \dots, P\right\}$. Definition of the brackets \eqref{Lipschitz bra} yields
\begin{multline*}
	\mathcal{I}_L :=	[\F-\G, \Q(\F) - \Q(\G)]  
	\\ = \sum_{i =1}^M \int_{ \bR} \left[ \Q(\F) - \Q(\G) \right]_i \,  \chi_i(v)^2 \md v +\sum_{i =M+1}^P \int_{ \bRfp} \left[ \Q(\F) - \Q(\G) \right]_i  \chi_i(v, I)^2 \md I.
\end{multline*}	
{The bi-linear structure of the collision operator \eqref{bi-linear}  }  with the notation  \eqref{H L} {together with } the weak form \eqref{weak form general} imply
\begin{multline}\label{pomocna 16}
	\mathcal{I}_L 
	=	\frac{1}{ 4} \sum_{i, j =1}^M   \int_{(\bR)^2} \int_{\bS }  \left( h_i(v) \, {\ell}_j(v_*) + {\ell}_i(v) h_j(v_*) \right)  	\Delta_{ij}(v, v_*) \, \mathcal{B}_{ij}  \, \md \sigma \,\md v_* \, \md v
	\\
	+ \frac{1}{ 4} \sum_{i, j = M + 1}^P    \int_{(\bRfp)^2} \int_{\bS \times [0,1]^2}  \left( h_i(v, I) \, {\ell}_j(v_*, I_*) + {\ell}_i(v, I) h_j(v_*, I_*) \right) 	\Delta_{ij}(v, I, v_*, I_*)   \\
	\\  \times \mathcal{B}_{ij}  \, \dpo \, \md \sigma \, \md r \, \md R \, \md v_* \, \md I_* \, \md v \, \md I
	\\
	+	{ \frac{1}{2}} \sum_{j=1}^M  \sum_{i=M+1}^P  \int_{\bR \times \bR \times \bRp} \int_{\bS \times [0,1]} \left( h_i(v, I) \, {\ell}_j(v_*) + {\ell}_i(v, I) h_j(v_*) \right)  \Delta_{ij}(v, I, v_*)   \\
	\times \mathcal{B}_{ij}  \, \dm \, \md \sigma \, \md R \, \md v_* \, \md v \, \md I,
\end{multline}
with notation 
\begin{equation*}
	\begin{split}
		& \Delta_{ij}(v, v_*) =  \chi_i(v') + \chi_j(v'_*) - \chi_i(v) - \chi_j(v_*), \quad i, j \in \left\{1, \dots, M \right\}, \\
		&	\Delta_{ij}(v, I, v_*) =  \chi_i(v', I') + \chi_j(v'_*) - \chi_i(v, I) - \chi_j(v_*), \quad j \in \left\{1, \dots, M \right\}, \ i \in \left\{ M+1, \dots, P \right\}, \\
		&	\Delta_{ij}(v, I, v_*, I_*) =  \chi_i(v', I') + \chi_j(v'_*, I'_*) - \chi_i(v, I) - \chi_j(v_*, I_*), \quad   i, j \in \left\{ M+1, \dots, P \right\}.
	\end{split}
\end{equation*}
First take $i, j \in \left\{1, \dots, M \right\}$.  Conservation law of the energy \eqref{m-m coll CL} implies the following bound
\begin{equation*}
	\begin{split}
		& h_i(v) \, \Delta_{ij}(v, v_*) \leq |h_i(v)| \left( \la v'\ra_i^2 + \la v'_*\ra_j^2 - \la v\ra_i^2 + \la v_*\ra_j^2  \right) = 2  |h_i(v)| \la v_*\ra_j^2,  \\
		& h_j(v_*) \, \Delta_{ij}(v, v_*) \leq |h_j(v_*)| \left( \la v'\ra_i^2 + \la v'_*\ra_j^2 + \la v\ra_i^2 - \la v_*\ra_j^2  \right) =  2  |h_j(v_*)| \la v \ra_i^2.
	\end{split}
\end{equation*}
The same computations are performed with brackets that include internal energy, for solely polyatomic interactions $i, j \in \left\{M+1, \dots, P \right\}$. For $ j \in \left\{1, \dots, M \right\}$ and $i  \in \left\{M+1, \dots, P \right\}$,
\begin{equation*}
	\begin{split}
		& h_i(v, I) \, \Delta_{ij}(v, I, v_*) \leq |h_i(v, I)| \left( \la v', I'\ra_i^2 + \la v'_*\ra_j^2 - \la v, I\ra_i^2 + \la v_*\ra_j^2  \right) = 2  |h_i(v, I)| \la v_*\ra_j^2,  \\
		& h_j(v_*) \, \Delta_{ij}(v, I, v_*) \leq |h_j(v_*)| \left( \la v', I'\ra_i^2 + \la v'_*\ra_j^2 + \la v, I\ra_i^2 - \la v_*\ra_j^2  \right) =  2  |h_j(v_*)| \la v, I \ra_i^2.
	\end{split}
\end{equation*}
These estimates imply for \eqref{pomocna 16},
\begin{multline*}
	\mathcal{I}_L \leq  {\frac{1}{2}} \sum_{i, j =1}^M   \int_{(\bR)^2} \int_{\bS }  \left(  |h_i(v)|  \, {\ell}_j(v_*) \la v_*\ra_j^2+ {\ell}_i(v) \la v \ra_i^2 |h_j(v_*)|  \right)  	  \mathcal{B}_{ij}  \, \md \sigma \,\md v_* \, \md v
	\\
	+ { \frac{1}{2}} \sum_{i, j = M + 1}^P    \int_{(\bRfp)^2} \int_{\bS \times [0,1]^2}  \left( |h_i(v,I)| \, {\ell}_j(v_*, I_*) \, \la v_*, I_*\ra_j^2 + {\ell}_i(v, I) \la v, I\ra_i^2 |h_j(v_*, I_*)| \right) 	
	\\  \times \mathcal{B}_{ij}  \, \dpo \, \md \sigma \, \md r \, \md R \, \md v_* \, \md I_* \, \md v \, \md I
	\\
	+	 \sum_{j=1}^M  \sum_{i=M+1}^P  \int_{\bR \times \bR \times \bRp} \int_{\bS \times [0,1]} \left( |h_i(v, I)| \, {\ell}_j(v_*) \la v_* \ra_j^2 + {\ell}_i(v, I) \la v, I \ra_i^2 |h_j(v_*)| \right)   \\
	\times \mathcal{B}_{ij} \, \dm \, \md \sigma \, \md R \, \md v_* \, \md v \, \md I.
\end{multline*}
Using upper bounds on the collision kernels  as stated in Section \ref{Sec: Assump Bij} and Lemma \ref{Lemma: upp bounds appendix}, together with the upper bounds on $\g$ from \eqref{gamma w},
\begin{align}
	\mathcal{I}_L &\leq  {\frac{1}{2}} \sum_{i, j =1}^M  \kappa_{ij}^{ub} \int_{(\bR)^2}    \left(  |h_i(v)|  \, {\ell}_j(v_*) \la v_*\ra_j^2+ {\ell}_i(v) \la v \ra_i^2 |h_j(v_*)|  \right) \la v\ra_i^{\g} \la v_* \ra_j^{\g} 	\,\md v_* \, \md v \nonumber
\nonumber	\\
& 	+ {\frac{1}{2}} \sum_{i, j = M + 1}^P   \kappa_{ij}^{ub}  \int_{(\bRfp)^2}   \left( |h_i(v,I)| \, {\ell}_j(v_*, I_*) \, \la v_*, I_*\ra_j^2 + {\ell}_i(v, I) \la v, I\ra_i^2 |h_j(v_*, I_*)| \right)
\nonumber	\\ 
	& \qquad\qquad\qquad\qquad\qquad\qquad\qquad  \times 	 
	\la v, I\ra_i^{\g} \la v_*, I_* \ra_j^{\g}  \, \md v_* \, \md I_* \, \md v \, \md I
\nonumber	\\
&	+	 \sum_{j=1}^M  \sum_{i=M+1}^P   \kappa_{ij}^{ub} 
	\int_{\bR \times \bR \times \bRp}   \left( |h_i(v, I)| \, {\ell}_j(v_*) \la v_* \ra_j^2 + {\ell}_i(v, I) \la v, I \ra_i^2 |h_j(v_*)| \right)   \\
	& \qquad\qquad\qquad\qquad\qquad\qquad\qquad  \times 	  \la v, I\ra_i^{\g} \la v_* \ra_j^{\g} \md v_* \, \md v \, \md I
\nonumber	\\
&	\leq  \max_{1 \leq i, j  \leq P}(\kappa_{ij}^{ub}) 
	\left( {\frac{1}{2}} \sum_{i, j =1}^M     \left(  \left\| h_i \right\|_{L^1_{\gm, i}}  \, \left\| {\ell}_j \right\|_{L^1_{2+\gm,j}}  + \left\| {\ell}_i \right\|_{L^1_{2+\gm, i}} \left\|h_j \right\|_{L^1_{\gm, j}}  
	\right)  
	 \right. \nonumber \\ & \left.
	+ {\frac{1}{2}} \sum_{i, j = M+1}^P     \left(  \left\| h_i \right\|_{L^1_{\gm, i}}  \, \left\| {\ell}_j \right\|_{L^1_{2+\gm,j}}  + \left\| {\ell}_i \right\|_{L^1_{2+\gm, i}} \left\|h_j \right\|_{L^1_{\gm, j}}  
	\right)   
	\right. \nonumber \\  & \left.
	+	 \sum_{j=1}^M  \sum_{i=M+1}^P   \left(  \left\| h_i \right\|_{L^1_{\gm, i}}  \, \left\| {\ell}_j \right\|_{L^1_{2+\gm,j}}  + \left\| {\ell}_i \right\|_{L^1_{2+\gm, i}} \left\|h_j \right\|_{L^1_{\gm, j}}  
	\right)   \right)
\nonumber	\\
	& = \max_{1 \leq i, j  \leq P}(\kappa_{ij}^{ub}) 
	\left(   \sum_{i, j =1}^M       \left\| h_i \right\|_{L^1_{\gm, i}}  \, \left\| {\ell}_j \right\|_{L^1_{2+\gm,j}}   
	+  \sum_{i, j = M+1}^P        \left\| h_i \right\|_{L^1_{\gm, i}}  \, \left\| {\ell}_j \right\|_{L^1_{2+\gm,j}}   
   \right. \nonumber \\  & \left.
	+	 \sum_{j=1}^M  \sum_{i=M+1}^P   \left(  \left\| h_i \right\|_{L^1_{\gm, i}}  \, \left\| {\ell}_j \right\|_{L^1_{2+\gm,j}}  + \left\| {\ell}_i \right\|_{L^1_{2+\gm, i}} \left\|h_j \right\|_{L^1_{\gm, j}}  
	\right)   \right)  =  \max_{1 \leq i, j  \leq P}(\kappa_{ij}^{ub})   \left\| \mathbb{L} \right\|_{L^1_{2+\gm}}  \left\| \mathbb{H} \right\|_{L^1_{\gm}} \label{Lipsch L H}
	\\ 
	&  \leq { 2 } \, \Cs \,  \max_{1 \leq i, j  \leq P}(\kappa_{ij}^{ub})  \,  \left\| \mathbb{H} \right\|_{L^1_{2}},
\label{Lipsch fin}
\end{align}
since $\F, \G \in \Omega$ and $\gm \leq 2$. This concludes the proof.
\end{proof}
{The idea of the proof of Theorem \ref{Th 2} is to use the fact that the Boltzmann operator is one-sided Lipschitz assuming only $(2 + \gm - \gw)^+$ moments, thus, an approximating sequence of solutions can be drawn from Theorem \ref{cauchy-1} and pass to the limit.  This follows a similar, but perhaps more direct, idea from \cite{mischler-wennberg} and can be found in detail in \cite[Lemma 23]{IG-Alonso-BAMS}.
\begin{proof}[Proof of Theorem \ref{Th 2}]
Let $0\leq \F_0\in\tilde\Omega$. 
Then, by density, there exits an approximating sequence   $\{\F^{j}_0\} \subset \Omega$ to $\F_0$, say strongly in $L^1(\tilde\Omega)$.  Now, following the computations performed for the one-sided Lipschitz condition {and the inequality \eqref{Lipsch L H}, } it holds that, for any $j\,,\,l\in \mathbb{N}$,
\begin{align}\label{Lipsc-*}
	\begin{split}
		\frac{\md}{\md t}\big\| \F^{j} - \F^{l} \big\|_{L^1_2}  &= {\left[  \F^{j}- \F^{l}, \Q(\F^{j}) - \Q(\F^{l}) \right]} \\
		&\leq  \max_{1 \leq i, j  \leq P}(\kappa_{ij}^{ub})  \, \big\| \F^j + \F^l \big\|_{L^1_{2+\gm}}  \,  \big\| \F^{j} - \F^{l} \big\|_{L^1_{2}}, 
	\end{split}
\end{align}
where $\F^j$ is the solution to  the Boltzmann problem associated to the initial condition $\F^j_0$, given in Theorem \ref{cauchy-1}  { and brackets are defined in \eqref{Lipschitz bra}.  }
Use moment interpolation { formula \eqref{mom interpolation} together with the H\"older inequality} to control the norm 
\begin{equation*}
\big\| \F^j + \F^l \big\|_{L^{1}_{2+\gm}} \leq \big\| \F^j + \F^l \big\|^{1-\theta}_{L^{1}_{2+k_1}}\big\| \F^j + \F^l \big\|^{\theta}_{L^{1}_{2+\gm+k_2}}\,,
\end{equation*}
choosing for any $\varepsilon\in(0,\gw)$ 
\begin{equation*}
k_1 = \gm - \gw + \varepsilon\,,\qquad k_2=\frac{\gw\,(\gm - \gw) + \varepsilon^2}{\varepsilon}\,, \qquad { \theta= \frac{\varepsilon (\gw - \varepsilon)}{\gw (\gm -\gw + \varepsilon)}}, 
\end{equation*}
and taking $\varepsilon$ sufficiently small such that $2 < 2+k_1 = (2 + \gm - \gw)^+$.  Consequently, it follows by propagation of moments, estimates \eqref{poly prop} or \eqref{prop below ks}, that
\begin{equation*}
\sup_{t\geq0}\big\| \F^j + \F^l \big\|_{L^{1}_{2+k_1}} \leq C_1\big(\mathfrak{m}_{2},\mathfrak{m}_{2+k_1}[\F_0]\big)\,.
\end{equation*}
In this inequality we used that $\max\big\{ \mathfrak{m}_{2+k_1}[\F^{j}_0], \mathfrak{m}_{2+k_1}[\F^{l}_0]\big\}\leq 2\,\mathfrak{m}_{2+k_1}[\F_0]$, taking $j$ and $l\in \mathbb{N}$ sufficiently large if necessary.  Also, this choice of $k_1$ and $k_2$ implies that
\begin{equation*}
\frac{\gm+k_2}{\gw}\,\theta = \frac{\gw^2-\varepsilon^2}{\gw^2}\,,
\end{equation*}
which leads, thanks to generation of moments (estimates \eqref{poly gen} or \eqref{Ek1}), that
\begin{equation*}
\big\| \F^j + \F^l \big\|^{\theta}_{L^{1}_{2+\gm+k_2}} \leq C_{2}\big(\mathfrak{m}_{2}\big)\Big(1 + t^{-\frac{\gw^2-\varepsilon^2}{\gw^2}}\Big)\,,
\end{equation*}
and consequently,
\begin{equation*}
\big\| \F^j + \F^l \big\|_{L^{1}_{2+\gm}} \leq C_{3}\big( \mathfrak{m}_{2},\mathfrak{m}_{2+k_1}[\F_0]\big) \Big(1 + t^{-\frac{\gw^2-\varepsilon^2}{\gw^2}}\Big)\,.
\end{equation*}
Then, performing integration in \eqref{Lipsc-*} it holds that
\begin{align*}
\big\| \F^{j}(t) - \F^{l}(t) \big\|_{L^1_2} &\leq \big\| \F^{j}(0) - \F^{l}(0)\|_{L^1_2}\exp\bigg( \max_{1 \leq i, j  \leq P}(\kappa_{ij}^{ub})\int_0^t \, \big\| \F^j + \F^l \big\|_{2+\gm} d\tau \bigg)\\
&\leq \big\| \F^{j}(0) - \F^{l}(0)\|_{L^1_2}\exp\bigg(\frac{\gw^2}{\varepsilon^2} \, \max_{1 \leq i, j  \leq P}(\kappa_{ij}^{ub})\, C_3\,\big(t + t^{\frac{\varepsilon^2}{\gw^2}}\big) \bigg)\,.
\end{align*}
Since $\{\F^{j}_0\}$ is Cauchy in $L^{1}_{2}$, previous estimate implies that $\{\F^{j}(t)\}$ is Cauchy in $L^{\infty}([0,T),L^{1}_{2})$ for any $T>0$ and, as such, there is a strong limit $\F$ for such sequence.  The theorem follows after passing to the limit since the limit $\F$ solves the mixture system with associated initial datum $\F_0$.  The fact that $\F\in\mathcal{C}([0,\infty), \tilde\Omega)  \cap \mathcal{C}^1((0,\infty), L^1_2)$ follows a standard argument.
\end{proof}}

\section*{Acknowledgments}
R. J. Alonso   was supported by the Conselho Nacional de Desenvolvimento Cient\'{i}fico e Tecnol\'{o}gico (CNPq), grant \emph{Bolsa de Produtividade em Pesquisa (303325/2019-4)}. I. M. Gamba was supported by USA grants from the National Science Foundation DMS-2009736, and the  Department of Energy grant  DE-SC0016283 on  \emph{Simulation Center for Runaway Electron Avoidance and Mitigation}.  M. Pavi\'c-\v Coli\'c gratefully acknowledges the support from the Science Fund of the Republic of Serbia, \emph{PROMIS, \#6066089, MaKiPol} and the Ministry of Education, Science and Technological Development of the Republic of Serbia  \emph{\#451-03-47/2023-01/200125}, as well as  from the \emph{Alexander von Humboldt Foundation}.

\subsection*{Data availability statement}
Data sharing not applicable to this article as no datasets were generated or analysed during the current study.

\appendix
\section{}\label{Sec: App}

\begin{lemma}[Bounds on the collision kernels]\label{Lemma bounds on Bij tilde}
	For the collision kernels  \eqref{m-m ass B tilde}, \eqref{p-m ass B tilde}, \eqref{p-p ass B tilde} the following estimates hold,
	\begin{alignat}{3}
		&  L_{ij} \la v \ra_i^{\gamma_{ij} }  -  \la v_* \ra_j^{\gamma_{ij}}   \leq  \tilde{\mathcal{B}}_{ij}(v, v_*) \leq  \la v \ra_i^{\gamma_{ij} }   + \la v_* \ra_j^{\gamma_{ij}}, \ i,j \in \left\{ 1, \dots, M \right\}  \label{m-m bounds on Bij tilde} \\
		& L_{ij} \la v, I \ra_i^{\gamma_{ij} }   - \la v_* \ra_j^{\gamma_{ij} }  \leq \tilde{\mathcal{B}}_{ij}(v, v_*, I) \leq  \la v, I \ra_i^{\gamma_{ij} }   + \la v_* \ra_j^{\gamma_{ij} },  \  j \in \left\{1,...,M\right\},  i \in \left\{ M+1,...,P \right\}  \label{p-m bounds on Bij tilde} \\
		& L_{ij} \la v, I \ra_i^{\gamma_{ij} }  -  \la v_*, I_* \ra_j^{\gamma_{ij}} \leq  \tilde{\mathcal{B}}_{ij}(v, v_*, I, I_*) \leq  \la v, I \ra_i^{\gamma_{ij} }   + \la v_*, I_* \ra_j^{\gamma_{ij}},  \ i, j \in \left\{ M+1, \dots, P \right\} \label{p-p bounds on Bij tilde}
	\end{alignat}
	where the involved constant $L_{ij}$  does not depend on the nature of interactions, but only on mass species and the rate $\gamma_{ij}$,
	\begin{equation}\label{L1 L2}
		L_{ij} = \left(\frac{\bs}{2}\right)^{\gamma_{ij}/2} \min\left\{1, 2^{1-\gamma_{ij}} \right\},
	\end{equation}
	with $\bs$ as defined in \eqref{sij min}. 
\end{lemma}
\begin{proof} Let $i, j \in \left\{ M+1, \dots, P \right\}$.
	\subsubsection*{Upper bound}	 
	Since the total energy in center-of-mass framework is only a part of the total energy in particles' framework, due to 
	\begin{equation*}
		\frac{m_i + m_j}{2  } |V|^2 + E_{ij} =\frac{m_i}{2} |v|^2 + I +  \frac{m_j}{2} |v_*|^2 + I_*,
	\end{equation*}
	with $E_{ij}$ from  \eqref{p-p ass B tilde}, it follows 
	\begin{equation}\label{Eij Etot}
		\frac{E_{ij}}{\mP}	 \leq   \frac{m_i}{2 \mP} |v|^2 + \frac{I}{\mP } + \frac{m_j}{2 \mP} |v_*|^2  + \frac{I_*}{\mP}.  
	\end{equation}
	Since $\gamma_{ij}/2 \leq 1$, 
	\begin{equation*}
		\left( \frac{E_{ij}}{\mP} \right)^{\gamma_{ij}/2}  \leq  \la v, I \ra_i^{\gamma_{ij} }   + \la v_*, I_* \ra_j^{\gamma_{ij} },
	\end{equation*}
	as stated in \eqref{p-p bounds on Bij tilde}.
	\subsubsection*{Lower bound}	
	Triangle inequality $| v- v_* | \geq |v| - |v_*|$ and few rearrangements imply
	\begin{equation*}
		\begin{split}
			\sqrt{\frac{E_{ij}}{\mP}} & \geq \frac{1}{\sqrt{2}} \left( \sqrt{\frac{\mu_{ij}}{2\mP}}  |v-v_*| + \sqrt{\frac{I}{\mP}} \right) \geq \sqrt{ \frac{m_j}{2(m_i+m_j)} } \left(  \sqrt{\frac{m_i}{2 \mP}} |v|    + \sqrt{\frac{I}{\mP}} \right) - \frac{1}{\sqrt{2}} \sqrt{\frac{m_j}{2 \mP}} |v_*|  \\
			&\geq \sqrt{ \frac{m_j}{2(m_i+m_j)} } \left( 1 + \sqrt{\frac{m_i}{2 \mP} |v|^2  + \frac{I}{\mP}  } \right) - \frac{1}{\sqrt{2}} \left( 1 + \sqrt{\frac{m_j}{2 \mP} |v_*|^2  + \frac{I_*}{\mP}  } \right) \\
			& \geq \sqrt{\frac{\bs}2} \, \la v, I \ra_i -   \la v_*, I_* \ra_j.
		\end{split}
	\end{equation*}
	Thus,
	\begin{equation*}
		\begin{split}
			\sqrt{\frac{\bs}2} \, \la v, I \ra_i \leq 	\sqrt{\frac{E_{ij}}{\mP}} +   \la v_*, I_* \ra_j.
		\end{split}
	\end{equation*}
	Taking the last inequality to the power $\gamma_{ij}$, 
	\begin{equation*}
		\begin{split}
			\left(\frac{\bs}{2}\right)^{\gamma_{ij}/2}  \, \la v, I \ra_i^{\gamma_{ij}} \leq \max\left\{ 1, 2^{\gamma_{ij}-1} \right\}	\left( 	\left( \frac{E_{ij}}{\mP} \right)^{\gamma_{ij}/2} +   \la v_*, I_* \ra_j^{\gamma_{ij}} \right),
		\end{split}
	\end{equation*}
	implying 
	\begin{equation*}
		\begin{split}
			\left( \frac{E_{ij}}{\mP} \right)^{\gamma_{ij}/2} \geq 		\left(\frac{\bs}{2}\right)^{\gamma_{ij}/2} \min\left\{ 1, 2^{1-\gamma_{ij}} \right\}  \, \la v, I \ra_i^{\gamma_{ij}} -    \la v_*, I_* \ra_j^{\gamma_{ij}},
		\end{split}
	\end{equation*}
	proving \eqref{p-p bounds on Bij tilde}.\\
	
	For other types of interactions, similar computations can be performed.
\end{proof}

\begin{lemma}\label{Lemma: upp bounds appendix}
	Let the energy in the center-of-mass framework $E_{ij}$ be \eqref{m-m en}, \eqref{p-m en} and \eqref{p-p en} depending on different combinations of $i, j \in \left\{1, \dots, P\right\}$ or nature of particles' interaction. The following estimates hold,
	\begin{alignat}{3}
		&   \frac{E_{ij}}{\mP} \leq  \la v \ra_i^2 \,  \la v_* \ra_j^2, \quad  \la v \ra_i  \leq  \la v' \ra_i \,  \la v'_* \ra_j, \quad i,j \in \left\{ 1, \dots, M \right\},  \label{m-m upp bounds} \\
		& \frac{E_{ij}}{\mP} \leq  \la v, I \ra_i^2 \,  \la v_* \ra_j^2,  \quad   \la v, I \ra_i, \ \la v_* \ra_j \leq  \la v', I' \ra_i \, \la v'_* \ra_j, \quad   j \in \left\{1,...,M\right\}, \ i \in \left\{ M+1,...,P \right\},  \label{p-m upp bounds} \\
		&  \frac{E_{ij}}{\mP} \leq  \la v, I \ra_i^2\, \la v_*, I_* \ra_j^2, \quad  \la v, I \ra_i \leq \la v', I' \ra_i\, \la v'_*, I'_* \ra_j,  \quad i, j \in \left\{ M+1, \dots, P \right\}. \label{p-p upp bounds}
	\end{alignat}
\end{lemma}
\begin{proof}
	Let	 $j \in \left\{1,...,M\right\}$,  $i \in \left\{ M+1,...,P \right\}$.
	From  \eqref{p-m en} it follows
	\begin{equation*}
		\frac{E_{ij}}{\mP} =\frac{1}{\mP} \left( \frac{\mu_{ij}}{2} |u|^2 + I \right) \leq \frac{m_i}{2\, \mP} |v|^2 + \frac{I}{\mP} +  \frac{m_j}{2 \, \mP} |v_*|^2  \leq   \la v, I \ra_i^2 \, \la v_* \ra_j^2.
	\end{equation*}  
	The conservation law of the energy \eqref{p-m coll CL} implies
	\begin{equation*}
		\frac{m_i}{2} |v|^2 + I, \ \frac{m_j}{2} |v_*|^2  \leq  \frac{m_i}{2} |v'|^2 + I' + \frac{m_j}{2} |v'_*|^2, 
	\end{equation*}
	which yields
	\begin{equation*}
		\la v, I \ra_i^2, \ \la v_* \ra_j^2 \leq  \la v', I' \ra_i^2 \, \la v'_* \ra_j^2.
	\end{equation*}
	Taking the square root, we get \eqref{p-m upp bounds}. \\
	
	For other combinations of $i$ and $j$ similar computations can be performed to conclude the proof.
\end{proof}

\begin{lemma}[ODI's Comparison Lemma for Moments Generation, \cite{IG-Alonso-BAMS}] \label{Lemma: comparison}
Let  $A$, $B$ and $c$ be positive constants and consider a function $y(t)$ which is absolute continuous in $t \in (0, \infty)$ and satisfies
\begin{equation*}
	y'(t) \leq B - A \, y(t)^{1+c}, \quad t>0.
\end{equation*}
Then
\begin{equation*}
	y(t) \leq z(t) := E \left(1 + \frac{K}{t^\beta}\right), \quad t>0,
\end{equation*}
for the choice
\begin{equation*}
	E = \left( \frac{B}{A} \right)^{\frac{1}{1+c}}, \quad \beta=\frac{1}{c}, \quad K= \left( c A \right)^{-1/c} E^{-1}.
\end{equation*}
\end{lemma}

\begin{proof} 
We start noticing that
\begin{equation*}
z'(t) = -\beta E K t^{-(\beta+1)}\,,
\end{equation*}
and that
\begin{equation*}
B-A\,z(t)^{1+c} = B - AE^{1+c}\bigg(1+\frac{K}{t^\beta}\bigg)^{1+c}\leq B - AE^{1+c}\bigg(1+\frac{K^{1+c}}{t^{\beta(1+c)}}\bigg)\,.
\end{equation*}
For the latter we invoked the binomial inequality  $( 1+ \frac{K}{t^\beta} )^{1+c} \geq  1+ (\frac{K}{t^\beta})^{1+c}$.  With the choice of $E,\,\beta$ and $K$ as suggested in the statement of the lemma it follows then 
\begin{equation}\label{ub}
B-A\,z(t)^{1+c}(t) \leq z'(t)\,,\qquad t>0\,.
\end{equation}
Next, note that for any $\epsilon>0$ the translation $z_\epsilon(t):=z(t-\epsilon)$ satisfies \eqref{ub} for $t>\epsilon$ due to the time invariance of the inequality.  Moreover, since $y(t)$ is absolute continuous in $t\in[\epsilon,\infty)$, there exists $\delta_*>0$ such that $z_{\epsilon}(\epsilon+\delta)\geq y(\epsilon+\delta)$ for any $\delta\in(0,\delta_*]$.  Consequently, we conclude the following setting for any $\epsilon>0$ and $\delta\in(0,\delta_*]$
\begin{align*}
\begin{cases}
&y'(t) \leq B  -  A\,y^{1 + c}(t)\,,\qquad  t>\epsilon+\delta\,,\\
&z'_\epsilon(t) \ge B  -  A\, z_\epsilon^{1 + c}(t)\,,\qquad  t>\epsilon+\delta\,,\\
&+\infty >z_\epsilon(\epsilon+\delta) \geq y(\epsilon+\delta)\,. 
\end{cases}
\end{align*}
Since both $z_\epsilon(t)$ and $y(t)$ are absolutely continuous in $t\in[\epsilon+\delta,\infty)$, we can use a standard comparison in ode to conclude that $z_\epsilon(t)\geq y(t)$ for $t>\epsilon+\delta$, valid for any $\epsilon>0$ and $\delta\in(0,\delta_*]$.  Thus, sending first $\delta$ to zero and then $\epsilon$ to zero it follows that $z(t)\geq y(t)$ for any $t>0$ which is the statement of the lemma.
\end{proof}

\begin{theorem}[Existence and Uniqueness Theory for ODE in Banach spaces]
	\label{Theorem general}
	Let $E:=(E,\left\| \cdot \right\|)$ be a Banach space, $\mathcal{S}$ be a bounded, convex and closed subset of $E$, and $\mathcal{Q}:\mathcal{S}\rightarrow E$ be an operator satisfying the following properties:
	\begin{itemize}
		\item[(a)] H\"{o}lder continuity condition
		\begin{equation*}
			\left\| \mathcal{Q}[u] - \mathcal{Q}[v] \right\| \leq C \left\| u-v \right\|^{\beta}, \ \beta \in (0,1), \ \forall u, v \in \mathcal{S};
		\end{equation*}
		\item[(b)] Sub-tangent condition
		\begin{equation*}
			\lim\limits_{h\rightarrow 0+} \frac{\text{dist}\left(u + h \mathcal{Q}[u], \mathcal{S} \right)}{h} =0, \ \forall u \in \mathcal{S};
		\end{equation*}
		\item[(c)] One-sided Lipschitz condition
		\begin{equation*}  
			\left[ \mathcal{Q}[u] - \mathcal{Q}[v], u - v \right] \leq C \left\| u-v \right\|, \ \forall u, v \in \mathcal{S},
		\end{equation*}
		where 	$\left[\varphi,\phi\right]=\lim_{h\rightarrow 0^-} h^{-1}\left(\left\| \phi + h \varphi \right\| - \left\| \phi \right\| \right)$.
	\end{itemize}
	Then the equation
	\begin{equation*}
		\begin{split}
			\partial_t  u = \mathcal{Q}[u], \ \text{for} \ t\in(0,\infty), \ \text{with initial data}  \
			u(0)= u_0 \ \text{in} \ \mathcal{S},
		\end{split}
	\end{equation*}
	has a unique solution in $C([0,\infty),\mathcal{S})\cap C^1((0,\infty),E)$.
\end{theorem}
The proof of this Theorem on ODE flows on Banach spaces  can be found in \cite{IG-Alonso-BAMS}.  	As pointed out in \cite{IG-Alonso-BAMS}, for $E:=L_{2}^1$, the Lipschitz brackets are
\begin{equation}\label{Lipschitz bra}
	[\phi, \psi] = \sum_{i =1}^M \int_{ \bR} \psi_{i}(v) \,  \text{sign}(\phi_i(v)) \la v \ra_i^2 \md v +\sum_{i =M+1}^P \int_{ \bRfp} \psi_{i}(v,I) \,  \text{sign}(\phi_i(v,I)) \la v,I \ra_i^2 \md v \, \md I.
\end{equation}

\bibliography{Alonso-Gamba-Pavic}

\providecommand{\noopsort}[1]{}\providecommand{\singleletter}[1]{#1}%
\begin{thebibliography}{10}

\bibitem{Alonso-cooling}
R.~Alonso, V.~Bagland, Y.~Cheng, and B.~Lods.
\newblock One-dimensional dissipative {B}oltzmann equation: measure solutions,
  cooling rate, and self-similar profile.
\newblock {\em SIAM J. Math. Anal.}, 50(1):1278--1321, 2018.

\bibitem{IG-Alonso-BAMS}
R.~J. Alonso and I.~M. Gamba.
\newblock Revisiting the {C}auchy problem for the {B}oltzmann equation for hard
  potentials with integrable cross section: from generation of moments to
  propagation of {$L^{\infty}$} bounds.
\newblock {\em ArXiv. 2211.09188}, 2022.

\bibitem{Alonso-Lods-Granular}
R.~J. Alonso and B.~Lods.
\newblock Free cooling and high-energy tails of granular gases with variable
  restitution coefficient.
\newblock {\em SIAM J. Math. Anal.}, 42(6):2499--2538, 2010.

\bibitem{Alonso-Orf}
R.~J. Alonso and H.~Orf.
\newblock Statistical moments and integrability properties of monatomic gas
  mixtures with long range interactions.
\newblock {\em ArXiv: 2204.09160}, 2022.

\bibitem{LD-ch-ens}
C.~Baranger, M.~Bisi, S.~Brull, and L.~Desvillettes.
\newblock On the {C}hapman-{E}nskog asymptotics for a mixture of monoatomic and
  polyatomic rarefied gases.
\newblock {\em Kinet. Relat. Models}, 11(4):821--858, 2018.

\bibitem{nicl-3}
N.~Bernhoff.
\newblock Compactness property of the linearized {B}oltzmann collision operator
  for a mixture of monatomic and polyatomic species.
\newblock ArXiv: 2201.01377, 2022.

\bibitem{nicl-1}
N.~Bernhoff.
\newblock Linearized {B}oltzmann {C}ollision {O}perator: {I}. {P}olyatomic
  {M}olecules {M}odeled by a {D}iscrete {I}nternal {E}nergy {V}ariable and
  {M}ulticomponent {M}ixtures.
\newblock {\em Acta Appl. Math.}, 183:3, 2023.

\bibitem{nicl-2}
N.~Bernhoff.
\newblock Linearized boltzmann collision operator: Ii. polyatomic molecules
  modeled by a continuous internal energy variable, 2023.

\bibitem{Bob-Moment-ineq}
A.~V. Bobylev.
\newblock Moment inequalities for the {B}oltzmann equation and applications to
  spatially homogeneous problems.
\newblock {\em J. Statist. Phys.}, 88(5-6):1183--1214, 1997.

\bibitem{IG-Bob-Panf-inelastic}
A.~V. Bobylev, I.~M. Gamba, and V.~A. Panferov.
\newblock Moment inequalities and high-energy tails for {B}oltzmann equations
  with inelastic interactions.
\newblock {\em J. Statist. Phys.}, 116(5-6):1651--1682, 2004.

\bibitem{bon-bou-bri-gre}
A.~Bondesan, L.~Boudin, M.~Briant, and B.~Grec.
\newblock Stability of the spectral gap for the {B}oltzmann multi-species
  operator linearized around non-equilibrium {M}axwell distributions.
\newblock {\em Commun. Pure Appl. Anal.}, 19(5):2549--2573, 2020.

\bibitem{Bor-Larsen}
C.~Borgnakke and P.~S. Larsen.
\newblock Statistical collision model for {M}onte {C}arlo simulation of
  polyatomic gas mixture.
\newblock {\em J. Comput. Phys.}, 18(4):405--420, 1975.

\bibitem{Bisi-Borsoni-Groppi}
T.~Borsoni, M.~Bisi, and M.~Groppi.
\newblock A general framework for the kinetic modelling of polyatomic gases.
\newblock {\em Commun. Math. Phys.}, 2022.

\bibitem{Bors-Bou-Salv-comp}
T.~Borsoni, L.~Boudin, and F.~Salvarani.
\newblock Compactness property of the linearized {B}oltzmann operator for a
  polyatomic gas undergoing resonant collisions.
\newblock ArXiv e-print 2204.10551, 2022.

\bibitem{Pavan-B-G}
L.~Boudin, B.~Grec, and V.~Pavan.
\newblock The {M}axwell-{S}tefan diffusion limit for a kinetic model of
  mixtures with general cross sections.
\newblock {\em Nonlinear Anal.}, 159:40--61, 2017.

\bibitem{nous}
L.~Boudin, B.~Grec, M.~Pavi\'c, and F.~Salvarani.
\newblock Diffusion asymptotics of a kinetic model for gaseous mixtures.
\newblock {\em Kinet. Relat. Models}, 6(1):137--157, 2013.

\bibitem{nous-LY}
L.~Boudin, B.~Grec, M.~Pavi\'c-\v{C}oli\'c, and S.~Simi\'{c}.
\newblock Energy method for the {B}oltzmann equation of monatomic gaseous
  mixtures.
\newblock ArXiv e-print 2110.07213, 2022.

\bibitem{Bou-Salv-res}
L.~Boudin, A.~Rossi, and F.~Salvarani.
\newblock A kinetic model of polyatomic gas with resonant collisions.
\newblock {H}{A}{L} preprint hal-03629556, 2022.

\bibitem{LD-Bourgat}
J.-F. Bourgat, L.~Desvillettes, P.~Le~Tallec, and B.~Perthame.
\newblock Microreversible collisions for polyatomic gases and {B}oltzmann's
  theorem.
\newblock {\em European J. Mech. B Fluids}, 13(2):237--254, 1994.

\bibitem{Bressan}
A.~Bressan.
\newblock Notes on the {B}oltzmann equation, {L}ecture notes for a summer
  course, {S.I.S.S.A.}, 2005.

\bibitem{bri}
M.~Briant.
\newblock Stability of global equilibrium for the multi-species {B}oltzmann
  equation in {$L^\infty$} settings.
\newblock {\em Discrete Contin. Dyn. Syst.}, 36(12):6669--6688, 2016.

\bibitem{bri-dau}
M.~Briant and E.~S. Daus.
\newblock The {B}oltzmann equation for a multi-species mixture close to global
  equilibrium.
\newblock {\em Arch. Ration. Mech. Anal.}, 222(3):1367--1443, 2016.

\bibitem{Brull-Comp}
S.~Brull, M.~Shahine, and P.~Thieullen.
\newblock Compactness property of the linearized boltzmann operator for a
  diatomic single gas model.
\newblock {\em Networks and Heterogeneous Media}, 2022.

\bibitem{Brull-Comp-2}
S.~Brull, M.~Shahine, and P.~Thieullen.
\newblock Fredholm property of the linearized boltzmann operator for a
  polyatomic single gas model, 2022.

\bibitem{dau-jun-mou-zam}
E.~S. Daus, A.~J\"ungel, C.~Mouhot, and N.~Zamponi.
\newblock Hypocoercivity for a linearized multispecies {B}oltzmann system.
\newblock {\em SIAM J. Math. Anal.}, 48(1):538--568, 2016.

\bibitem{Erica-pspde}
E.~de~la Canal, I.~M.~Gamba, and M.~Pavić-Čolić.
\newblock {O}n {E}xistence, {U}niqueness and {B}anach {S}pace {R}egularity
  for {S}olutions of {B}oltzmann {E}quations {S}ystems for {M}onatomic {G}as
  {M}ixtures.
\newblock volume 352, page 99 – 121, 2021.

\bibitem{LD-mom-meth}
L.~Desvillettes.
\newblock Some applications of the method of moments for the homogeneous
  {B}oltzmann and {K}ac equations.
\newblock {\em Arch. Rational Mech. Anal.}, 123(4):387--404, 1993.

\bibitem{LD-Toulouse}
L.~Desvillettes.
\newblock Sur un mod\`{e}le de type {B}orgnakke–{L}arsen conduisant \`{a} des
  lois d’energie non-lin\'{e}aires en temp\'{e}rature pour les gaz parfaits
  polyatomiques.
\newblock {\em Ann. Fac. Sci. Toulouse Math.}, 6(0):257--262, 1997.

\bibitem{DesMonSalv}
L.~Desvillettes, R.~Monaco, and F.~Salvarani.
\newblock A kinetic model allowing to obtain the energy law of polytropic gases
  in the presence of chemical reactions.
\newblock {\em Eur. J. Mech. B Fluids}, 24(2):219--236, 2005.

\bibitem{MPC-Dj-T-O}
V.~Djordji\'{c}, G.~Oblapenko, M.~Pavi\'{c}-\v{C}oli\'{c}, and M.~Torrilhon.
\newblock Boltzmann collision operator for polyatomic gases in agreement with
  experimental data and {DSMC} method.
\newblock {\em Contin. Mech. Thermodyn.}, 35(1):103--119, 2023.

\bibitem{MPC-Dj-S}
V.~Djordji\'{c}, M.~Pavi\'{c}-\v{C}oli\'{c}, and N.~Spasojevi\'{c}.
\newblock Polytropic gas modelling at kinetic and macroscopic levels.
\newblock {\em Kinet. Relat. Models}, 14(3):483--522, 2021.

\bibitem{MPC-Dj-T}
V.~Djordji\'{c}, M.~Pavi\'{c}-\v{C}oli\'{c}, and M.~Torrilhon.
\newblock Consistent, {E}xplicit and {A}ccessible {B}oltzmann {C}ollision
  {O}perator for {P}olyatomic {G}ases.
\newblock {\em Phys. Rev. E}, 104:025309, 2021.

\bibitem{Duan-Li-poly}
R.~Duan and Z.~Li.
\newblock Global bounded solutions to the {B}oltzmann equation for a polyatomic
  gas, 2022.

\bibitem{IG-Panf-Vill}
I.~M. Gamba, V.~Panferov, and C.~Villani.
\newblock Upper {M}axwellian bounds for the spatially homogeneous {B}oltzmann
  equation.
\newblock {\em Arch. Ration. Mech. Anal.}, 194(1):253--282, 2009.

\bibitem{MPC-IG-ARMA}
I.~M. Gamba and M.~Pavi\'{c}-\v{C}oli\'{c}.
\newblock On existence and uniqueness to homogeneous {B}oltzmann flows of
  monatomic gas mixtures.
\newblock {\em Arch. Ration. Mech. Anal.}, 235(1):723--781, 2020.

\bibitem{MPC-IG-poly}
I.~M. Gamba and M.~Pavi\'{c}-\v{C}oli\'{c}.
\newblock On the {C}auchy problem for {B}oltzmann equation modeling a
  polyatomic gas.
\newblock {\em J. Math. Phys.}, 64:013303, 2023.

\bibitem{IG-wave-turbulence}
I.~M. Gamba, L.~M. Smith, and M.-B. Tran.
\newblock On the wave turbulence theory for stratified flows in the ocean.
\newblock {\em Math. Models Methods Appl. Sci.}, 30(1):105--137, 2020.

\bibitem{Gio}
V.~Giovangigli.
\newblock {\em Multicomponent flow modeling}.
\newblock Modeling and Simulation in Science, Engineering and Technology.
  Birkh\"{a}user Boston, Inc., Boston, MA, 1999.

\bibitem{Groppi-Spiga}
M.~Groppi and G.~Spiga.
\newblock Kinetic approach to chemical reactions and inelastic transitions in a
  rarefied gas.
\newblock {\em J. Math. Chem.}, 26:197--219, 1999.

\bibitem{Lu-Mouhot-1}
X.~Lu and C.~Mouhot.
\newblock On measure solutions of the {B}oltzmann equation, part {I}: moment
  production and stability estimates.
\newblock {\em J. Differential Equations}, 252(4):3305--3363, 2012.

\bibitem{Martin-ODE}
R.~H. Martin, Jr.
\newblock {\em Nonlinear operators and differential equations in {B}anach
  spaces}.
\newblock Pure and Applied Mathematics. Wiley-Interscience [John Wiley \&
  Sons], New York-London-Sydney, 1976.

\bibitem{mischler-wennberg}
S.~Mischler and B.~Wennberg.
\newblock On the spatially homogeneous {B}oltzmann equation.
\newblock {\em Ann. Inst. H. Poincar\'{e} C Anal. Non Lin\'{e}aire},
  16(4):467--501, 1999.

\bibitem{Kusto-book}
E.~Nagnibeda and E.~Kustova.
\newblock {\em Non-equilibrium reacting gas flows}.
\newblock Heat and Mass Transfer. Springer-Verlag, Berlin, 2009.

\bibitem{MPC-MT-Kac}
M.~Pavi\'{c}-\v{C}oli\'{c} and M.~Taskovi\'{c}.
\newblock Propagation of stretched exponential moments for the {K}ac equation
  and {B}oltzmann equation with {M}axwell molecules.
\newblock {\em Kinet. Relat. Models}, 11(3):597--613, 2018.

\bibitem{IG-Alonso-Task-Pavl}
M.~Taskovi\'{c}, R.~J. Alonso, I.~M. Gamba, and N.~Pavlovi\'{c}.
\newblock On {M}ittag-{L}effler moments for the {B}oltzmann equation for hard
  potentials without cutoff.
\newblock {\em SIAM J. Math. Anal.}, 50(1):834--869, 2018.

\bibitem{Wenn}
B.~Wennberg.
\newblock Entropy dissipation and moment production for the {B}oltzmann
  equation.
\newblock {\em J. Statist. Phys.}, 86(5-6):1053--1066, 1997.

\end{thebibliography}
\end{document}